\documentclass[12pt, a4paper]{article}
\usepackage{amsfonts}
\usepackage{amssymb}
\usepackage{amsmath}
\usepackage{graphicx}
\usepackage{tikz}
\usepackage{pgflibraryarrows}
\usepackage{pgflibrarysnakes}

\usetikzlibrary{patterns}
\newtheorem{theorem}{Theorem}

\newtheorem{rem}[theorem]{Remark}

\newtheorem{definition}[theorem]{Definition}

\newtheorem{lemma}[theorem]{Lemma}

\newtheorem{proposition}[theorem]{Proposition}
\newtheorem{remark}[theorem]{Remark}

\newenvironment{proof}[1][Proof]{\noindent\textbf{#1.} }{\ \rule{0.5em}{0.5em}}

\begin{document}

\date{}
\title{Finite time blow-up and condensation for the bosonic Nordheim equation.}
\author{}
\maketitle

\begin{center}
\bigskip M. Escobedo\footnotemark[1]$^{,}$\footnotemark[2], J. J. L. Vel\'{a}%
zquez\footnotemark[3] \bigskip
\end{center}

\footnotetext[1]{%
Departamento de Matem\'{a}ticas. Universidad del Pa\'{\i}s Vasco UPV/EHU.
Apartado 644. E-48080 Bilbao, Spain. E-mail: miguel.escobedo@ehu.es} 
\footnotetext[2]{%
Basque Center for Applied Mathematics (BCAM), Alameda de Mazarredo 14,
E--48009 Bilbao, Spain.} \footnotetext[3]{%
Institute of Applied Mathematics, University of Bonn, Endenicher Allee 60,
53115 Bonn, Germany. E-mail: velazquez@iam.uni-bonn.de}

\noindent\textbf{Abstract.} The homogeneous bosonic Nordheim equation is a
kinetic equation describing the dynamics of the distribution of particles in
the space of moments for a homogeneous, weakly interacting,  quantum gas of
bosons. We show the existence of classical solutions of the homogeneous
bosonic Nordheim equation that blow up in finite time.
We also prove finite time condensation for a class of weak solutions of the kinetic
 equation. 


\section{Introduction}

\setcounter{equation}{0} \setcounter{theorem}{0}

The dynamics of the distribution of particles in the space of moments $%
F\left( t,p\right) $ for a  homogeneous, weakly interacting, quantum gas of bosons can be
described by the Nordheim equation:%
\begin{align}
\partial_{t}F_{1} & =\int_{\mathbb{R}^{3}}\int_{\mathbb{R}^{3}}\int_{\mathbb{%
R}^{3}}q\left( F\right) \mathcal{M}d^{3}p_{2}d^{3}p_{3}d^{3}p_{4}\, ,\ \
p_{1}\in\mathbb{R}^{3}\, ,\ \ t>0,  \label{E2} \\
F_{1}\left( 0,p\right) & =F_{0}\left( p\right) \, ,\ \ p_{1}\in\mathbb{R}%
^{3},  \label{E2b} \\
q\left( F\right) & =q_{3}\left( F\right) +q_{2}\left( F\right) ,\ \ \ \ \
\epsilon=\frac{\left\vert p\right\vert ^{2}}{2},  \label{T4E1a} \\
\mathcal{M} & =\mathcal{M}\left( p_{1},p_{2};p_{3},p_{4}\right)
=\delta\left( p_{1}+p_{2}-p_{3}-p_{4}\right) \delta\left( \epsilon
_{1}+\epsilon_{2}-\epsilon_{3}-\epsilon_{4}\right) ,  \label{T4E1b} \\
q_{3}\left( F\right) & =F_{3}F_{4}\left( F_{1}+F_{2}\right)
-F_{1}F_{2}\left( F_{3}+F_{4}\right) ,  \label{Q1E1} \\
q_{2}\left( F\right) & =F_{3}F_{4}-F_{1}F_{2},   \label{Q1E2}
\end{align}
where we use the notation $F_{j}=F\left( t,p_{j}\right) ,\ j\in \mathbb{R}%
^{3}.$ The system of units has been chosen in order to have particles with
mass equal to one. This system of equations was formulated by Nordheim (cf. 
\cite{Nor}). One of the main reasons that explain why this system has been
extensively studied in the physical literature is because it has been
considered by several authors as a convenient way to approximate the
dynamics of formation of Bose-Einstein condensates.

The system (\ref{E2})-(\ref{Q1E2}) can be thought as a generalization of the
classical Boltzmann equation of gas dynamics. The main difference between
both types of equations is the presence of the terms $q_{3}\left( F\right) $
(cf. (\ref{Q1E1})). These terms are cubic in the distribution of particles $%
F $ in spite of the fact that only binary collisions are taken into account
in the derivation of (\ref{E2})-(\ref{Q1E2}). The reason for the onset of
the cubic terms $q_{3}\left( F\right) $ in (\ref{E2}) is that the counting
of the number of collisions yielding the evolution of $F$ must be made using
Bose statistics, instead of the classical statistics which is used to
compute the number of collisions among different types of particles, in the
derivation of Boltzmann equation (cf. \cite{KN}, \cite{Nor}). From the
physical point of view it is possible to acquire an intuitive picture about
the quantum statistical effects assuming that the phase space is divided in
a family of small macroscopic domains, but large enough to accomodate a
large number of quantum cells with a volume of the order of the cube of
Planck's constant. The distribution function $F$ can then be characterized
by the occupation numbers of each macroscopic cell. In the bosonic case the
the wave function describing the whole system must be symmetric with respect
to permutations of all the particles. The main consequence of this is that
computing the number of collisions between particles Bose-Einstein
statistics must be used instead of classical statistical distributions. If
we choose the units of length and momentum to have as $F$ the total number
of occupied states in a given cell with respect to the total number of
quantum states admissible in such a cell, it turns out the the rate of
change of the number of occupied states at a point $p_{1},$ due to
interactions with particles in $\left[ p_{2},p_{2}+d^{3}p_{2}\right] $ and
yielding particles with moments $\left[ p_{3},p_{3}+d^{3}p_{3}\right] ,\ %
\left[ p_{4},p_{4}+d^{3}p_{4}\right] $ respectively, is given by:%
\begin{equation*}
\mathcal{M}\left[ F_{3}F_{4}\left( 1+F_{1}+F_{2}\right) -F_{1}F_{2}\left(
1+F_{3}+F_{4}\right) \right] d^{3}p_{2}d^{3}p_{3}d^{3}p, 
\end{equation*}
with $\mathcal{M}$ as in (\ref{T4E1b}), and where we have also assumed that
the unit of time has been choosen in order to obtain relevant changes of $F$
in time scales of order one.

The presence of the terms $q_{3}\left( F\right) $ in (\ref{E2}) produces
important differences between the behaviour of the solutions of the Nordheim
equation and Boltzmann equation. One of the main differences between the
dynamical behaviour of the solutions of both equations is that, as we will
prove in this paper, there exist solutions of the Nordheim equation which
are initially bounded and decay sufficiently fast at infinity to have finite
energy but become unbounded in finite time. This behaviour is very different
from the behaviour of the solutions of Boltzmann equation. Indeed, T.
Carleman proved already, (cf. \cite{Ca1}, \cite{Ca2}), that the solutions of
the homogeneous Boltzmann equation, with the cross section associated to
hard spheres interactions, and initial data bounded by $\left(
1+\left\vert p\right\vert ^{2}\right) ^{-\gamma}$ with $\gamma>3$ are
globally defined and bounded in time (see \cite{V} for a more detailed
discussion on that subject). In this paper we will show the existence of a
large class of solutions of (\ref{E2})-(\ref{Q1E2}) which satisfy similar
boundedness conditions for large values of $\left\vert p\right\vert $ but
become unbounded for finite $p.$

Blow up in finite time for the solutions of (\ref{E2})-(\ref{Q1E2}) has been
conjectured in the physical literature on the basis of numerical simulations
and physical arguments (cf. \cite{JPR}, \cite{LLPR}, \cite{LY1}, \cite{LY2}, 
\cite{ST1}, \cite{ST2}, \cite{Sv}). These studies are restricted to
spherically symmetric distributions $F.$ According to the numerical
simulations in those papers, many solutions of (\ref{E2})-(\ref{Q1E2}) which
are initially bounded develop a singularity in finite time at $p=0$. The
behaviour of this solution at the time of the formation of the singularity
can be described by means of a power law $\left\vert p\right\vert
^{-\alpha}.$ The time $T$ where such blow up takes place is usually
considered to be the time at which a Bose-Einstein condensate is formed.
According to the numerical simulations in \cite{LLPR}, \cite{ST1}, \cite{ST2}
the number of particles at the value $p=0$ at the time $t=0$ is zero, due to
the integrability of the singularity there. However, the scenarios for the
evolution of the distribution function $F$ suggested in \cite{JPR}, \cite%
{LLPR}, \cite{ST1}, \cite{ST2} indicate that for times $t>T$ a macroscopic
fraction of particles appears at the point $p=0.$ In mathematical terms, the
distribution $F$ becomes a measure containing a Dirac mass
at $p=0$  for times $t>T$.  Kinetic equations describing the joint evolution of distribution
functions in the presence of a condensed part have also been also obtained
in several other articles of the physical literature using different types
of physical approximations (cf. \cite{BS}, \cite{E}, \cite{GLBDZ}, \cite{IG}%
, \cite{PBS}, \cite{S}, \cite{ZNG}).

The different behaviour for the solutions of the Boltzmann and Nordheim
equations is due to the presence of the cubic terms $q_{3}\left( F\right) $
in (\ref{E2}), i.e. the terms not included in the classical Boltzmann
equation are a consequence of using the Bose-Einstein statistics. These
terms are the dominant ones for large densities, and they are the ones
yielding blow-up for $F.$

The cross section that appears in (\ref{E2})-(\ref{Q1E2}) is constant for
all the collisions preserving momentum and energy. This is the cross section
commonly used in the physical literature and it is usually justified on the
basis of the so-called Born approximation. The underlying idea behind this
approximation in the simplest case is the following. It is assumed that at
the fundamental level the system of quantum particles can be described by a
hamiltonian of the form $H=H_{0}+H_{1}$ where $H_{0}$ is the hamiltonian
describing a noninteracting system with $N$ particles. The term $H_{0}$\ can
include also some confining potential or equivalently some boundary
conditions ensuring that the particles remain in a bounded region.\ The
eigenvectors of $H_{0}$ can be labelled by means of a set of variables, and
the macroscopic states will be defined assuming that a large number of these
states are contained in each of them. On the other hand, the term $H_{1}$ is
typically a pair interaction potential which consists in the sum over all
possible pairs of particles of energy interactions induced by a potential $V.
$ The kinetic description given by Nordheim equation is expected to be a
good approximation of this quantum system if the number of particles $N$
tends to infinity and the interaction potentials and typical particle
energies are rescaled in a suitable way. However, no rigorous derivation of
Nordheim equation taking as starting point a Hamiltonian system is currently
available, although there exist some partial results in this direction (cf. 
\cite{Pu}).

Born's approximation allows to obtain the transition probabilities between
two given states of the system. The applicability of this approximation
requires to have an integrable interaction potential between particles $V$
whose range of interaction is much shorter than the characteristic De
Broglie length associated to the particles of the system. However, the
constancy of the differential cross section in the center of mass system
also takes place in some cases in which the Born approximation is not
strictly valid, for instance for hard spheres with a radius much smaller
than the characteristic De Broglie length of the system (cf. for instance 
\cite{HY}).

It is clarifying to rewrite (\ref{E2}) in the center of mass reference
system. This allows to obtain a precise geometrical interpretation of the
meaning of Born's approximation. Notice that we assume that the mass of
particles is one (see the formula for $\epsilon$ in (\ref{T4E1a})). In order
to perform some of the integrations in (\ref{E2}) we use the change of
variables $\left( p_{2},p_{3},p_{4}\right) \rightarrow\left(
p_{2},P,Q\right) $ where:%
\begin{equation*}
P=\frac{p_{3}+p_{4}}{2}\ \ ,\ \ Q=\frac{p_{3}-p_{4}}{2}=\frac{\left(
p_{3}-P\right) -\left( p_{4}-P\right) }{2}. 
\end{equation*}

Notice that since the mass of the particles has been normalized to one, $P$
is the velocity of the center of mass of the system. On the other hand $Q$
is a vector along the direction connecting the vectors $p_{3}$ and $p_{4}.$
This vector is invariant under change of reference system. It will provide a
measure of the deviation of the vector connecting the directions in the
center of mass system.

Notice that:%
\begin{align*}
\delta\left( p_{1}+p_{2}-p_{3}-p_{4}\right) \cdot\delta\left( \epsilon
_{1}+\epsilon_{2}-\epsilon_{3}-\epsilon_{4}\right) =\delta\left(
p_{1}+p_{2}-2P\right) \delta\left( \frac{\left( p_{1}-p_{2}\right) ^{2}}{4}%
-Q^{2}\right) ,
\end{align*}%
\begin{equation*}
d^{3}p_{3}d^{3}p_{4}=2^{3}d^{3}Pd^{3}Q. 
\end{equation*}

We also write $Q$ in spherical coordinates and write $p_{3},p_{4}$ in terms
of $P,\ Q:$%
\begin{equation*}
Q=\left\vert Q\right\vert \omega\text{ \ \ with \ }\omega\in S^{2}\ \ ,\ \
p_{3}=P+Q\ \ ,\ \ p_{4}=P-Q. 
\end{equation*}

Then, after some computations we can rewrite (\ref{E2}) as:%
\begin{equation*}
\partial_{t}F_{1}=8\int_{\mathbb{R}^{3}}\int_{S^{2}}\left. q\left( F\right)
\right\vert _{P=\frac{p_{1}+p_{2}}{2},\left\vert Q\right\vert =\frac {%
\left\vert p_{1}-p_{2}\right\vert }{2}\ }\left\vert Q\right\vert
d^{3}p_{2}d^{2}\omega. 
\end{equation*}

It becomes apparent from this formula that for any incoming direction $p_{2}$
colliding with $p_{1}$ the conservation of energy in the center of mass
system implies $\left\vert Q\right\vert =\frac{\left\vert
p_{1}-p_{2}\right\vert }{2}$ and the direction of the vector $Q$ is chosen
as a homogeneous probability measure in the unit sphere $S^{2}.$ This choice
of the direction of the outgoing particles by means of a uniformly
distributed measure in the space of all directions is the main distinguished
feature of Born's approximation.

It is relevant to mention that from the physical point of view, the
stationary solutions of (\ref{E2}), (\ref{T4E1a})-(\ref{Q1E2}) might be
expected to be the Bose-Einstein distributions:%
\begin{equation}
F_{BE}\left( p\right) =m_{0}\delta\left( p-p_{0}\right) +\frac{1}{\exp\left( 
\frac{\beta\left\vert p-p_{0}\right\vert ^{2}}{2}+\alpha\right) -1} 
\label{St1}
\end{equation}
where $m_{0}\geq0,$ $\beta\in\left( 0,\infty\right] ,$ $0\leq\alpha<\infty$
and $\alpha\cdot m_{0}=0,\ p_{0}\in\mathbb{R}^{3}$. Notice, however that the
precise sense in which the measures $F_{BE}$ are stationary solutions of (%
\ref{E2}), (\ref{T4E1a})-(\ref{Q1E2}) is not clear, because the right-hand
side of (\ref{E2}) is not well defined for measures containing Dirac masses.
There are several possible ways of justifying that the measures $F_{BE}$ in (%
\ref{St1}) are, at least in some sense, steady state distributions of (\ref%
{E2}), (\ref{T4E1a})-(\ref{Q1E2}). First notice that $q\left( F_{BE}\right)
=0,$ with $q\left( \cdot\right) $ as in (\ref{Q1E1}), (\ref{Q1E2}).
Nevertheless this argument must be taken with some care, because the in
general is not possible to define the right-hand side of (\ref{E2}) for more
general measures containing Dirac masses. Another way to see that the
distributions $F_{BE}$ must play the role of the stationary solutions for (%
\ref{E2}), (\ref{T4E1a})-(\ref{Q1E2}) is to notice that these distributions
are maximizers of the entropy associated to the system for a given value of
the number of particles and energy. More precisely, the entropy for unit of
volume of a homogeneous system of bosons with distribution function $F$ in
the space $p$ is given by:%
\begin{equation}
S=\int\left[ \left( 1+F\right) \log\left( 1+F\right) -F\log\left( F\right) %
\right] d^{3}p.   \label{Ent}
\end{equation}

This functional is increasing along the solutions of (\ref{E2}), (\ref{T4E1a}%
)-(\ref{Q1E2}). On the other hand it can be seen that the maximizer of the
functional (\ref{Ent}) with the constraints $\int F\left( p\right)
d^{3}p=M,\ \int F\left( p\right) pd^{3}p=p_{0},\ \int F\left( p\right) \frac{%
\left\vert p\right\vert ^{2}}{2}d^{3}p=E$ is given by the distribution $%
F_{BE}$ (cf. \cite{Huang}).

We will restrict our analysis in the following to isotropic distributions.
Therefore:%
\begin{equation*}
F\left( t,p\right) =F\left( t,\mathcal{R}p\right) ,\ \ \mathcal{R}\in
SO\left( 3\right) ,\ \ p\in\mathbb{R}^{3}\ \ ,\ \ t\geq0. 
\end{equation*}

It then follows that there exists a function $f=f\left( \epsilon,t\right) $
where $\epsilon$ is as in (\ref{T4E1a}) such that:%
\begin{equation*}
f\left( t,\epsilon\right) =F\left( t,p\right) . 
\end{equation*}

Given a spherically symmetric solution of (\ref{E2})-(\ref{Q1E2}) it is
possible to write an equation for $f\left( t,\epsilon\right) $. (see \cite%
{ST2} for details). We first use the formula 
\begin{equation*}
\delta(p_{1}+p_{2}-p_{3}-p_{4})=\frac{1}{\left( 2\pi\right) ^{3}}\int_{%
\mathbb{R}^{3}}e^{ik\cdot\left( p_{1}+p_{2}-p_{3}-p_{4}\right) }d^{3}k, 
\end{equation*}
which is valid in the sense of distributions. We next write $p_{2},\ p_{3},\
p_{4}$ as well as $k$ using spherical coordinates. Integrating the angular
variables, (\ref{E2}) becomes:%
\begin{equation}
\partial_{t}f_{1}=32\pi\int_{0}^{\infty}\int_{0}^{\infty}\int_{0}^{\infty
}Dq\left( f\right) \delta\left( \epsilon_{1}+\epsilon_{2}-\epsilon
_{3}-\epsilon_{4}\right) d\epsilon_{2}d\epsilon_{3}d\epsilon_{4},\ 
\label{F3E1}
\end{equation}
where:%
\begin{equation*}
D=\frac{1}{\left\vert p_{1}\right\vert }\int_{0}^{\infty}\left[ \prod
_{j=1}^{4}\sin\left( \left\vert p_{j}\right\vert \lambda\right) \right] 
\frac{d\lambda}{\lambda^{2}}. 
\end{equation*}

This integral can be explicitly computed by means of elementary arguments.
Using the fact that $\left\vert p_{1}\right\vert ^{2}+\left\vert
p_{2}\right\vert ^{2}=\left\vert p_{3}\right\vert ^{2}+\left\vert
p_{4}\right\vert ^{2}$ which is due to the presence of the Dirac mass in (%
\ref{F3E1}) it follows that (cf. \cite{ST2}):%
\begin{equation*}
D=\frac{\pi}{4}\min\left\{ \sqrt{\epsilon_{1}},\sqrt{\epsilon_{2}},\sqrt{%
\epsilon_{3}},\sqrt{\epsilon_{4}}\right\} 
\end{equation*}
whence, integrating the Dirac mass with respect to the variable $%
\epsilon_{2} $ we obtain: 
\begin{equation}
\partial_{t}f_{1}=\frac{8\pi^{2}}{\sqrt{2}}\int_{0}^{\infty}\int_{0}^{\infty
}q\left( f\right) Wd\epsilon_{3}d\epsilon_{4},\   \label{F3E2}
\end{equation}
where:%
\begin{equation}
W=\frac{\min\left\{ \sqrt{\epsilon_{1}},\sqrt{\epsilon_{2}},\sqrt {%
\epsilon_{3}},\sqrt{\epsilon_{4}}\right\} }{\sqrt{\epsilon_{1}}} ,\ \
\epsilon_{2}=\epsilon_{3}+\epsilon_{4}-\epsilon_{1},   \label{F3E3}
\end{equation}
and $q\left( \cdot\right) $ is as in (\ref{T4E1a}) with $\epsilon
_{2}=\epsilon_{3}+\epsilon_{4}-\epsilon_{1}$.

It is worth to notice that $f\left( t,\epsilon\right) $ is not a particle
density in the energy space. The density of particles in the energy space is
actually given by:%
\begin{equation}
g\left( t,\epsilon\right) =4\pi\sqrt{2\epsilon}f\left( t,\epsilon\right) . 
\label{F3E3a}
\end{equation}

We can rewrite (\ref{F3E2}), (\ref{F3E3}) using the density $g:$%
\begin{align}
\partial_{t}g_{1} & =32\pi^{3}\int_{0}^{\infty}\int_{0}^{\infty}q\left( 
\frac{g}{4\pi\sqrt{2\epsilon}}\right) \Phi d\epsilon_{3}d\epsilon _{4},\ 
\label{F3E4} \\
\Phi & =\min\left\{ \sqrt{\epsilon_{1}},\sqrt{\epsilon_{2}},\sqrt {%
\epsilon_{3}},\sqrt{\epsilon_{4}}\right\} \ \ ,\ \ \epsilon_{2}=\epsilon
_{3}+\epsilon_{4}-\epsilon_{1}.   \label{F3E5}
\end{align}

In spite of the fact that $f$ is not the density of particles in the space
of energy, the formulation of the problem (\ref{F3E2}), (\ref{F3E3}) is
suitable in order to prove local well-posedness results for classical
solutions. On the other hand, in order to understand phenomena like
Bose-Einstein condensation, or in general any phenomenon characterized by
the presence of a positive amount of particles near the origin, it is often
more convenient to use the density of particles $g.$ Indeed, in such a case
the formation of condensates is characterized by the onset of a Dirac mass
at $\epsilon=0.$ In this paper we will switch between both functions $f,\ g$
depending on the question under consideration.

Equation (\ref{F3E4}) shows that Dirac mass distributions of $g$ placed at
the origin are too singular to yield a well defined right-hand side. Dirac
masses at $\epsilon=\epsilon_{0}>0$ are not a serious mathematical
difficulty. However, it must be noticed that such measures, in the original
set of variables $p$ correspond to measures supported on the sphere $%
\left\vert p\right\vert =\sqrt{2\epsilon_{0}}>0.$ Many of the mathematical
difficulties arising in the study of the equations (\ref{F3E4}), (\ref{F3E5}%
) have their root in the singular character of Dirac masses concentrated at
the origin for these equations.

Equation (\ref{F3E2}), (\ref{F3E3}) has been extensively studied in the
physical literature, usually by means of numerical simulations (cf. \cite%
{JPR}, \cite{LLPR}, \cite{LY1}, \cite{LY2}, \cite{ST1}, \cite{ST2}, \cite{Sv}%
). However, there are not many rigorous mathematical results for this
equation. A theory of weak solutions has been developed by X. Lu in \cite
{Lu1, Lu2, Lu3}. It has been proved in these papers that it
is possible to define a concept of measure valued solutions for (\ref{F3E2}%
), (\ref{F3E3}) and that at least one of such solutions exists globally in
time for a large class of initial data. This solution converges in the weak
topology of measures to the stationary state $F_{BE}$ having the same mass
as the initial data $F_{0}$ as $t\rightarrow\infty$ for long times (cf \cite%
{Lu2}).

Another class of solutions of (\ref{F3E2}), (\ref{F3E3}) has been
constructed in \cite{EMV1}, \cite{EMV2}. The solutions constructed in these
papers are locally defined in time. They satisfy the equation in a classical
sense for $\epsilon>0$ and have the singular behaviour $f\left(
t,\epsilon\right) \sim a\left( t\right) \epsilon^{-\frac{7}{6}}$ as $%
\epsilon\rightarrow0$ for some function $a\left( t\right) .$ The most
relevant feature of these solutions is that they yield a nonzero flux of
particles towards the origin. Therefore, the number of particles in the
region $\left\{ \epsilon>0\right\} $ is a decreasing function. It is not
clear at the moment if the solutions in \cite{EMV1}, \cite{EMV2} can be
extended to the set $\left\{ \epsilon \geq0\right\} $ in order to obtain
some kind of mass preserving weak solution of (\ref{F3E2}), (\ref{F3E3}),
but it has been proved in \cite{Lu3} that they are not a weak solution of (%
\ref{F3E2}), (\ref{F3E3}) in the sense of the definition given in \cite{Lu1}.

A review of the currently available mathematical results for the solutions
of (\ref{F3E2}), (\ref{F3E3}) can be found in \cite{Spohn}. This paper
discuss also the known mathematical properties and the expected behaviour
for the function $f\left( t,\epsilon\right) $ in the presence of condensate.

In this paper we first prove the existence of a large class of initial data for (%
\ref{F3E2}), (\ref{F3E3}) for which there exist classical solutions defined
for a finite time interval $0<t<T$ but becoming unbounded as $t\rightarrow T.
$ This blow-up result supports the scenario for the dynamical formation of
Bose-Einstein condensates suggested in \cite{JPR}, \cite{LLPR}, \cite{ST1}, 
\cite{ST2}, \cite{Sv}. The  scenario  suggested in those papers
is the following: there exists  a classical, solution $f$  of the Nordheim equation in a time interval,  which   blows up  in  finite time $T _{ max }$. This solution can be extended as a weak solution containing a Dirac mass at the origin for all $t>T _{ max }$.

The results of this paper support this scenario in the following sense. Theorem \ref{main} shows the existence of local in time bounded solutions (mild solutions)  that blow up in finite time $T _{ max }$ for a non empty set of initial data. In Theorem  \ref{Theoremtenfive} we prove the existence of solutions, which are global weak solutions, that are bounded mild solutions for some finite time, and that, after some time, contain a Dirac mass at the origin. Notice that we do not prove that the onset of the Dirac mass formation takes place precisely at $T _{ max }$, but only  at some time $T _{ cond }\ge T _{ max }$.

We now describe some of the main ideas used in the proof of this blow-up
result. As a first step we need to prove local existence of classical
solutions for the Cauchy problem associated to (\ref{F3E2}), (\ref{F3E3})
with initial data decreasing like a power law for large $\epsilon.$ This
problem can be solved using the methods introduced by T. Carleman in his
seminal work about the well posedness of the spatially homogeneous Boltzmann
equation (cf. \cite{Ca1}, \cite{Ca2}). Actually the main difficulty proving
local existence for (\ref{F3E2}), (\ref{F3E3}) has more to do with the
quadratic terms in (\ref{Q1E2}) than with the cubic terms (\ref{Q1E1}). The
reason for this is that the main difficulty that must be solved both for
Boltzmann and Nordheim equations in order to prove local well-posedness is
to find a class of functions whose behaviour for large $\epsilon$ is
preserved in some suitable iterative scheme. For large values of $\epsilon$
the dominant terms in the equation are the quadratic ones while the cubic
terms can be treated as some kind of perturbation. It was found by T.
Carleman that a suitable class of functions that allow to prove
well-posedness by means of an iterative argument are the functions bounded
as $C\left( 1+\epsilon\right) ^{-\gamma}$ with $\gamma>3.$ Actually
this class of functions allows to prove well posedness even for
nonspherically symmetric distributions. In the Boltzmann case, due to the
conservation of the energy and the number of particles it is possible to
prove global existence of solutions (cf. \cite{Ca1}, \cite{Ca2}).

The derivation of optimal decay estimates for the solutions of the
homogeneous Boltzmann equation which allow to prove conservation of the
energy of the system has been extensively studied (cf. \cite{MW}, \cite{P},
as well as the review \cite{V} and the references therein). However, since
this is an issue more related to the classical Boltzmann equation than to
the specific effects induced by the cubic terms in (\ref{F3E2}), (\ref{F3E3}%
) we have preferred to use the methods of T. Carleman to prove local
well-posedness. The reason being that, although Carleman's method requires
to impose decay estimates for the solutions more restrictive than some of
the more recent approaches, it uses simpler arguments. In spite of this this
approach will be enough to obtain a large class of initial data yielding
blow-up for the solutions of (\ref{F3E2}), (\ref{F3E3}) in finite time.

In order to prove the blow-up for the solutions some additional methods are
needed. We will use first a crucial monotonicity property that has been
obtained in \cite{Lu3}. This property allows to obtain an estimate for the
net number of collisions taking place between particles with small energy.
Roughly speaking this number measures the rate of change of the number of
small particles associated to $g.$ It turns out that, for a given amount of
mass in the region $\left\{ \epsilon\in\left[ 0,R\right] \right\} ,$ the net
number of collisions between particles with small energy yields big changes
in the mass of the system, except if the distribution $g$ is very close, in
the weak topology, to a Dirac mass supported at some point $%
\epsilon=\epsilon _{0}>0$. Given that the total mass of the system is
bounded such large changes of the mass are not admissible. Therefore, if the
initial number of particles with small energy is sufficiently large and the
solutions do not become unbounded in finite time, only one alternative is
left, namely, the distribution $g$ close, in the weak topology, to a Dirac
mass at a particular value of the energy $\epsilon_{0}>0.$

The rigorous proof of this concentration estimate will made use of a key
Measure Theory result that describes in a rather precise way the degree of
concentration of arbitrary measures defined in intervals $\left[ 0,R\right] $
with $R$ small.

Finally we will prove that the concentration of $g$ at a positive value of
the energy cannot take place for a set of times too large if $f_{0}\left(
\epsilon\right) \geq\nu>0$ for small $\epsilon$. This will be seen using two
arguments. First, we will prove, using again the monotonicity property
mentioned above, that the condition $f_{0}\left( \epsilon\right) \geq\nu>0$
implies that the number of particles with energy in the interval $\left[ 0,R%
\right] $ with $R$ small, can be bounded from below for times of order one.
Using this, we can see that there would be a fast transfer of particles with
energy $\epsilon=\epsilon_{0}$ towards lowest energy values. This would
contradict also the conservation of the total number of particles and then,
the only alternative left is blow-up in finite time for $f.$

This will be seen deriving an estimate for the transfer of particles from
the peak at $\epsilon=\epsilon_{0}$ to the region of smaller particles. This
estimate will show that for solutions of (\ref{F3E2}), (\ref{F3E3}) with $%
f_{0}\left( \epsilon\right) \geq\nu>0$ the transfer of particles from the
peak at $\epsilon=\epsilon_{0}$ towards smaller particles is very large.

The previous arguments will be made precise deriving detailed estimates of
the measure of the set of times in which the distribution $g$ behaves in
each specific form.

On the other hand, the reason for the onset of the alternative which states
that the net rate of collisions can take place only if $g$ concentrates near
a Dirac mass is because for small values of $\epsilon$ the dominant terms of
the equation are the cubic ones in (\ref{T4E1a}). The equation verified by $g
$ if only the terms in $q_{3}\left( \cdot\right) $ are kept is (cf. (\ref%
{F3E4})):%
\begin{equation}
\partial_{t}g_{1}=32\pi^{3}\int_{0}^{\infty}\int_{0}^{\infty}q_{3}\left( 
\frac{g}{4\pi\sqrt{2\epsilon}}\right) \Phi d\epsilon_{3}d\epsilon _{4}\ \ 
\label{St3}
\end{equation}
and it can be readily seen that every distribution of the form $%
g=M_{0}\delta_{\epsilon_{0}}$ with $M_{0}\geq0,\ \epsilon_{0}>0$ is a
stationary solution of (\ref{St3}).

Similar arguments are applied to weak solutions and yield also
condensation in finite time.

The plan of the paper is the following. In Section 2 we define the concept
of solution which will be used and state the blow-up result which we prove
in this paper. In Section 3 we prove a local well posedness result for the
equation (\ref{F3E2}), (\ref{F3E3}) for suitable initial data. Section 4
describes a basic monotonicity property of the solutions of (\ref{F3E2}), (%
\ref{F3E3}) which will be used in Section 5 to derive a first estimate for
the number of collisions which change the energy of the particles. Section 6
contains a crucial measure theory result which states in a precise
quantitative way that an arbitrary measure in $\left[ 0,1\right] ,$ either
is concentrated in a small set, or it has its mass spread among two
measurable sets which are ``sufficiently separated". This measure theory
result is used in Section 7 to transform the estimate obtained in Section 5
to another estimate which states that the solutions of (\ref{F3E2}), (\ref%
{F3E3}), must have their mass concentrated in a narrow peak if they are
defined during sufficiently long times. Section 8 contains some estimates
which prove that the portion of particle distributions concentrated in
narrow peaks with small energy would be transferred very fast towards even
smaller energy values. This is used in Section 9 to conclude the Proof of
the blow-up result obtained in this paper. Section 10 contains several
results which prove that some suitable weak solutions of (\ref{F3E2}), (\ref%
{F3E3}) yield finite time condensation.

In all the paper $C$ will be a generic numerical constant which can change
from line to line. We will denote the Lebesgue measure of a measurable set $%
A\subset\mathbb{R}$ as $\left\vert A\right\vert .$

\section{Main results of the paper.}

\setcounter{equation}{0} \setcounter{theorem}{0}

\subsection{Definition of mild solution.}

In order to formulate the main blow-up result of this paper we need to
introduce some functional spaces and to precise the concept of solution of (%
\ref{F3E2}), (\ref{F3E3}).

Given $\gamma\in\mathbb{R}$ we will denote as $L^{\infty}\left( \mathbb{R}%
^{+};\left( 1+\epsilon\right) ^{\gamma}\right) $ the space of functions such
that:%
\begin{equation*}
\left\Vert f\right\Vert _{L^{\infty}\left( \mathbb{R}^{+};\left(
1+\epsilon\right) ^{\gamma}\right) }=\sup_{\epsilon\geq0}\left\{ \left(
1+\epsilon\right) ^{\gamma}f\left( \epsilon\right) \right\} <\infty. 
\end{equation*}

Notice that $L^{\infty}\left( \mathbb{R}^{+};\left( 1+\epsilon\right)
^{\gamma}\right) $ is a Banach space with the norm $\left\Vert \cdot
\right\Vert _{L^{\infty}\left( \mathbb{R}^{+};\left( 1+\epsilon\right)
^{\gamma}\right) }.$

We will use also the spaces $X _{ T_1, T_2 }= L_{loc}^{\infty}\left( \left[
T_{1},T_{2}\right) ;L^{\infty}\left( \mathbb{R}^{+};\left( 1+\epsilon\right)
^{\gamma}\right) \right) $ of the functions $f$ satisfying:%
\begin{equation*}
\sup_{t\in K}\left\Vert f\left( t,\cdot\right) \right\Vert _{L^{\infty
}\left( \mathbb{R}^{+};\left( 1+\epsilon\right) ^{\gamma}\right) }<\infty, 
\end{equation*}
for any compact $K\subset\left[ T_{1},T_{2}\right) .$ Notice that the
spaces $X _{ T_1, T_2 }$ are not Banach spaces. We will use the space $L^{\infty}\left( \left[
T_{1},T_{2}\right] ;L^{\infty}\left( \mathbb{R}^{+};\left( 1+\epsilon
\right) ^{\gamma}\right) \right) $ which is the Banach space of functions
such that: 
\begin{equation*}
\left\Vert f\right\Vert _{L^{\infty}\left( \left[ T_{1},T_{2}\right]
;L^{\infty}\left( \mathbb{R}^{+};\left( 1+\epsilon\right) ^{\gamma}\right)
\right) }=\sup_{t\in\left[ T_{1},T_{2}\right] }\left\Vert f\left(
t,\cdot\right) \right\Vert _{L^{\infty}\left( \mathbb{R}^{+};\left(
1+\epsilon\right) ^{\gamma}\right) }<\infty. 
\end{equation*}

\begin{definition}
\label{mild}Suppose that $\gamma>3$ and $0\le T_{1}<T_{2}<+\infty$. We will
say that a function $f\in L_{loc}^{\infty}\left( \left[ T_{1},T_{2}\right)
;L^{\infty}\left( \mathbb{R}^{+};\left( 1+\epsilon\right) ^{\gamma}\right)
\right) $ is a mild solution of (\ref{F3E2}), (\ref{F3E3}) on $(0, T)$ with initial data $f_0$ if it satisfies:%
\begin{equation}
f\left( t,\epsilon_{1}\right) =f_{0}\left( \epsilon_{1}\right) \Psi\left(
t,\epsilon_{1}\right) +\frac{8\pi^{2}}{\sqrt{2}}\int_{T_{1}}^{t}\frac {%
\Psi\left( t,\epsilon_{1}\right) }{\Psi\left( s,\epsilon_{1}\right) }%
\int_{0}^{\infty}\int_{0}^{\infty}f_{3}f_{4}\left( 1+f_{1}+f_{2}\right)
Wd\epsilon_{3}d\epsilon_{4}ds   \label{F3E6}
\end{equation}
a.e. $t\in\left[ T_{1},T_{2}\right) $, where: 
\begin{align}
a\left( t,\epsilon_{1}\right)& =\frac{8\pi^{2}}{\sqrt{2}}\int_{0}^{\infty
}\int_{0}^{\infty}f_{2}\left( 1+f_{3}+f_{4}\right)
Wd\epsilon_{3}d\epsilon_{4}\ \ ,\ \ \Psi\left( t,\epsilon_{1}\right)\notag\\
&
=\exp\left( -\int_{T_{1}}^{t}a\left( s,\epsilon_{1}\right) ds\right) . 
\label{F3E6a}
\end{align}
\end{definition}

\begin{remark}
Since $\gamma>3$, for any $t\in[T_{1}, T_{2})$ the integral term 
$$\int
_{0}^{\infty}\int_{0}^{\infty}f_{3}f_{4}\left( 1+f_{1}+f_{2}\right)
Wd\epsilon_{3}d\epsilon_{4}$$ 
can be estimated by a constant $C\left\Vert
f\right\Vert _{L^{\infty}\left( \mathbb{R}^{+};\left( 1+\epsilon\right)
^{\gamma}\right) }\left( 1+\left\Vert f\right\Vert _{L^{\infty}\left( 
\mathbb{R}^{+};\left( 1+\epsilon\right) ^{\gamma}\right) }\right) $ where $C 
$ is a numerical constant. The term $a\left( t,\epsilon_{1}\right) $ can be
estimated by 
$$C\left( 1+\left\Vert f\right\Vert _{L^{\infty}\left( \mathbb{R}%
^{+};\left( 1+\epsilon\right) ^{\gamma}\right) }\right) \int
_{0}^{\infty}\int_{0}^{\infty}f_{2}Wd\epsilon_{3}d\epsilon_{4}.$$ 
By the
definition of $W$, and using $\epsilon_{2}$ as one of the integration
variables, we estimate $\int_{0}^{\infty}\int_{0}^{\infty}f_{2}Wd\epsilon
_{3}d\epsilon_{4}$ as $\int_{0}^{\infty}f_{2}\left( \sqrt{\epsilon_{2}}+%
\sqrt{\epsilon_{1}}\right) Wd\epsilon_{2}\leq C\left\Vert f\right\Vert
_{L^{\infty}\left( \mathbb{R}^{+};\left( 1+\epsilon\right) ^{\gamma }\right)
}.$ Therefore,  if $f\in L_{loc}^{\infty}\left( \left[ T_{1},T_{2}\right)
;L^{\infty}\left( \mathbb{R}^{+};\left( 1+\epsilon\right) ^{\gamma}\right)
\right) $ all the terms in (\ref{F3E6}) are well defined\ for $%
T_{1}\leq t<T_{2}$.
\end{remark}

\subsection{Blow-up Result.}

The main blow-up result of this paper is the following:

\begin{theorem}
\label{main} There exist $\
\theta_{\ast}>0$ such that,  for all  $M>0,\ E>0,$ $\nu>0,\ \gamma>3$,  there exists $ 
\rho_0=\rho_0\left( M,E,\nu\right) >0$, $ K^{\ast}=K^{\ast}\left( M,E,\nu\right)
>0,\ T_{0}=T_{0}\left( M,E\right) $ satisfying the following property.  For any $f_{0}\in
L^{\infty}\left( \mathbb{R}^{+};\left( 1+\epsilon\right) ^{\gamma}\right) $
such that$\ $%
\begin{align}
4\pi\sqrt{2}\int_{\mathbb{R}^{+}}f_{0}\left( \epsilon\right) \sqrt{\epsilon }%
d\epsilon & =M\ ,\ \ 4\pi\sqrt{2}\int_{\mathbb{R}^{+}}f_{0}\left(
\epsilon\right) \sqrt{\epsilon^{3}}d\epsilon=E,  \label{C1} 
\end{align}
and
 \begin{equation}
 \label{Condicion}
 \sup _{ 0\le \rho \le \rho _0}\left[
 \min\left\{
\inf _{ 0\le R\le \rho  }\frac {1} {\nu R^{3/2}}\int _0^R f_0(\epsilon)\sqrt \epsilon d\epsilon,  \frac {1} {K^*\rho ^{\theta_*}}\int _0^\rho f_0(\epsilon)\sqrt \epsilon d\epsilon
 \right\}
 \right]\ge 1,
 \end{equation}
 there exists a unique mild solution  of   (\ref{F3E2}), (\ref{F3E3}) in the sense of Definition \ref{mild}
 $f\in L_{loc}^{\infty}\left( \left[
0,T_{\max}\right) ;L^{\infty}\left( \mathbb{R}^{+};\left( 1+\epsilon \right)
^{\gamma}\right) \right)
$, with initial data $f_0$,  defined for a maximal existence time $%
T_{\max}<T_{0}$ and that satisfies:
\begin{equation*}
\limsup_{t\rightarrow T_{\max}^{-}}\left\Vert f\left( \cdot,t\right)
\right\Vert _{L^{\infty}\left( \mathbb{R}^{+}\right) }=\infty. 
\end{equation*}
\end{theorem}

The above Theorem  means that initial data  $f_{0}\in
L^{\infty}\left( \mathbb{R}^{+};\left( 1+\epsilon\right) ^{\gamma}\right) $, with a sufficiently large density around $\epsilon =0$, blows up in finite time.
More precisely, the condition (\ref{Condicion}) means that  there exists $\rho \in(0, \rho _0)$ satisfying:
\begin{align}
\int_{0}^{R}f_{0}\left( \epsilon\right) \sqrt{\epsilon}d\epsilon & \geq\nu
R^{\frac{3}{2}}\ \ \text{for }0<R\leq\rho\ \ ,\ \ \int_{0}^{\rho}f_{0}\left(
\epsilon\right) \sqrt{\epsilon}d\epsilon\geq K^{\ast}\left( \rho\right)
^{\theta_{\ast}},   \label{C2}
\end{align} 
The second condition in (\ref{C2}) holds if the distribution $f_0$ has a mass sufficiently large in a ball with radius $\rho$ for some $\rho$ sufficiently small. The first condition is satisfied if $f_0(\epsilon)\ge 3\nu/2$ for all $\varepsilon $ sufficiently small. Since $\theta_*$ might  be small, the first condition in (\ref{C2}) does not implies the second.

Our results  do not provide an explicit functional relation for the functions \hfill \break $\rho _0(M, E, \nu)$, $K^*(M, E, \nu)$ and $T_0(M, E)$ in terms of their arguments. Therefore, for a given initial data $f_0 \in
L^{\infty}\left( \mathbb{R}^{+};\left( 1+\epsilon\right) ^{\gamma}\right) $ it is not easy to check if condition (\ref{C2}) is satisfied. However, it is simple to verify that the class of functions $f_0 \in
L^{\infty}\left( \mathbb{R}^{+};\left( 1+\epsilon\right) ^{\gamma}\right) $ satisfying such condition is not empty. Indeed, we have:

\begin{proposition}
\label{example}
Given $M>0,  E>0, \nu>0, \gamma>0$  there exists a family of function $f_{0}\in
L^{\infty}\left( \mathbb{R}^{+};\left( 1+\epsilon\right) ^{\gamma}\right) $ satisfying (\ref{C1})  and for which (\ref{Condicion}) holds with $\theta_*, K^*$ and $\rho _0$ as in Theorem \ref{main}.
\end{proposition}

\begin{proof} Consider $\zeta\in C_0([0, \infty))$ such that:
\begin{equation}
\label{zetafunction}
0\le \zeta\le 2,\quad \int _0^\infty \zeta(s) \sqrt s ds=1, \quad \zeta(s) \ge \frac {1} {4} ,\,\forall s\in [0, 1], \quad \hbox{supp} (\zeta)=[0, 2].
\end{equation}
We also consider functions  functions $\overline f_1\ge 0$ and $\overline f_2\ge 0$,  $\overline f_1 \in
L^{\infty}\left( \mathbb{R}^{+};\left( 1+\epsilon\right) ^{\gamma}\right) $,  $\overline f_2 \in
L^{\infty}\left( \mathbb{R}^{+};\left( 1+\epsilon\right) ^{\gamma}\right) $
such that
\begin{eqnarray}
&&\int _0^\infty \overline f_1(\epsilon)\sqrt \epsilon d \epsilon=\int _0^\infty \overline f_2(\epsilon)\sqrt \epsilon d \epsilon=\frac{M}{2}, \label{MandE1}\\
&&\int _0^\infty \overline f_1(\epsilon)\epsilon ^{3/2}d \epsilon=\frac {E} {4},\quad \int _0^\infty \overline f_2(\epsilon)\epsilon ^{3/2}d \epsilon=\frac {3E} {4}.\label{MandE2}
\end{eqnarray}
For example, fix a function $\varphi \ge 0$, $\varphi \in C_0(0, \infty)$ such that
$$
\int _0^\infty \varphi (\epsilon) \sqrt{\epsilon} d\epsilon=\frac {M} {2}
$$
and let be 
$$
C_1=\int _0^\infty \varphi (\epsilon) \epsilon^{3/2} d\epsilon.
$$
Then, we may take:
\begin{eqnarray}
\overline{f}_k(\epsilon)=\mu _k^{3/2}\varphi (\mu _k \epsilon), \,\,\,k=1, 2,
\end{eqnarray}
where $\mu _1=\frac {4C_1} {E}$ and $\mu _2=\frac {4C_1} {3E}$. \\

We define now
\begin{equation}
\label{Datos}
f_0(\epsilon)=\frac {1} {\rho ^{\beta +\frac {1} {2}}}\zeta \left(\frac {\epsilon} {\rho } \right)+\kappa_1\overline{f}_1(\epsilon)+ \kappa_2\, \overline{f}_2(\epsilon)
\end{equation}
where $0<\beta <1$, $\kappa_1 >0$, $\kappa_2>0$ to be precised.

Given $\rho >0$, we choose $\kappa_1$ and $\kappa_2$ satisfying:
\begin{eqnarray}
\kappa_1+\kappa_2&=&1-\frac {2\rho ^{1-\beta }} {M} \label{systemone}\\
\kappa_1+3\kappa_2&=&4-\frac {4C_0 \rho ^{2-\beta }} {E}  \label{systemtwo}
\end{eqnarray}
where
$$C_0=\int _0^\infty s^{3/2} \zeta (s)ds\in (0, 2).$$
With such a choice, the function $f_0$ satisfies (\ref{C1}). The solutions of system (\ref{systemone}), (\ref{systemtwo}) are such that:
\begin{eqnarray}
\kappa_1=\frac{1}{3}+\lambda _1(\rho , M, E),\,\,\,\,\lim _{ \rho \to 0 }\lambda _1(\rho , M, E)=0,\\
\kappa_2=\frac {2} {3}+\lambda _2(\rho , M, E),\,\,\,\,\,\lim _{ \rho \to 0 }\lambda _2(\rho , M, E)=0
\end{eqnarray} 
Therefore, there exists $\overline\rho (M, E)>0$, such that  if $\rho <\overline \rho (M, E)$,  we have $\kappa_1>0$, $\kappa_2>0$ and then $f_0\ge 0$.

We now choose $\beta $ and $\rho $ in order  for $f_0$ to satisfy (\ref{C2}). To this end we first observe that:
\begin{equation*}
\int _0^Rf_0(\epsilon)\sqrt\epsilon d\epsilon \ge \frac {1} {\rho ^{\beta+1/2}}\int _0^R\zeta\left(\frac {\epsilon} {\rho }\right)\sqrt\epsilon d\epsilon
\end{equation*}
Using (\ref{zetafunction}) we obtain:
\begin{equation}
\label{CondicionUno}
\forall R\le \rho, \,\,\,\,\,\frac {1} {\rho ^{\beta+1/2}}\int _0^R\zeta\left(\frac {\epsilon} {\rho }\right)\sqrt\epsilon d\epsilon\ge 
\frac {1} {4\rho ^{\beta+1/2}}\int _0^R\sqrt\epsilon d\epsilon=\frac {R^{3/2}} {6\rho ^{\beta+1/2}}
\end{equation}
Then, the first condition in (\ref{C2}) holds if:
\begin{equation}
\label{condicionrho}
\rho <\left(\frac {1} {6\nu}\right)^{\frac {1} {\beta+1/2}}.
\end{equation}
Using (\ref{CondicionUno}) with $R=\rho$ as well as (\ref{zetafunction}), we obtain
\begin{equation*}
\frac {1} {\rho ^{\beta+1/2}}\int _0^R\zeta\left(\frac {\epsilon} {\rho }\right)\sqrt\epsilon d\epsilon\ge 
\frac {1} {6}\rho ^{1-\beta}
\end{equation*}
We now chose $\beta$ such that
\begin{equation}
\label{beta}
1-\theta_*<\beta <1.
\end{equation}
Then, if 
\begin{equation}
\label{rhoone}
\rho <\left(\frac {1} {6K^*}\right)^{\frac {1} {\beta-1+\theta_*}}
\end{equation}
the second condition in (\ref{C2}) is satisfied. Choosing then
\begin{equation}
\label{rhotwo}
0<\rho <\min\left(\rho _0(M, E, \nu), \overline \rho (M, E),\left(\frac {1} {6K^*(M, E)}\right)^{\frac {1} {\beta-1+\theta_*}},  \left(\frac {1} {6\nu}\right)^{\frac {1} {\beta+1/2}}\right)
\end{equation}
and $\beta$ as in (\ref{beta}), all the conditions in (\ref{C1}) (\ref{C2}) are then satisfied.

An example of functions $\zeta$, $\overline f_1, \overline f_2$, and constants $\mu _1, \mu _2$ satisfying all the requirements are the following: $\zeta=\frac {3} {2} \chi _{ [0, 1] }$, $\varphi =\frac {3M} {4} \chi _{ [0, 1] }$, $\mu _1=\frac{6M}{5E}$, $\mu _2=\frac{2M}{5E}$ (cf. Figure 2.1  below). 
\end{proof}

\begin{figure}[ptb]
\begin{center}
\begin{tikzpicture}[scale=0.7]
\draw[->] (-7,-5) -- (1,-5) node[right]{\scriptsize$\epsilon$};
\draw[->] (-5,-6) -- (-5,3.3);
\draw[-] (-4, -4.8) -- (-2, -4.8);
\draw(-2,-5)--(-2,-5.1)node[below]{\scriptsize$1$};
\draw(-5, -4.8)--(-5, -4.8)node{-};
\draw(-5, -4.8)--(-5, -4.8)node[left]{\scriptsize$M/2$};
\draw[-] (-4.7, -4.3) -- (-4, -4.3);
\draw(-4,-5)--(-4,-5.1)node[below]{\scriptsize$1/3$};
\draw(-5, -4.3)--(-5, -4.3)node{-};
\draw(-5, -4.3)--(-5, -4.3)node[left]{\scriptsize$\frac {2+3^{3/2}} {4}M$};
\draw (-4.7, -5) -- (-4.7, -5.1)node[below]{\scriptsize$\rho $};
\draw[-] (-5, 2.8) -- (-4.7, 2.8);
\draw(-5, 2.8)--(-5, 2.8)node{-};
\draw(-5, 2.8)--(-5, 2.8)node[left]{\scriptsize$\frac {3} {2\rho ^{\beta +\frac {1} {2}}}$};
\draw[dotted](-4.7, -4.3)--(-4.7, 2.8);
\draw[dotted](-4, -4.8)--(-4, -4.3);
\draw[dotted](-2, -5)--(-2, -4.8);
\draw [fill=gray] (-5, -5)-- (-5, 2.8) -- (-4.7, 2.8) --(-4.7 , -4.3) --(-4, -4.3)--(-4, -4.8)--(-2, -4.8)--(-2, -5);
\end{tikzpicture}
\caption{An example of initial datum (\ref{Datos}) for $M>0, \frac {2M} {5E}=1, 0<\rho <<1.$}
\end{center}
\end{figure}
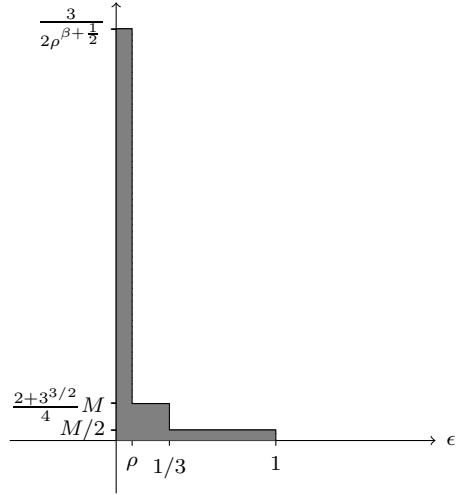
It is apparent from the above arguments that the function $f_0$ may be chosen as regular as desired. This shows that the onset of blow up is not related with the regularity of the solution.
\begin{remark}
The result is local for given
values of $E$ and $M$ in the sense that the main requirements for blow-up
are to have a lower estimate for $f_{0}$ for small energies as well as a
sufficiently large amount of particles in the interval $\epsilon\in\left[
0,\rho\right] $. 
\end{remark}
\subsection{Weak solutions.}

The theory of weak solutions of (\ref{F3E2}), (\ref{F3E3}) has been
developed by X. Lu in \cite{Lu1, Lu2, Lu3}. It  allows to deal with
measured valued solutions and suits very well to the purpose of  considering the finite time formation of Dirac 
mass in the solutions of (\ref{F3E2}), (\ref{F3E3}).

Since we are interested in the condensation
phenomena, it is convenient to use the equation for the mass density $g$,
instead of $f$ instead of (\ref{F3E2}), (\ref{F3E3}). We will denote as $%
\mathcal{M}_{+}\left( \mathbb{R}^{+};1+\epsilon\right) $ the set of Radon
measures $g$ in $\mathbb{R}^{+}$ satisfying:%
\begin{equation*}
\int\left( 1+\epsilon\right) g\left( \epsilon\right) d\epsilon<\infty 
\end{equation*}

We will use the notation $g\left( \epsilon\right) $ in spite of the fact
that $g$ is a measure.

\begin{definition}
\label{weak}We will say that $g\in C\left( \left[ 0,T\right) ;\mathcal{M}%
_{+}\left( \mathbb{R}^{+};\left( 1+\epsilon\right) \right) \right) $ is a
weak solution of (\ref{F3E4}), (\ref{F3E5}) on $(0, T)$, with initial datum $g_{0}\in%
\mathcal{M}_{+}\left( \mathbb{R}^{+};1+\epsilon\right)$,  if, \ for any $%
\varphi\in C_{0}^{2}\left( \left[ 0,T\right) ,\left[ 0,\infty\right) \right)$, the following identity holds:%
\begin{align}
&-\int_{\mathbb{R}^{+}}g_{0}\left( \epsilon\right) \varphi\left(
0,\epsilon\right) d\epsilon  =\int_{0}^{T}\int_{\mathbb{R}%
^{+}}g\partial_{t}\varphi d\epsilon dt+  \notag \\
& \hskip1cm+\frac{1}{2^{\frac{5}{2}}}%
\int_{0}^{T}\int_{\mathbb{R}^{+}}\int_{\mathbb{R}^{+}}\int_{\mathbb{R}^{+}}%
\frac {g_{1}g_{2}g_{3}\Phi}{\sqrt{\epsilon_{1}\epsilon_{2}\epsilon_{3}}}%
Q_{\varphi }d\epsilon_{1}d\epsilon_{2}d\epsilon_{3}dt  +\notag \\
& \hskip2cm+\frac{\pi}{2}\int_{0}^{T}\int_{\mathbb{R}^{+}}\int_{\mathbb{R}%
^{+}}\int_{\mathbb{R}^{+}}\frac{g_{1}g_{2}\Phi}{\sqrt{\epsilon_{1}\epsilon
_{2}}}Q_{\varphi}d\epsilon_{1}d\epsilon_{2}d\epsilon_{3}dt\ \ \ 
\label{Z1E2N}
\end{align}
where $\Phi$ is as in (\ref{F3E5}) and:%
\begin{equation}
Q_{\varphi}=\varphi\left( \epsilon_{3}\right) +\varphi\left( \epsilon
_{1}+\epsilon_{2}-\epsilon_{3}\right) -2\varphi\left( \epsilon_{1}\right) 
\label{Z1E4N}
\end{equation}
\end{definition}

\begin{definition}
\label{weakf}
We say that $f$ is a weak solution of  (\ref{F3E2}), (\ref{F3E3}) in $(0, T)$ with initial datum $f_0$, if  
$g_0=4\pi\sqrt{2\epsilon}f_0\left(\epsilon\right)\in
\mathcal{M}_{+}\left( \mathbb{R}^{+};1+\epsilon\right)$, and  $g=4\pi\sqrt{2\epsilon}f\left( t,\epsilon\right)  
\in C\left( \left[ 0,T\right) ;\mathcal{M}
_{+}\left( \mathbb{R}^{+};\left( 1+\epsilon\right) \right) \right) $ is a
weak solution of (\ref{F3E4}), (\ref{F3E5}) on $(0, T)$, with initial datum $g_{0}$ in the sense of Definition \ref{weak}.
\end{definition}

\begin{rem}
We prove in Lemma \ref{der17} that every mild solution $f$  in the sense
of Definition \ref{mild} is a weak solution in the sense of Definition \ref
{weakf}.
\end{rem}

It has been proved by X. Lu in Theorem 2 of \cite{Lu1} that for all  $g_{0}\in
\mathcal{M}_{+}\left( \mathbb{R}^{+};1+\epsilon\right)$, there exists a  global weak solution of (\ref{F3E4}), (\ref{F3E5})
 in the sense of the Definition \ref{weak}. This solution satisfies (\ref{Z1E2N}), (\ref{Z1E4N}) for all $\varphi \in C^2_b([0, \infty))$ and with $T=\infty$. Moreover  the two quantities $\int _{ [0, \infty) } g(t, \epsilon )d\epsilon$  and $\int _{ [0, \infty) }\epsilon g(t, \epsilon )d\epsilon$ are constant in time, for all $t\ge 0$.

\subsection{Finite time condensation results.}
We may now state our two results on finite time condensation.

\begin{theorem}
\label{Cond1} There exist $\
\theta_{\ast}>0$ with the following property. For all  $M>0,\ E>0,$ $\nu>0$,  there exists $ 
\rho_0=\rho_0\left( M,E,\nu\right) >0$, $ K^{\ast}=K^{\ast}\left( M,E,\nu\right)
>0,\ T_{0}=T_{0}\left( M,E\right) $ such that for any  weak solution of (\ref{F3E4}%
), (\ref{F3E5}) on $(0, T_0)$ in the sense of Definition \ref{weak} with $g_{0}\in\mathcal{%
M}_{+}\left( \mathbb{R}^{+};\left( 1+\epsilon\right) \right) $ satisfying$\ $%
\begin{eqnarray}
&&4\pi\sqrt{2}\int_{\mathbb{R}^{+}}g_{0}\left( \epsilon\right) d\epsilon  =M\
,\ \ 4\pi\sqrt{2}\int_{\mathbb{R}^{+}}g_{0}\left( \epsilon\right) \epsilon
d\epsilon=E\   \label{Z1E5N} \\
&& \sup _{ 0\le \rho \le \rho _0}\left[
 \min\left\{
\inf _{ 0\le R\le \rho  }\frac {1} {\nu R^{3/2}}\int _0^R g_0(\epsilon) d\epsilon,  \frac {1} {K^*\rho ^{\theta_*}}\int _0^\rho g_0(\epsilon) d\epsilon
 \right\}
 \right]\ge 1,  \label{Z1E6}
\end{eqnarray}
we have:
\begin{equation}
\sup_{0<t\leq T_{0}}\int_{\left\{ 0\right\} }g\left( t,\epsilon\right)
d\epsilon>0.\   \label{Z1E7N}
\end{equation}
\end{theorem}

Notice that the construction of the weak solutions of (\ref{F3E4}), (\ref%
{F3E5}) does not rule out the possibility of having ``instantaneous
condensation", i.e.:%
\begin{equation*}
\sup_{0<t\leq T^{\ast}}\int_{\left\{ 0\right\} }g\left( t,\epsilon\right)
d\epsilon>0 
\end{equation*}
for any $T^{\ast}>0$. However, it is possible to construct weak solutions of (\ref%
{F3E4}), (\ref{F3E5}) such that $\sup_{0<t\leq T_{\ast}}\int_{\left\{
0\right\} }g\left( t,\epsilon\right) d\epsilon=0$ for some $0<T^{\ast
}<T_{0},$ but satisfying (\ref{Z1E7N}). For such solutions we would have then
condensation in a finite, but positive time, as it has been suggested in the
physical literature (cf. \cite{JPR}, \cite{LLPR}, \cite{ST1}, \cite{ST2}).
More precisely, we have:

\begin{theorem}
\label{Theoremtenfive}
There exist $\
\theta_{\ast}>0$ with the following property. For all  $M>0,\ E>0,$ $\nu>0,\ \gamma>3$,  there exists $ 
\rho_0=\rho_0\left( M,E,\nu\right) >0$, $ K^{\ast}=K^{\ast}\left( M,E,\nu\right)
>0,\ T_{0}=T_{0}\left( M,E\right) $ such that for any  $f_{0}\in
L^{\infty}\left( \mathbb{R}^{+};\left( 1+\epsilon\right) ^{\gamma}\right) $
satisfying
\begin{eqnarray}
&&\hskip -1cm 4\pi\sqrt{2}\int_{\mathbb{R}^{+}}f_{0}\left( \epsilon\right) \sqrt{\epsilon }%
d\epsilon  =M\ ,\ \ 4\pi\sqrt{2}\int_{\mathbb{R}^{+}}f_{0}\left(
\epsilon\right) \sqrt{\epsilon^{3}}d\epsilon=E\   \label{Z1E8N} \\
&&\hskip -1cm \sup _{ 0\le \rho \le \rho _0}\left[
 \min\left\{
\inf _{ 0\le R\le \rho  }\frac {1} {\nu R^{3/2}}\int _0^R f_0(\epsilon)\sqrt \epsilon d\epsilon,  \frac {1} {K^*\rho ^{\theta_*}}\int _0^\rho f_0(\epsilon)\sqrt \epsilon d\epsilon
 \right\}
 \right]\ge 1,  \label{Z1E9}
\end{eqnarray}
there exists a weak solution $g$ of (\ref{F3E4}), (\ref{F3E5}),  with $%
g_{0}\left( \epsilon\right) =4\pi\sqrt{2\epsilon}f_{0}\left( \epsilon\right)
$, and there exists $T_{\ast}>0$ such that the following holds:%
\begin{equation}
\sup_{0\leq t\leq T_{\ast}}\left\Vert f\left( t,\cdot\right) \right\Vert
_{L^{\infty}\left( \mathbb{R}^{+}\right) }<\infty\ \ ,\ \ \sup_{T_{\ast
}<t\leq T_{0}}\int_{\left\{ 0\right\} }g\left( t,\epsilon\right)
d\epsilon>0\   \label{Z2E1}
\end{equation}
where $g=4\pi\sqrt{2\epsilon}f.$
\end{theorem}

\section{Local well-posedness Theorem.}

\setcounter{equation}{0} \setcounter{theorem}{0}

As a first step we prove the local well-posedness of (\ref{F3E2}), (\ref%
{F3E3}) for initial data in $L^{\infty}\left( \mathbb{R}^{+};\left(
1+\epsilon\right) ^{\gamma}\right) $ with $\gamma>3.$ We first summarize
some topological results which will be used in the arguments.

\subsection{Topological preliminaries.}

We introduce a concept of weak topology in the spaces $L^{\infty}\left( 
\mathbb{R}^{+};\left( 1+\epsilon\right) ^{\gamma}\right) $ by means of the
functionals:%
\begin{equation}
L_{\varphi}\left[ f\right] =\int_{\mathbb{R}^{+}}f\varphi d\epsilon \ \ ,\ \
\varphi\in C_{0}\left( \left[ 0,\infty\right) \right) .   \label{top}
\end{equation}
We will denote as $\Theta$ the topology induced by the functionals $\varphi$
in $L^{\infty}\left( \mathbb{R}^{+};\left( 1+\epsilon\right) ^{\gamma
}\right)$. We define also the set: 
\begin{equation*}
\mathcal{K}_{R}^{\gamma}=\left\{ f\geq0,\ f\in L^{\infty}\left( \mathbb{R}%
^{+};\left( 1+\epsilon\right) ^{\gamma}\right) :\left\Vert f\right\Vert
_{L^{\infty}\left( \mathbb{R}^{+};\left( 1+\epsilon\right) ^{\gamma}\right)
}\leq R\right\} \ \ ,\ \ 0<T<\infty. 
\end{equation*}

\begin{lemma}
\label{topol}For any $R>0,\ \gamma>1,$ the topological space $\left( 
\mathcal{K}_{R}^{\gamma},\Theta\right) $ is metrizable and compact.
\end{lemma}

\begin{proof}
In order to check that the topology $\Theta$ is metrizable in $\mathcal{K}%
_{R}^{\gamma}$ we select a countable set of functions $\bar{\varphi}_{n}\in
C_{0}\left( \left[ 0,\infty\right) \right) $ which is dense in $C_{0}\left( %
\left[ 0,\infty\right) \right) $ in the sense of the uniform topology in
compact sets. We then define a metric in $\mathcal{K}_{R}^{\gamma }$ by
means of:%
\begin{equation}
\operatorname*{dist}\left( f,g\right) =\sum_{n}\frac{2^{-n}}{\left\Vert \bar{%
\varphi}_{n}\right\Vert _{L^{1}\left( \left[ 0,\infty\right) \right) }}%
\left\vert \int_{\left[ 0,\infty\right) }\left( f-g\right) \bar{\varphi }%
_{n}d\epsilon\right\vert .   \label{metTop}
\end{equation}

It is now standard to check that every set in $\Theta$ contains a ball with
the form $\left\{ g\in\mathcal{K}_{T,R}^{\gamma}:\operatorname*{dist}\left(
f_{0},g\right) <\delta\right\} $ for some $f_{0}\in\mathcal{K}%
_{T,R}^{\gamma} $ and $\delta>0$ and also that such a ball contains a set of 
$\Theta.$

In order to prove the compactness of $\mathcal{K}_{R}^{\gamma}$ with this
topology it is enough to prove, due to the metrizability of the space, that
every sequence has a convergent subsequence. Given a sequence $\left\{
f_{n}\right\} \subset\mathcal{K}_{R}^{\gamma}$ we have, since $\gamma>1$,
that $\int_{\mathbb{R}^{+}}f_{n}d\epsilon\leq CR$ for some $C$ independent
on $n.$ Then, there exists a subsequence of $\left\{ f_{n}\right\} $ which
will be denoted by the same indexes, as well as a Radon measure $\mu\in 
\mathcal{M}_{+}\left( \mathbb{R}^{+}\right) $ such that $f_{n}%
\rightharpoonup\mu$ as $n\rightarrow\infty$ in $\mathcal{M}_{+}\left( 
\mathbb{R}^{+}\right) .$

The definition of $L^{\infty}\left( \mathbb{R}^{+};\left( 1+\epsilon\right)
^{\gamma}\right) $ and $\mathcal{K}_{R}^{\gamma}$ imply:%
\begin{equation*}
\int_{\mathbb{R}^{+}}f_{n}\left( 1+\epsilon\right) ^{\gamma}\varphi
d\epsilon\leq R\int_{\mathbb{R}^{+}}\varphi d\epsilon=R\left\Vert
\varphi\right\Vert _{L^{1}\left( \left[ 0,\infty\right) \right) }, 
\end{equation*}
for any $\varphi\in C_{0}\left( \left[ 0,\infty\right) \right) .$ Taking
subsequences and passing to the limit in this formula we obtain:%
\begin{equation*}
\int_{\mathbb{R}^{+}}\mu\left( 1+\epsilon\right) ^{\gamma}\varphi
d\epsilon\leq R\int_{\mathbb{R}^{+}}\varphi d\epsilon=R\left\Vert
\varphi\right\Vert _{L^{1}\left( \left[ 0,\infty\right) \right) }, 
\end{equation*}
and this implies that $\mu\left( 1+\epsilon\right) ^{\gamma}\in\left(
L^{1}\left( \left[ 0,\infty\right) \right) \right) ^{\ast}=L^{\infty }\left( %
\left[ 0,\infty\right) \right) .$ (cf. \cite{B}, \cite{DS}). A similar
argument yields $\mu\geq0.$ Then $\mu\in\mathcal{K}_{R}^{\gamma}$ and the
result follows.
\end{proof}

We now define the space of functions in which we will prove well posedness.
We will denote as $\mathcal{X}_{R,T}^{\gamma}$ the space of functions $%
C\left( \left[ 0,T\right] ;\left( \mathcal{K}_{R}^{\gamma},\Theta\right)
\right) $ endowed with the metric:%
\begin{equation}
\operatorname*{dist}_{\mathcal{X}_{R,T}^{\gamma}}\left( f_{1},f_{2}\right)
=\sup_{t\in\left[ 0,T\right] }d_{\mathcal{K}_{R}^{\gamma}}\left( f_{1}(t),
f_{2}(t)\right) ,   \label{metr}
\end{equation}
where $d_{\mathcal{K}_{R}^{\gamma}}$ is the distance associated to the
topological space $\left( \mathcal{K}_{R}^{\gamma},\Theta\right) $ (cf.
Lemma \ref{topol}). The space $\mathcal{X}_{R,T}^{\gamma}$ is a complete
metric space for any $R>0,\ 0<T<\infty,\ \gamma>1.$ We recall in the next
Proposition the characterization of the compact sets of $\mathcal{X}%
_{R,T}^{\gamma}$.

\begin{proposition}
\label{AsAr}Let $R>0,\ 0<T<\infty,\ \gamma>1.$ Suppose that $\mathcal{F}%
\subset\mathcal{X}_{R,T}^{\gamma}$ is an equicontinuous family of functions,
i.e. for any $\varepsilon>0$ there exists $\delta>0$ such that, for any $f\in%
\mathcal{F}$ and any $t_{1},t_{2}\in\left[ 0,T\right] $ satisfying $%
\left\vert t_{1}-t_{2}\right\vert <\delta$ we have $d_{\mathcal{K}%
_{R}^{\gamma}}\left( f\left( t_{1}\right) ,f\left( t_{2}\right) \right)
<\varepsilon.$\ Then, the family $\mathcal{F}$ is compact in the topology of 
$\mathcal{X}_{R,T}^{\gamma}.$
\end{proposition}

\begin{proof}
It is just a consequence of Arzel\`{a}-Ascoli Theorem for continuous
functions with values in general compact metric spaces (cf. \cite{DS}, \cite%
{F}).
\end{proof}

It will be convenient to reformulate Proposition \ref{AsAr} in a form that
is more convenient to use in terms of test functions.

\begin{proposition}
\label{compCrit}Suppose that $\mathcal{F}\subset\mathcal{X}_{R,T}^{\gamma}$
is a family of functions satisfying that, for any $\varphi\in C_{0}\left( %
\left[ 0,\infty\right) \right) $ and any $\varepsilon>0$ there exists $%
\delta>0$ such that, for any $f\in\mathcal{F}$ and any $t_{1},t_{2}\in\left[
0,T\right] $ satisfying $\left\vert t_{1}-t_{2}\right\vert <\delta$ the
functions $\psi_{\varphi}\left( t;f\right) =\int_{\mathbb{R}^{+}}f\left(
t,\epsilon\right) \varphi\left( \epsilon\right) d\epsilon$ satisfy:%
\begin{equation*}
\left\vert \psi_{\varphi}\left( t_{1};f\right) -\psi_{\varphi}\left(
t_{2};f\right) \right\vert <\varepsilon, 
\end{equation*}
for any $f\in\mathcal{F}$. Then, the family $\mathcal{F}$ is compact in $%
\mathcal{X}_{R,T}^{\gamma}.$
\end{proposition}

\begin{proof}
We use the metric defined in (\ref{metTop}). Given $\varepsilon>0$ as well
as the definition of $\mathcal{X}_{R,T}^{\gamma}$ it follows that there
exists $N$ large enough such that 
\begin{equation*}
\sum_{n\geq N}\frac{2^{-n}}{\left\Vert \bar{\varphi}_{n}\right\Vert
_{L^{1}\left( \left[ 0,\infty\right) \right) }}\left\vert \int_{\left[
0,\infty\right) }\left( f\left( t_{1}\right) -f\left( t_{2}\right) \right) 
\bar{\varphi}_{n}d\epsilon\right\vert <\frac{\varepsilon}{2}, 
\end{equation*}
for any $t_{1},t_{2}\in\left[ 0,T\right] .$ On the other hand, using the
property satisfied by the family $\mathcal{F}$ it follows that, there exists 
$\delta>0$ such that, if $\left\vert t_{1}-t_{2}\right\vert <\delta$ we have:%
\begin{equation*}
\sum_{n<N}\frac{2^{-n}}{\left\Vert \bar{\varphi}_{n}\right\Vert
_{L^{1}\left( \left[ 0,\infty\right) \right) }}\left\vert \int_{\left[
0,\infty\right) }\left( f\left( t_{1}\right) -f\left( t_{2}\right) \right) 
\bar{\varphi }_{n}d\epsilon\right\vert <\frac{\varepsilon}{2}. 
\end{equation*}

Therefore $d_{\mathcal{K}_{R}^{\gamma}}\left( f\left( t_{1}\right) ,f\left(
t_{2}\right) \right) <\varepsilon$ and applying Proposition \ref{AsAr} the
result follows.
\end{proof}

\subsection{Statement of the Local Well-Posedness Theorem.}

\begin{theorem}
\label{localExistence}Suppose that $f_{0}\in L^{\infty}\left( \mathbb{R}%
^{+};\left( 1+\epsilon\right) ^{\gamma}\right) $ with $\gamma>3.$ There
exists $T>0$, depending only on $\left\Vert f_{0}\left( \cdot\right)
\right\Vert _{L^{\infty}\left( \mathbb{R}^{+};\left( 1+\epsilon\right)
^{\gamma}\right) }$, and there exists a unique mild solution of (\ref{F3E2}%
), (\ref{F3E3}), $f\in L_{loc}^{\infty}\left( \left[ 0,T\right) ;L^{\infty
}\left( \mathbb{R}^{+};\left( 1+\epsilon\right) ^{\gamma}\right) \right) $
in the sense of Definition \ref{mild}.

The obtained solution $f$ satisfies:%
\begin{equation}
4\pi\sqrt{2}\int_{0}^{\infty}f_{0}\left( \epsilon\right) \epsilon
^{w}d\epsilon=4\pi\sqrt{2}\int_{0}^{\infty}f\left( t,\epsilon\right)
\epsilon^{w}d\epsilon\, ,\ \ t\in\left( 0,T\right) \, ,\ \ \ w\in\left\{ 
\frac{1}{2},\ \frac{3}{2}\right\} .   \label{F3E6c}
\end{equation}

The function $f$ is in the space $W^{1,\infty}\left( \left( 0,T\right)
;L^{\infty}\left( \mathbb{R}^{+}\right) \right) $ and it satisfies (\ref%
{F3E2}) $a.e.\ \epsilon\in\mathbb{R}^{+}$ for any $t\in\left( 0,T_{\max
}\right) .$ Moreover, $f$ can be extended as a mild solution of (\ref{F3E2}%
), (\ref{F3E3}) to a maximal time interval $\left( 0,T_{\max}\right) $ with $%
0<T_{\max}\leq\infty.$ If $T_{\max}<\infty$ we have:%
\begin{equation*}
\limsup_{t\rightarrow T_{\max}^{-}}\left\Vert f\left( t,\cdot\right)
\right\Vert _{L^{\infty}\left( \mathbb{R}^{+}\right) }=\infty. 
\end{equation*}
\end{theorem}

We will split the Proof of Theorem \ref{localExistence}. We first prove the
existence of one mild solution using Schauder's fixed point Theorem. Suppose
that $T>0.$ We define for each $\gamma>3$ the following auxiliary operator $%
\mathcal{T}:\mathcal{X}_{T}^{\gamma}\rightarrow\mathcal{X}_{T}^{\gamma}:$%
\begin{align}
\mathcal{T}\left( f\right) \left( t,\epsilon_{1}\right) &=f_{0}\left(
\epsilon_{1}\right) \Psi\left( t,\epsilon_{1}\right)+\notag\\
& +\frac{8\pi^{2}}{\sqrt{2%
}}\int_{0}^{t}\frac{\Psi\left( t,\epsilon_{1}\right) }{\Psi\left(
s,\epsilon_{1}\right) }\int_{0}^{\infty}\!\!\int_{0}^{\infty}\!\!f_{3}f_{4}\left(
1+f_{1}+f_{2}\right) Wd\epsilon_{3}d\epsilon_{4}ds,   \label{F3E8}
\end{align}
where $\Psi$ is as in (\ref{F3E6a}). Given $\gamma>3$ we define the
following functional $\psi_{\gamma}:L^{\infty}\left( \mathbb{R}^{+};\left(
1+\epsilon\right) ^{\gamma}\right) \rightarrow\mathbb{R}^{+}.$ 
\begin{equation}
\psi_{\gamma}\left[ f\right] =\left\Vert f\right\Vert _{L^{\infty}\left( 
\mathbb{R}^{+};\left( 1+\epsilon\right) ^{\gamma}\right) }+\int_{0}^{\infty}
\epsilon^{\frac{3}{2}}f\left( \epsilon\right) d\epsilon.   \label{F8E2}
\end{equation}

The following estimate will play a crucial role in all this Section.

\begin{proposition}
\label{functest}Let $\gamma>3.$ Suppose that the operator: 
$$\mathcal{T}%
:L_{loc}^{\infty}\left( \left[ 0,T\right) ;L^{\infty}\left( \mathbb{R}%
^{+};\left( 1+\epsilon\right) ^{\gamma}\right) \right) \rightarrow
L_{loc}^{\infty}\left( \left[ 0,T\right) ;L^{\infty}\left( \mathbb{R}%
^{+};\left( 1+\epsilon\right) ^{\gamma}\right) \right) $$ 
is defined as in (%
\ref{F3E8}). There exists $0<\theta<1$ and $C>0$ both of them depending only
on $\gamma$ such that for any $0\leq t\leq T:$%
\begin{align}
& \psi_{\gamma}\left[ f\left( t,\cdot\right) \right] \leq\psi_{\gamma }\left[
f_{0}\right] +Ct\left( 1+\sup_{0\leq s\leq t}\left\Vert f\left(
t,\cdot\right) \right\Vert _{L^{\infty}\left( \mathbb{R}^{+}\right) }\right)
\left( \sup_{0\leq s\leq t}\int f\left( s,\epsilon\right) d\epsilon\right)
^{2}+  \notag \\
& +Ct\!\!\sup_{0\leq s\leq t}\left\Vert f\left( s,\cdot\right) \right\Vert
_{L^{\infty}\left( \mathbb{R}^{+};\left( 1+\epsilon\right) ^{\gamma }\right)
}\times \nonumber \\
&\times \left[ \left( 1+2\sup_{0\leq s\leq t}\left\Vert f\left( s,\cdot\right)
\right\Vert _{L^{\infty}\left( \mathbb{R}^{+}\right) }\right) \!\!
\sup_{0\leq s\leq t}\int_{0}^{\infty}\left( \sqrt{\epsilon }+\left(
\epsilon\right) ^{\frac{3}{2}}\right) f\left( s,\epsilon\right) d\epsilon+
\right.  \notag \\
& \left. \hskip 1cm +\left( \sup_{0\leq s\leq t}\int\left( 1+\epsilon^{%
\frac{3}{2}}\right) f\left( s,\epsilon\right) d\epsilon\right) ^{2}\right] 
 +\theta\sup_{0\leq s\leq t}\left\Vert f\left( s,\cdot\right) \right\Vert
_{L^{\infty}\left( \mathbb{R}^{+};\left( 1+\epsilon\right) ^{\gamma }\right)
}.   \label{F6E1}
\end{align}
\end{proposition}

In order to prove Proposition \ref{functest} we need two auxiliary Lemmas.

\begin{lemma}
\label{IntEst}
\label{LinfEst}Let $\gamma>3$ and $\mathcal{T}$ be as in 
Proposition (\ref{functest}). There exists $C>0$ depending
only on $\gamma$ such that:%
\begin{align}
\int_{0}^{\infty} \epsilon^{\frac{3}{2}}\mathcal{T} (f)\left( t,\epsilon
\right) d\epsilon\leq & \int_{0}^{\infty}\left( \epsilon\right) ^{\frac {3}{2%
}}f_{0}\left( \epsilon\right) d\epsilon+Ct\left( 1+2\sup_{0\leq s\leq
t}\left\Vert f\left( s,\cdot\right) \right\Vert _{L^{\infty}\left( \mathbb{R}%
^{+}\right) }\right) \times  \notag \\
& \hskip -2cm \times\left( \sup_{0\leq s\leq t}\left\Vert f\left( s,\cdot\right)
\right\Vert _{L^{\infty}\left( \mathbb{R}^{+};\left( 1+\epsilon\right)
^{\gamma}\right) }\right) \left( \sup_{0\leq s\leq t}\int_{0}^{\infty}\sqrt{%
\epsilon}f\left( s,\epsilon\right) d\epsilon\right) ,   \label{T4E4a}
\end{align}
for $0\leq t\leq T.$
\end{lemma}

\begin{proof}
Notice that, using the symmetry of the integral under the exchange $%
\epsilon_{3}\leftrightarrow\epsilon_{4}$ as well as the fact that $W\leq 
\sqrt{\frac{\epsilon_{4}}{\epsilon_{1}}}$ if $\epsilon_{3}\geq\epsilon_{4}$
and that $\gamma>3.$:%
\begin{align}
& \left\Vert \int\int Wf_{3}f_{4}\left( 1+f_{1}+f_{2}\right) d\epsilon
_{3}d\epsilon_{4}\right\Vert _{L^{1}\left( \mathbb{R}^{+};\epsilon^{\frac {3%
}{2}}d\epsilon\right) }  \notag \\
& \leq C\left( 1+2\left\Vert f\right\Vert _{L^{\infty}\left( \mathbb{R}%
^{+}\right) }\right) \left\Vert f\right\Vert _{L^{\infty}\left( \mathbb{R}%
^{+};\left( 1+\epsilon\right) ^{\gamma}\right) }\int_{0}^{\infty }\frac{%
d\epsilon_{1}}{\left( 1+\epsilon_{1}\right) ^{\gamma-2}}\int
_{0}^{\infty}f_{4}\sqrt{\epsilon_{4}}d\epsilon_{4}  \label{F4E3} \\
& \leq C\left( 1+2\left\Vert f\right\Vert _{L^{\infty}\left( \mathbb{R}%
^{+}\right) }\right) \left\Vert f\right\Vert _{L^{\infty}\left( \mathbb{R}%
^{+};\left( 1+\epsilon\right) ^{\gamma}\right) }\int_{0}^{\infty }f_{4}\sqrt{%
\epsilon_{4}}d\epsilon_{4}.\   \notag
\end{align}

Using the fact that $f\geq0$ we obtain $0\leq\Psi\left( \epsilon
_{1},t\right) \leq1$,\ $0\leq\frac{\Psi\left( \epsilon_{1},t\right) }{%
\Psi\left( \epsilon_{1},s\right) }\leq1,$ whence (\ref{T4E4a}) follows.
\end{proof}

We need to derive detailed estimates of the function $a\left( t,\epsilon
_{1}\right) $ in (\ref{F3E6a}).

\begin{lemma}
\label{RepA}Suppose that $f\in L^{\infty}\left( \mathbb{R}^{+};\left(
1+\epsilon\right) ^{\gamma}\right) $ with $\gamma>3.$ The function $a\left(
t,\epsilon_{1}\right) $ defined in (\ref{F3E6a}) can be written as:%
\begin{equation}
a\left( t,\epsilon_{1}\right) =\frac{8\pi^{2}\sqrt{\epsilon_{1}}}{\sqrt{2}}%
\int_{0}^{\infty}f\left( t,\epsilon\right) \sqrt{\epsilon}d\epsilon +S\left[
f\right] \left( t,\epsilon_{1}\right) , \   \label{F8E6}
\end{equation}
where 
\begin{equation}
S\left[ f\right] \left( t,\epsilon_{1}\right) =S_{1}\left[ f\right] \left(
t,\epsilon_{1}\right) +S_{2}\left[ f\right] \left( t,\epsilon _{1}\right) , 
\label{F8E5a}
\end{equation}
with:%
\begin{equation}
S_{1}\left[ f\right] =\frac{8\pi^{2}\epsilon_{1}}{\sqrt{2}}\int_{0}^{\infty
}f_{2}\omega\left( \frac{\epsilon_{2}}{\epsilon_{1}}\right) d\epsilon _{2}\
\ ,\ \ S_{2}\left[ f\right] =\frac{8\pi^{2}}{\sqrt{2}}\int_{0}^{\infty}%
\int_{0}^{\infty}f_{2}\left( f_{3}+f_{4}\right) Wd\epsilon
_{3}d\epsilon_{4},   \label{F8E5}
\end{equation}
and:%
\begin{equation}
\omega\left( x\right) =\frac{x^{\frac{3}{2}}}{3}\ \ ,\ \
x\leq1\text{\ \ and \ }\omega\left( x\right) =\left( x-\sqrt {x}+\frac{1}{3}%
\right) \, ,\ \ x\geq1.   \label{F8E7}
\end{equation}

If $f\geq0$ we have $S_{1}\left[ f\right] \geq0$ and $S_{2}\left[ f\right]
\geq0.$
\end{lemma}

\begin{proof}
Using the change of variables $\epsilon_{2}=\epsilon_{3}+\epsilon_{4}-%
\epsilon_{1}\ \ ,\ \ \xi=\epsilon_{3}-\epsilon_{4}$ we obtain:%
\begin{align*}
a\left( t,\epsilon_{1}\right) &=\frac{4\pi^{2}}{\sqrt{2}}\int_{0}^{\infty
}f_{2}d\epsilon_{2}\int_{-\left( \epsilon_{2}+\epsilon_{1}\right) }^{\left(
\epsilon_{2}+\epsilon_{1}\right) }d\xi W\left( \epsilon_{1},\epsilon _{2},%
\frac{\epsilon_{2}+\epsilon_{1}+\xi}{2},\frac{\epsilon_{2}+\epsilon _{1}-\xi%
}{2}\right)+\notag\\
& +S_{2}\left[ f\right] \left( t,\epsilon_{1}\right) , 
\end{align*}
where $S_{2}\left[ f\right] $ is as in (\ref{F8E5}).

Using the symmetry of the function $W\left( \epsilon_{1},\epsilon_{2},\frac{%
\epsilon_{2}+\epsilon_{1}+\xi}{2},\frac{\epsilon_{2}+\epsilon_{1}-\xi }{2}%
\right) $ with respect to the transformation $\xi\rightarrow\left(
-\xi\right) ,$ as well as the definition of $W,\ $we obtain:%
\begin{equation*}
a\left( t,\epsilon_{1}\right) =\frac{8\pi^{2}\epsilon_{1}}{\sqrt{2}}\int
_{0}^{\infty}f_{2}\Omega\left( \frac{\epsilon_{2}}{\epsilon_{1}}\right)
d\epsilon_{2}+S_{2}\left[ f\right] \left( t,\epsilon_{1}\right) , 
\end{equation*}
with:%
\begin{equation*}
\Omega\left( x\right) =\int_{0}^{\left( x+1\right) }d\xi W\left( 1,x,\frac{%
x+1-\xi}{2}\right) \, ,\ \ x\geq0. 
\end{equation*}

We can compute $\Omega\left( x\right) $ treating separately the cases $%
x\leq1 $ and $x>1.$ Using the definition of $W$ we obtain:%
\begin{equation*}
\Omega\left( x\right) =\sqrt{x}+\frac{\left( x\right) ^{\frac{3}{2}}}{3}\ \ 
\text{if\ }x\leq1\ \ ,\ \ \ \Omega\left( x\right) =\left( x+\frac {1}{3}%
\right) \ \ \text{if\ }x>1. 
\end{equation*}

We can then write $\Omega\left( x\right) =\sqrt{x}+\omega\left( x\right) $
with $\omega\left( \cdot\right) $ as in (\ref{F8E7}). This gives (\ref{F8E6}%
). Using the fact that $\omega\left( x\right) \geq0$ we conclude the Proof
of Lemma \ref{RepA}.
\end{proof}

The following Lemma is important to control the behaviour of $\mathcal{T}%
\left( f\right) \left( t,\epsilon\right) $ for large values of $\epsilon.$
Its proof uses in a crucial way the structure of the quadratic terms of the
equation (\ref{F3E2}).

\begin{lemma}
\label{LinfEst}Let $\gamma>3$ and $\mathcal{T}$ be as in 
Proposition (\ref{functest}). There exists $\theta
\in\left( 0,1\right) $ and $C>0$ both of them depending only on $\gamma$ and
such that, for any $t\in\left[ 0,T\right] $ the following estimate holds:%
\begin{align}
\left\Vert \mathcal{T}\left( f\right) \right\Vert _{L^{\infty}\left( \mathbb{%
R}^{+};\left( 1+\epsilon\right) ^{\gamma}\right) }\left( t\right) &
\leq\left\Vert f_{0}\right\Vert _{L^{\infty}\left( \mathbb{R}^{+};\left(
1+\epsilon\right) ^{\gamma}\right) }+  \notag \\
&\hskip -2cm  +t\sup_{0\leq s\leq t}\left\Vert f\left( s,\cdot\right) \right\Vert
_{L^{\infty}\left( \mathbb{R}^{+};\left( 1+\epsilon\right) ^{\gamma }\right)
}\left( \sup_{0\leq s\leq t}\int\left( 1+\epsilon^{\frac{3}{2}}\right)
f\left( s,\epsilon\right) d\epsilon\right) ^{2}+  \notag \\
&\hskip -2cm  +Ct\left( 1+\sup_{0\leq s\leq t}\left\Vert f\left( s,\cdot\right)
\right\Vert _{L^{\infty}\left( \mathbb{R}^{+};\left( 1+\epsilon\right)
^{\gamma}\right) }\right) \left( \sup_{0\leq s\leq t}\int_{0}^{\infty
}f\left( s,\epsilon\right) d\epsilon\right) ^{2}+  \notag \\
&\hskip -2cm +Ct\sup_{0\leq s\leq t}\left\Vert f\left( s,\cdot\right) \right\Vert
_{L^{\infty}\left( \mathbb{R}^{+};\left( 1+\epsilon\right) ^{\gamma }\right)
}\left( \sup_{0\leq s\leq t}\int_{0}^{\infty}\left( 1+\left( \epsilon\right)
^{\frac{3}{2}}\right) f\left( s,\epsilon\right) d\epsilon\right) +  \notag \\
& \hskip -2cm  +\theta\sup_{0\leq s\leq t}\left\Vert f\left( s,\cdot\right) \right\Vert
_{L^{\infty}\left( \mathbb{R}^{+};\left( 1+\epsilon\right) ^{\gamma }\right)
}.   \label{F8E9}
\end{align}
\end{lemma}

\begin{proof}
We estimate the operator $\mathcal{T}\left( f\right) $ in the norm $%
L^{\infty}\left( \mathbb{R}^{+};\left( 1+\epsilon\right) ^{\gamma}\right) .$
Notice that (\ref{F3E8}) implies:%
\begin{equation}
\left\Vert \mathcal{T}\left( f\right) \right\Vert _{L^{\infty}\left( \mathbb{%
R}^{+};\left( 1+\epsilon\right) ^{\gamma}\right) }\leq\left\Vert
f_{0}\right\Vert _{L^{\infty}\left( \mathbb{R}^{+};\left( 1+\epsilon\right)
^{\gamma}\right) }+J_{1}+J_{2}+J_{3},\   \label{F8E1}
\end{equation}
with:%
\begin{align*}
J_{1} & =\frac{8\pi^{2}}{\sqrt{2}}\left\Vert \int_{0}^{t}\frac{\Psi\left(
t,\epsilon_{1}\right) }{\Psi\left( s,\epsilon_{1}\right) }\int_{0}^{\infty
}\int_{0}^{\infty}f_{3}f_{4}Wd\epsilon_{3}d\epsilon_{4}ds\right\Vert
_{L^{\infty}\left( \mathbb{R}^{+};\left( 1+\epsilon\right) ^{\gamma }\right)
}, \\
J_{2} & =\left\Vert \int_{0}^{t}\frac{\Psi\left( t,\epsilon_{1}\right) }{%
\Psi\left( s,\epsilon_{1}\right) }f_{1}\int_{0}^{\infty}\int_{0}^{\infty
}f_{3}f_{4}Wd\epsilon_{3}d\epsilon_{4}ds\right\Vert _{L^{\infty}\left( 
\mathbb{R}^{+};\left( 1+\epsilon\right) ^{\gamma}\right) }, \\
J_{3} & =\left\Vert \int_{0}^{t}\frac{\Psi\left( t,\epsilon_{1}\right) }{%
\Psi\left( s,\epsilon_{1}\right) }\int_{0}^{\infty}\int_{0}^{%
\infty}f_{3}f_{4}f_{2}Wd\epsilon_{3}d\epsilon_{4}ds\right\Vert
_{L^{\infty}\left( \mathbb{R}^{+};\left( 1+\epsilon\right) ^{\gamma}\right)
}.
\end{align*}

The terms $J_{2}$ can be readily estimated: 
\begin{equation}
J_{2}\leq t\sup_{0\leq s\leq t}\left\Vert f\left( s,\cdot\right) \right\Vert
_{L^{\infty}\left( \mathbb{R}^{+};\left( 1+\epsilon\right) ^{\gamma }\right)
}\left( \sup_{0\leq s\leq t}\int\left( 1+\epsilon^{\frac{3}{2}}\right)
f\left( s,\epsilon\right) d\epsilon\right) ^{2}.   \label{F8E3a}
\end{equation}

In order to estimate $J_{3}$ we use the symmetry in the variables $%
\epsilon_{3},\ \epsilon_{4}$ to obtain:%
\begin{equation}
J_{3}\leq2\int_{0}^{t}\left\Vert
\int_{0}^{\infty}f_{2}d\epsilon_{2}\int_{\left( 0,\infty\right)
}\chi_{\left\{ \epsilon_{3}\geq\epsilon _{4}\right\}
}f_{3}f_{4}d\epsilon_{4}\right\Vert _{L^{\infty}\left( \mathbb{R}^{+};\left(
1+\epsilon\right) ^{\gamma}\right) }ds.   \label{F8E3}
\end{equation}

Using the fact that in the region $\left\{ \epsilon_{3}\geq\epsilon
_{4}\right\} ,\ \epsilon_{2}\geq0$ we have $\epsilon_{3}\geq\frac {%
\epsilon_{1}}{2}$ as well as the definition of the norm $\left\Vert f\left(
s,\cdot\right) \right\Vert _{L^{\infty}\left( \mathbb{R}^{+};\left(
1+\epsilon\right) ^{\gamma}\right) }$ we arrive at: 
\begin{align}
J_{3} & \leq C\int_{0}^{t}\left\Vert f\left( s,\cdot\right) \right\Vert
_{L^{\infty}\left( \mathbb{R}^{+};\left( 1+\epsilon\right) ^{\gamma }\right)
}\left\Vert \frac{1}{\left( 2+\epsilon_{1}\right) ^{\gamma}}%
\int_{0}^{\infty}f_{2}d\epsilon_{2}\int_{\frac{\epsilon_{1}}{2}%
}^{\infty}f_{4}d\epsilon_{4}\right\Vert _{L^{\infty}\left( \mathbb{R}%
^{+};\left( 1+\epsilon\right) ^{\gamma}\right) }ds  \notag \\
& \leq Ct\sup_{0\leq s\leq t}\left\Vert f\left( s,\cdot\right) \right\Vert
_{L^{\infty}\left( \mathbb{R}^{+};\left( 1+\epsilon\right) ^{\gamma }\right)
}\left( \sup_{0\leq s\leq t}\int_{0}^{\infty}f\left( s,\epsilon\right)
d\epsilon\right) ^{2}.   \label{F8E4}
\end{align}

The term $J_{1}$ must be estimated more carefully. We use at this point
ideas closely related to those in \cite{Ca1},\cite{Ca2}. Suppose that $%
\frac{1}{2}<\mu<1$ and let us write $M_{0}\left( t\right) =\int g\left(
t,\epsilon\right) d\epsilon.$ This number is the total mass of the particles
and for solutions of (\ref{F3E2}), (\ref{F3E3}) can be expected to be
constant. However, this has not yet been proved and we must therefore keep $%
M_{0}\left( t\right) $ as a function of $t.$ Using (\ref{F3E3a}), (\ref%
{F3E6a}) and Lemma \ref{RepA} we obtain $\frac{\Psi\left( t,\epsilon
_{1}\right) }{\Psi\left( s,\epsilon_{1}\right) }\leq\exp\left( -\pi \sqrt{%
\epsilon_{1}}\int_{s}^{t}M_{0}\left( \xi\right) d\xi\right)$. Then,
splitting the domain of integration in $J_{1}$ in the two subdomains
indicated in Figure 3.1 we obtain: 
\begin{align*}
J_{1} & \leq J_{1,1}+J_{1,2}, \\
J_{1,1} & =\frac{8\pi^{2}}{\sqrt{2}}\left\Vert \int_{0}^{t}\exp\left( -\pi%
\sqrt{\epsilon_{1}}\int_{s}^{t}M_{0}\left( \xi\right) d\xi\right) \left(
\int_{\left( 1-\mu\right) \epsilon_{1}}^{\infty}f_{3}d\epsilon _{3}\right)
^{2}ds\right\Vert _{L^{\infty}\left( \mathbb{R}^{+};\left( 1+\epsilon\right)
^{\gamma}\right) }, \\
J_{1,2} & =\frac{16\pi^{2}}{\sqrt{2}}\left\Vert \int_{0}^{t}\!\exp\left( -\pi%
\sqrt{\epsilon_{1}}\int_{s}^{t}\!M_{0}\left( \xi\right) d\xi\right)\times\right.\\
&\hskip 4cm \left.\times\left( \int_{\mu\epsilon_{1}}^{\infty}\!f_{3}d\epsilon_{3}\right) \left(
\int_{0}^{\left( 1-\mu\right) \epsilon_{1}}\!\!\!\!\sqrt{\frac{\epsilon_{4}}{%
\epsilon_{1}}}f_{4}d\epsilon_{4}\right) ds\right\Vert _{L^{\infty}\left( 
\mathbb{R}^{+};\left( 1+\epsilon\right) ^{\gamma}\right) }.
\end{align*}

\begin{figure}[ptb]
\centerline{
\begin{tikzpicture}[scale=0.8]
\draw [lightgray, fill=lightgray] (2.8, 2.1)-- (-5,2.1) -- (-5, -2) --(-2, -5) -- (2.8,-5)--(2.8, 2.1);
\draw [pattern=vertical lines] (2.8, -4.255)-- (-2.75,-4.255) -- (-2.75, -5) --(2.8, -5);
\draw [pattern=dots] (2.8, 2.1)-- (-4.255,2.1) -- (-4.255, -4.255) --(2.8, -4.255) --(2.8, 2.1);
\draw [pattern=horizontal lines ] (-4.255, 2.1)-- (-4.255, -2.75) -- (-5, -2.75) --(-5, 2.1);
\draw[->] (-5.2,-5) -- (3.2,-5) node[below]{\scriptsize$\varepsilon _3$};
\draw[->] (-5,-5.2) -- (-5, 2.5) node[left]{\scriptsize$\varepsilon _4$};
\draw(-2, -5)--(-2, -5)node[below]{\scriptsize$\varepsilon _1$};
\draw(-5, -2)--(-5, -2)node[left]{\scriptsize$\varepsilon _1$};
\draw[-](-5, -2)--(-2, -5);
\draw[-](-2, -5)--(-2, -5.1);
\draw[-](-2.75, -5)--(-2.75, -5.1);
\draw[-](-4.255, -5)--(-4.255, -5.1);
\draw(-5.1, -2)--(-5, -2);
\draw(-5.1, -2.75)--(-5, -2.75);
\draw(-5.1, -4.255)--(-5, -4.255);
\draw(-2.75, -5)--(-2.75, -5)node[below]{\scriptsize$\mu\, \varepsilon _1$};
\draw(-5, -2.75)--(-5, -2.75)node[left]{\scriptsize$\mu\, \varepsilon _1$};
\draw(-4.255, -5)--(-4.255, -5)node[below]{\scriptsize $(1-\mu)\,\varepsilon _1$};
\draw(-5, -4.255)--(-5, -4.255)node[left]
{\scriptsize$(1-\mu )\,\varepsilon _1$};
\draw(-4.255, -5)--(-4.255, 2.1);
\draw(-5, -4.255)--(2.8, -4.255);
\draw(-2.75, -5)--(-2.75, 2.1)node{};
\draw(-5, -2.75)--(2.8, -2.75);
\end{tikzpicture}
}
\caption{($\protect\varepsilon=3$, $\protect\mu=3/4$) The domain of
integration of the integral $J_{1}$, in grey, is covered by the union of the
domain with points, the domain with vertical lines and that with horizontal
lines. In the part of the grey domain covered by points the function $W$ is
bounded by one. In the part of vertical and horizontal lines the grey domain
covered by vertical lines it is less than or equal to $\protect\sqrt{\protect%
\varepsilon_{4}}/\protect\sqrt{\protect\varepsilon _{3}}$. }
\end{figure}
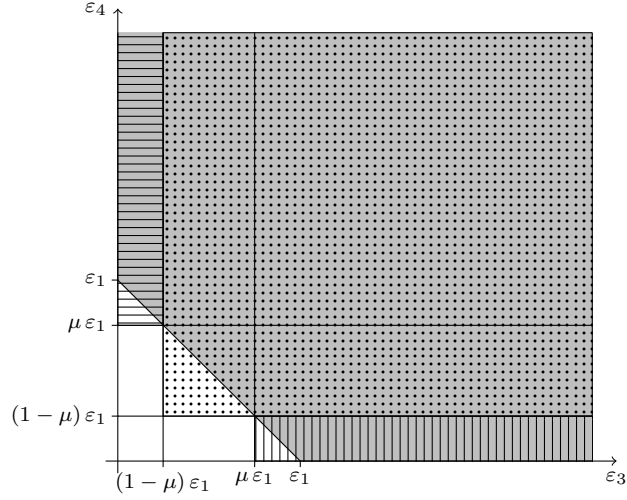

We first estimate $J_{1,1}$: 
\begin{align*}
& J_{1,1}\leq\frac{8\pi^{2}}{\sqrt{2}}\left\Vert \int_{0}^{t}\exp\left( -\pi%
\sqrt{\epsilon_{1}}\int_{s}^{t}M_{0}\left( \xi\right) d\xi\right) \left(
\int_{\left( 1-\mu\right) \epsilon_{1}}^{\infty}f_{3}d\epsilon _{3}\right)
^{2}ds\right\Vert _{L^{\infty}\left( \mathbb{R}^{+};\left( 1+\epsilon\right)
^{\gamma}\right) } \\
& \leq C\left\Vert \frac{1}{1+\left( \epsilon_{1}\right) ^{\frac{3}{2}}}%
\int_{0}^{t}\left( \int_{\left( 1-\mu\right)
\epsilon_{1}}^{\infty}f_{3}d\epsilon_{3}\right) \left( \int_{\left(
1-\mu\right) \epsilon_{1}}^{\infty}\left( 1+\left( \epsilon_{4}\right) ^{%
\frac{3}{2}}\right) f_{4}d\epsilon_{4}\right) ds\right\Vert
_{L^{\infty}\left( \mathbb{R}^{+};\left( 1+\epsilon\right) ^{\gamma}\right) }
\\
& \leq C\left\Vert \int_{0}^{t}\frac{\left\Vert f\left( s,\cdot\right)
\right\Vert _{L^{\infty}\left( \mathbb{R}^{+};\left( 1+\epsilon\right)
^{\gamma}\right) }}{\left( 1+\epsilon_{1}\right) ^{\gamma+\frac{1}{2}}}%
\left( \int_{0}^{\infty}\left( 1+\left( \epsilon_{4}\right) ^{\frac{3}{2}%
}\right) f_{4}d\epsilon_{4}\right) ds\right\Vert _{L^{\infty}\left( \mathbb{R%
}^{+};\left( 1+\epsilon\right) ^{\gamma}\right) } \\
& \leq C\int_{0}^{t}\left\Vert f\left( s,\cdot\right) \right\Vert
_{L^{\infty}\left( \mathbb{R}^{+};\left( 1+\epsilon\right) ^{\gamma }\right)
}\left( \int_{0}^{\infty}\left( 1+\left( \epsilon_{4}\right) ^{\frac{3}{2}%
}\right) f_{4}d\epsilon_{4}\right) ds,
\end{align*}
where the constant $C$ depends on $\mu,\gamma.$

We now estimate $J_{1,2}$ which is the most delicate term. We fix $L>0$ and
treat separately the cases $\epsilon_{1}\leq L$ and $\epsilon_{1}>L.$ In the
region where $\epsilon_{1}\leq L$ we have the pointwise estimate: 
\begin{align*}
& J_{1,2}\leq\frac{16\pi^{2}}{\sqrt{2}}\int_{0}^{t}\exp\left( -\pi \sqrt{%
\epsilon_{1}}\int_{s}^{t}M_{0}\left( \xi\right) d\xi\right) \left(
\int_{\mu\epsilon_{1}}^{\infty}f_{3}d\epsilon_{3}\right) \left( \int
_{0}^{\left( 1-\mu\right) \epsilon_{1}}\sqrt{\frac{\epsilon_{4}}{\epsilon_{1}%
}}f_{4}d\epsilon_{4}\right) ds \\
& \leq\frac{16\pi^{2}}{\sqrt{2}}\sqrt{\left( 1-\mu\right) }%
\int_{0}^{t}\left( \int_{0}^{\infty}f_{3}d\epsilon_{3}\right) ^{2}ds\ \ ,\ \
\epsilon_{1}\leq L.
\end{align*}

On the other hand, if $\epsilon_{1}>L$ we obtain: 
\begin{align*}
& J_{1,2}\leq2\pi\int_{0}^{t}\frac{\exp\left( -\pi\sqrt{\epsilon_{1}}%
\int_{s}^{t}M_{0}\left( \xi\right) d\xi\right) }{\sqrt{\epsilon_{1}}}\left(
\int_{\mu\epsilon_{1}}^{\infty}f_{3}d\epsilon_{3}\right) M_{0}\left(
s\right) ds \\
& \leq\frac{2\pi}{\left( \gamma-1\right) }\left( \sup_{0\leq s\leq
t}\left\Vert f\left( s,\cdot\right) \right\Vert _{L^{\infty}\left( \mathbb{R}%
^{+};\left( 1+\epsilon\right) ^{\gamma}\right) }\right) \int _{0}^{t}\frac{%
\exp\left( -\pi\sqrt{\epsilon_{1}}\int_{s}^{t}M_{0}\left( \xi\right)
d\xi\right) }{\sqrt{\epsilon_{1}}}\frac{M_{0}\left( s\right) ds}{\left(
1+\mu\epsilon_{1}\right) ^{\gamma-1}} \\
& \leq\frac{2}{\left( \gamma-1\right) }\left( \sup_{0\leq s\leq t}\left\Vert
f\left( s,\cdot\right) \right\Vert _{L^{\infty}\left( \mathbb{R}^{+};\left(
1+\epsilon\right) ^{\gamma}\right) }\right) \frac {1}{\epsilon_{1}\left(
1+\mu\epsilon_{1}\right) ^{\gamma-1}},
\end{align*}
where we have used the definition of $M_{0}\left( s\right) $ and (\ref{F3E3a}%
). Using now the inequality $\frac{1}{\epsilon_{1}}\leq\frac {L+1}{L}\frac{1%
}{1+\mu\epsilon_{1}}$ that holds for $\epsilon_{1}\geq L$ we obtain: 
\begin{equation*}
J_{1, 2}\le\frac{2}{\left( \gamma-1\right) }\left( \sup_{0\leq s\leq
t}\left\Vert f\left( s,\cdot\right) \right\Vert _{L^{\infty}\left( \mathbb{R}%
^{+};\left( 1+\epsilon\right) ^{\gamma}\right) }\right) \frac{L+1}{L}\frac{1%
}{\left( 1+\mu\epsilon_{1}\right) ^{\gamma}},\,\,\varepsilon_{1}\ge L. 
\end{equation*}

We now use that $\gamma>3.$ Choosing then $L$ large and $\mu$ sufficiently
close to one (both depending on $\gamma$ we obtain:%
\begin{equation*}
\frac{2}{\left( \gamma-1\right) }\frac{L+1}{L}\frac{1}{\left( 1+\mu
\epsilon_{1}\right) ^{\gamma}}\leq\frac{\theta}{\left( 1+\epsilon
_{1}\right) ^{\gamma}}
\end{equation*}
with $\theta<1$ independent on $\epsilon_{1}.$ Therefore, we obtain, adding
the contributions from the regions where $\epsilon_{1}\leq L$ and $%
\epsilon_{1}>L:$%
\begin{equation}
J_{1,2}\leq\frac{16\pi^{2}}{\sqrt{2}}\sqrt{\left( 1-\mu\right) }\left(
1+L\right) ^{\gamma}\int_{0}^{t}\left( \int_{0}^{\infty}f_{3}d\epsilon
_{3}\right) ^{2}ds+\theta\sup_{0\leq s\leq t}\left\Vert f\left(
s,\cdot\right) \right\Vert _{L^{\infty}\left( \mathbb{R}^{+};\left(
1+\epsilon\right) ^{\gamma}\right) }.   \label{F8E8a}
\end{equation}

We combine now (\ref{F8E3a}), (\ref{F8E4}), (\ref{F8E8a}) to obtain (\ref%
{F8E9}) whence the Lemma follows.
\end{proof}

\begin{proof}[Proof of Proposition \protect\ref{functest}]
It is just a consequence of (\ref{F8E2}), Lemma \ref{IntEst} and Lemma \ref%
{LinfEst}.
\end{proof}

We now prove the following result.

\begin{lemma}
\label{cont}Assume $\gamma>3$. There exists $B_{0}>0$ depending on $%
\left\Vert f_{0}\right\Vert _{L^{\infty}\left( \mathbb{R}^{+};\left(
1+\epsilon\right) ^{\gamma}\right) },$ such that, for any $R>B_{0}$ there
exists $T_{\ast}\left( R\right) $ such that, for any $T<T_{\ast}\left(
R\right) $ the operator $\mathcal{T}$, defined by means of (\ref{F3E8}),
transforms $\mathcal{X}_{R,T}^{\gamma}$ into itself.
\end{lemma}

\begin{proof}
Since $\gamma>3$ we can estimate $\int\!\left( 1+\epsilon^{\frac{3}{2}}\right)
f\left( \epsilon\right) d\epsilon$ by of $C\left\Vert f\left(
s,\cdot\right) \right\Vert _{L^{\infty}\left( \mathbb{R}^{+};\left(
1+\epsilon\right) ^{\gamma}\right) }.$ Using Lemma \ref{LinfEst} we obtain: 
\begin{align*}
\left\Vert \mathcal{T}\left( f\right) \right\Vert _{L^{\infty}\left( \mathbb{%
R}^{+};\left( 1+\epsilon\right) ^{\gamma}\right) }\left( t\right) &
\leq\left\Vert f_{0}\right\Vert _{L^{\infty}\left( \mathbb{R}^{+};\left(
1+\epsilon\right) ^{\gamma}\right) }+Ct\left( \sup_{0\leq s\leq t}\left\Vert
f\left( s,\cdot\right) \right\Vert _{L^{\infty}\left( \mathbb{R}^{+};\left(
1+\epsilon\right) ^{\gamma}\right) }\right) ^{3}+ \\
&\hskip -2.4cm  +Ct\!\left( 1+\sup_{0\leq s\leq t}\left\Vert f\left( s,\cdot\right)
\right\Vert _{L^{\infty}\left( \mathbb{R}^{+};\left( 1+\epsilon\right)
^{\gamma}\right) }\right) \! \left( \sup_{0\leq s\leq t}\left\Vert f\left(
s,\cdot\right) \right\Vert _{L^{\infty}\left( \mathbb{R}^{+};\left(
1+\epsilon\right) ^{\gamma}\right) }\right) ^{2}\!+ \\
&\hskip -2.4cm +Ct\left( \sup_{0\leq s\leq t}\left\Vert f\left( s,\cdot\right)
\right\Vert _{L^{\infty}\left( \mathbb{R}^{+};\left( 1+\epsilon\right)
^{\gamma}\right) }\right) ^{2}+\theta\sup_{0\leq s\leq t}\left\Vert f\left(
s,\cdot\right) \right\Vert _{L^{\infty}\left( \mathbb{R}^{+};\left(
1+\epsilon\right) ^{\gamma}\right) },
\end{align*}
if $0\leq t\leq T.$ Therefore, since $\left\Vert f\right\Vert _{L^{\infty
}\left( \left[ 0,T\right] ;L^{\infty}\left( \mathbb{R}^{+};\left(
1+\epsilon\right) ^{\gamma}\right) \right) }\leq R$ we obtain:%
\begin{equation}
\sup_{0\leq t\leq T}\left\Vert \mathcal{T}\left( f\right) \right\Vert
_{L^{\infty}\left( \mathbb{R}^{+};\left( 1+\epsilon\right) ^{\gamma }\right)
}\left( t\right) \leq\left\Vert f_{0}\right\Vert _{L^{\infty }\left( \mathbb{%
R}^{+};\left( 1+\epsilon\right) ^{\gamma}\right) }+CT\left(
R^{2}+R^{3}\right) +\theta R.   \label{F9E1}
\end{equation}

Since $\theta<1$ it then follows that, choosing $R$ sufficiently large and
then $T$ small (depending on $B$) we would obtain that the right-hand side
of (\ref{F9E1}) is smaller than $R,$ whence the result follows.
\end{proof}

We also have the following:

\begin{lemma}
\label{weakC}Suppose that $R\geq B_{0}$ and $T<T_{\ast}\left( R\right) $ are
as in Lemma \ref{cont}. Then, the operator $\mathcal{T}$ defined by means of
(\ref{F3E8}) is continuous in the metric space $\mathcal{X}_{R,T}^{\gamma}$.
\end{lemma}

\begin{proof}
Since $\mathcal{X}_{R,T}^{\gamma}$ is a metric space, it is sufficient to
check the result for sequences. Suppose then that we have a sequence $%
\{f_{n}\}_{n\ge0}\in\mathcal{X}_{R,T}^{\gamma}$ and $\mu\in\mathcal{X}%
_{R,T}^{\gamma}$ such that $f_{n}\rightarrow\mu$ in the topology of $%
\mathcal{X}_{R,T}^{\gamma}$. Let $\varphi\in C_{0}\left( \left[
0,\infty\right) \right) $ be a test function. We need to show that we can
pass to the limit in:%
\begin{equation*}
\int_{\mathbb{R}^{+}}\mathcal{T}\left( f_{n}\right) \left( t,\epsilon
_{1}\right) \varphi\left( \epsilon_{1}\right) d\epsilon_{1}
\end{equation*}
uniformly in $t\in\left[ 0,T\right] .$ We notice first that, due to the
boundedness of the functions in $\mathcal{X}_{R,T}^{\gamma},$ weak
convergence of a sequence $\left\{ f_{n}\right\} $ in the topology $\Theta$
to $\mu \in\mathcal{X}_{R,T}^{\gamma}$ implies the convergence of integrals
like $\int_{\left[ R_{1},R_{2}\right] }f_{n}\left( t,\epsilon\right)
d\epsilon$ to $\int_{\left[ R_{1},R_{2}\right] }\mu d\epsilon$ for any $%
0\leq R_{1}\leq R_{2}<\infty.$ Indeed, this can be seen approximating the
characteristic function of the interval $\left[ R_{1},R_{2}\right] $ by a
set of continuous functions $\varphi_{m}\in C_{0}\left( \left[ 0,T\right]
\times\left[ 0,\infty\right) \right) $ in the topology of $L^{1}\left( \left[
0,T\right] \times\left[ 0,\infty\right) \right) .$ The $L^{\infty}$
estimates for the functions $f_{n}\in\mathcal{K}_{T}$ imply that, for any
given $\varepsilon>0$ choosing $m$ large enough we have $\left\vert \int_{%
\left[ R_{1},R_{2}\right] }f_{n}d\epsilon-\int f_{n}\varphi
_{m}d\epsilon\right\vert <\frac{\varepsilon}{2}$ uniformly in $t\in\left[ 0,T%
\right] .$ Choosing then $n$ large enough we obtain $\left\vert \int_{\left[
R_{1},R_{2}\right] }\mu\varphi_{m}d\epsilon-\int
f_{n}\varphi_{m}d\epsilon\right\vert <\frac{\varepsilon}{2},$ whence the
desired convergence follows.

We consider first the term $f_{0}\left( \epsilon_{1}\right) \Psi\left(
t,\epsilon_{1}\right) .$ The function $f_{0}\left( \epsilon_{1}\right) $ is
fixed, independent on $n.$ We need to compute the pointwise limit as $%
n\rightarrow\infty$ of%
\begin{align*}
a_{n}\left( t,\epsilon_{1}\right) = &\frac{8\pi^{2}}{\sqrt{2}}%
\int_{0}^{\infty}\int_{0}^{\infty}f_{n,2}\left( 1+f_{n,3}+f_{n,4}\right)
Wd\epsilon_{3}d\epsilon_{4}\\
 \Psi_{n}\left( t,\epsilon_{1}\right)
= & \exp\left( -\int_{0}^{t}a_{n}\left( s,\epsilon_{1}\right) ds\right) . 
\end{align*}

Since $\varphi$ is compactly supported we need to compute this limit only in
bounded regions. Suppose that $\epsilon_{1}\leq L_{1}.$ We then rewrite $%
a_{n}\left( t,\epsilon_{1}\right) $ as:%
\begin{align}
a_{n}\left( t,\epsilon_{1}\right) = & \frac{8\pi^{2}}{\sqrt{2}}\int \int_{%
\left[ 0,L_{2}\right] ^{2}}f_{n,2}\left( 1+f_{n,3}+f_{n,4}\right)
Wd\epsilon_{3}d\epsilon_{4}+  \notag \\
& +\frac{8\pi^{2}}{\sqrt{2}}\int\int_{\mathbb{R}_{+}^{2}\setminus\left[
0,L_{2}\right] ^{2}}f_{n,2}\left( 1+f_{n,3}+f_{n,4}\right) Wd\epsilon
_{3}d\epsilon_{4},\   \label{F4E5}
\end{align}
where $L_{2}>2L_{1}$. In order to pass to the limit in the first integral on
the right-hand side we change variables to have as integration variables $%
\left( \epsilon_{2},\epsilon_{3}\right) $ for the integral containing the
term $f_{n,2}f_{n,3}$ and the integration variables $\left( \epsilon
_{2},\epsilon_{4}\right) $ for the integral containing the term $%
f_{n,2}f_{n,4}.$ In the case of the term containing only $f_{n,2}$ either
choice of variables of integration is valid. We then obtain integrals in
sets with test functions chosen as characteristic functions multiplying a
continuous function which contains the dependence on $W.$ Due to the
estimates for the functions $f_{n}\in\mathcal{K}_{T}$ we can pass to the
limit in these integrals as $n\rightarrow\infty$. The resulting limit is,
after returning to the original integration variables: 
\begin{equation*}
\frac{8\pi^{2}}{\sqrt{2}}\int\int_{\left[ 0,L_{2}\right] ^{2}}\mu_{2}\left(
1+\mu_{3}+\mu_{4}\right) Wd\epsilon_{3}d\epsilon_{4}. 
\end{equation*}

The last integral in (\ref{F4E5}) can be estimated using $\epsilon_{2}$ as
one of the integration variables and the fact that $L_{2}>2L_{1}:$%
\begin{equation*}
\frac{8\pi^{2}}{\sqrt{2}}\int\int_{\mathbb{R}_{+}^{2}\setminus\left[ 0,L_{2}%
\right] ^{2}}f_{n,2}\left( 1+f_{n,3}+f_{n,4}\right) Wd\epsilon
_{3}d\epsilon_{4}\leq C\int_{\frac{L_{2}}{2}}^{\infty}f_{n}\left(
\epsilon\right) \epsilon d\epsilon. 
\end{equation*}

Using then the fact that $\int_{0}^{\infty}f_{n}\left( t,\epsilon\right)
\epsilon^{\frac{3}{2}}d\epsilon\leq2\left\Vert f_{0}\right\Vert _{\mathcal{Y}%
}$ we obtain:%
\begin{equation*}
\frac{8\pi^{2}}{\sqrt{2}}\int\int_{\mathbb{R}_{+}^{2}\setminus\left[ 0,L_{2}%
\right] ^{2}}f_{n,2}\left( 1+f_{n,3}+f_{n,4}\right) Wd\epsilon
_{3}d\epsilon_{4}\leq\frac{C}{\sqrt{L_{2}}}, 
\end{equation*}
where $C$ is independent on $n.$ Taking then the limit $L_{2}\rightarrow
\infty$ and then $n\rightarrow\infty$ we obtain that:%
\begin{equation}
a_{n}\left( t,\epsilon_{1}\right) \rightarrow a\left( t,\epsilon _{1}\right)
=\frac{8\pi^{2}}{\sqrt{2}}\int\int\mu_{2}\left( 1+\mu_{3}+\mu_{4}\right)
Wd\epsilon_{3}d\epsilon_{4},\   \label{F4E6}
\end{equation}
as $n\rightarrow\infty$ for each $\epsilon_{1}\geq0$ uniformly in $t\in\left[
0,T\right] .$ Since $a_{n}\left( t,\epsilon_{1}\right) \geq0$ we have $%
\Psi_{n}\left( t,\epsilon_{1}\right) \leq1.$ We can use then Lebesgue's
Theorem to obtain:%
\begin{equation}
\int_{\mathbb{R}^{+}}f_{0}\left( \epsilon_{1}\right) \Psi_{n}\left(
t,\epsilon_{1}\right) \varphi\left( t,\epsilon_{1}\right) d\epsilon
_{1}dt\rightarrow\int_{\mathbb{R}^{+}}f_{0}\left( \epsilon_{1}\right)
\Psi\left( t,\epsilon_{1}\right) \varphi\left( t,\epsilon_{1}\right)
d\epsilon_{1}dt,\   \label{F4E7}
\end{equation}
as $n\rightarrow\infty,$ uniformly in $t\in\left[ 0,T\right] $ with:%
\begin{equation*}
\Psi\left( t,\epsilon_{1}\right) =\exp\left( -\int_{0}^{t}a\left(
s,\epsilon_{1}\right) ds\right) , 
\end{equation*}
and $a\left( t,\epsilon_{1}\right) $ as in (\ref{F4E6}).

We now pass to the limit in the term resulting from the last term in (\ref%
{F3E8}). To this end we compute the limits of the following integrals:%
\begin{align*}
I_{1,n} & =\frac{8\pi^{2}}{\sqrt{2}}\int_{0}^{t}\int_{0}^{\infty}\int
_{0}^{\infty}\int_{0}^{\infty}\frac{\Psi_{n}\left( t,\epsilon_{1}\right) }{%
\Psi_{n}\left( s,\epsilon_{1}\right) }f_{n,3}f_{n,4}W\varphi\left(
t,\epsilon_{1}\right) d\epsilon_{1}d\epsilon_{3}d\epsilon_{4}ds, \\
I_{2,n} & =\frac{8\pi^{2}}{\sqrt{2}}\int_{0}^{t}\int_{0}^{\infty}\int
_{0}^{\infty}\int_{0}^{\infty}\frac{\Psi_{n}\left( t,\epsilon_{1}\right) }{%
\Psi_{n}\left( s,\epsilon_{1}\right) }f_{n,1}f_{n,3}f_{n,4}W\varphi\left(
t,\epsilon_{1}\right) d\epsilon_{1}d\epsilon_{3}d\epsilon_{4}ds, \\
I_{3,n} & =\frac{8\pi^{2}}{\sqrt{2}}\int_{0}^{t}\int_{0}^{\infty}\int
_{0}^{\infty}\int_{0}^{\infty}\frac{\Psi_{n}\left( t,\epsilon_{1}\right) }{%
\Psi_{n}\left( s,\epsilon_{1}\right) }f_{n,2}f_{n,3}f_{n,4}W\varphi\left(
t,\epsilon_{1}\right) d\epsilon_{1}d\epsilon_{3}d\epsilon_{4}ds.
\end{align*}

The limit of the first integral, $\lim_{n\rightarrow\infty}I_{1,n}$, is
obtained in the same way as (\ref{F4E7}). We split the integral $I_{1,n}$ in
the variables $\left( \epsilon_{3},\epsilon_{4}\right) $ in the regions $%
\left[ 0,L\right] ^{2}$ and $\mathbb{R}_{+}^{2}\setminus\left[ 0,L_{2}\right]
^{2}$ with $L_{2}$ large to be determined. On the other hand, the
integration in $\epsilon_{1}$ takes place in a compact set due to the fact
that $\varphi$ is compactly supported. The contribution to the integral due
to the set $\left( \epsilon_{3},\epsilon_{4}\right) \in\mathbb{R}%
_{+}^{2}\setminus\left[ 0,L_{2}\right] ^{2}$ is uniformly small if $L_{2}$
is large (independently on $n$) arguing as before. We can take the limit of $%
\frac{\Psi_{n}\left( t,\epsilon_{1}\right) }{\Psi_{n}\left( s,\epsilon
_{1}\right) }$ as $n\rightarrow\infty$ using Lebesgue's Theorem, whence:%
\begin{equation*}
I_{1,n}\rightarrow\frac{8\pi^{2}}{\sqrt{2}}\int_{0}^{t}\int_{0}^{\infty}%
\int_{0}^{\infty}\int_{0}^{\infty}\frac{\Psi\left( t,\epsilon_{1}\right) }{%
\Psi\left( s,\epsilon_{1}\right) }\mu_{3}\mu_{4}W\varphi\left(
t,\epsilon_{1}\right) d\epsilon_{1}d\epsilon_{3}d\epsilon_{4}ds\, ,\ \
n\rightarrow\infty. 
\end{equation*}

In order to take the limit of $I_{2,n}$ we first use the fact that the
integration takes place in a bounded interval $\left[ 0,L_{1}\right] $ due
to the choice of $\varphi.$ We use Egorov's Theorem to approximate
uniformly, for each $\varepsilon_{0}$ small, $\frac{\Psi_{n}\left(
t,\epsilon _{1}\right) }{\Psi_{n}\left( s,\epsilon_{1}\right) }$ as $\frac{%
\Psi\left( t,\epsilon_{1}\right) }{\Psi\left( s,\epsilon_{1}\right) }$ in a
set $A\subset\left[ 0,L_{1}\right] $ such that $\left\vert \left[ 0,L_{1}%
\right] \setminus A\right\vert \leq\varepsilon_{0}.$ Notice that the set $A$
depends on $t$ and $s.$ We then have 
\begin{equation*}
\left\vert \int_{0}^{L_{1}}\frac{\Psi_{n}\left( t,\epsilon_{1}\right) }{%
\Psi_{n}\left( s,\epsilon_{1}\right) }f_{n,1}Wd\epsilon_{1}-\int_{A}\frac{%
\Psi\left( t,\epsilon_{1}\right) }{\Psi\left( s,\epsilon_{1}\right) }%
f_{n,1}Wd\epsilon_{1}\right\vert \leq C\left\Vert f\right\Vert _{\mathcal{X}%
_{R,T}^{\gamma}}\varepsilon_{0}, 
\end{equation*}
uniformly in bounded sets of $\epsilon_{3},\epsilon_{4}$ in bounded sets if $%
n$ is chosen sufficiently large. It then follows that:%
\begin{equation*}
I_{2,n}\rightarrow\frac{8\pi^{2}}{\sqrt{2}}\int_{0}^{t}\int_{0}^{\infty}%
\int_{0}^{\infty}\int_{0}^{\infty}\frac{\Psi\left( t,\epsilon_{1}\right) }{%
\Psi\left( s,\epsilon_{1}\right) }\mu_{1}\mu_{3}\mu_{4}W\varphi\left(
t,\epsilon_{1}\right) d\epsilon_{1}d\epsilon_{3}d\epsilon_{4}ds\, ,\ \
n\rightarrow\infty. 
\end{equation*}

We now consider the limit of the term $I_{3,n}.$ We can again restrict the
domain of integration to a large cube $\left[ 0,L\right] ^{3}$ because the
contribution of the tails can be estimated uniformly in $n$ as $L\rightarrow
\infty.$ On the other hand we can apply again Egorov's Theorem to show that
for any $\varepsilon_{0}>0$ there exists a set $A\subset\left[ 0,L\right] $,
depending on $s$ and $t$, such that $\left\vert \left[ 0,L\right] \setminus
A\right\vert \leq\varepsilon_{0}$ with the property that $\frac{\Psi
_{n}\left( t,\epsilon_{1}\right) }{\Psi_{n}\left( s,\epsilon_{1}\right) }$
converges uniformly to $\frac{\Psi\left( t,\epsilon_{1}\right) }{\Psi\left(
s,\epsilon_{1}\right) }$ uniformly on $A$. We now change the integration
variables by means of:%
\begin{equation*}
\left( \epsilon_{1},\epsilon_{3},\epsilon_{4}\right) \rightarrow\left(
\epsilon_{2},\epsilon_{3},\epsilon_{4}\right) =\left(
\epsilon_{3}+\epsilon_{4}-\epsilon_{1},\epsilon_{3},\epsilon_{4}\right) . 
\end{equation*}
This transformation brings the set $\left( \left[ 0,L\right] \setminus
A\right) \times\left[ 0,L\right] ^{2}$ to a new set $B$ whose intersection
with the domain of integration $\Omega_{L}$ has a measure of order $%
C_{L}\varepsilon_{0},$ with $C_{L}$ depending on $L,$ but not on $%
\varepsilon_{0},n.$ We then need to take the limit as $n\rightarrow\infty$
in:%
\begin{align*}
& \frac{8\pi^{2}}{\sqrt{2}}\int_{0}^{t}\int_{0}^{L}\int_{0}^{L}\int_{0}^{L}%
\frac{\Psi_{n}\left( t,\epsilon_{3}+\epsilon_{4}-\epsilon_{2}\right) }{%
\Psi_{n}\left( s,\epsilon_{3}+\epsilon_{4}-\epsilon_{2}\right) }%
f_{n,2}f_{n,3}f_{n,4}W\varphi\left( t,\epsilon_{1}\right) d\epsilon
_{2}d\epsilon_{3}d\epsilon_{4}ds \\
& =\frac{8\pi^{2}}{\sqrt{2}}\int_{0}^{t}\int\int\int_{\Omega_{L}\setminus B}%
\left[ \cdot\cdot\cdot\right] +\frac{8\pi^{2}}{\sqrt{2}}\int_{0}^{t}\int\int%
\int_{B}\left[ \cdot\cdot\cdot\right] .
\end{align*}

The last integral can be estimated as:%
\begin{equation*}
\frac{8\pi^{2}}{\sqrt{2}}\int_{0}^{t}\int\int\int_{B}\left[ \cdot\cdot \cdot%
\right] \leq C_{L}\varepsilon_{0}. 
\end{equation*}

On the other hand, in the integral $\frac{8\pi^{2}}{\sqrt{2}}%
\int_{0}^{t}\int\int\int_{\Omega_{L}\setminus B}\left[ \cdot\cdot\cdot\right]
$ we have uniform convergence of $\frac{\Psi_{n}\left( t,\cdot\right) }{\Psi
_{n}\left( s,\cdot\right) }$ to $\frac{\Psi\left( t,\cdot\right) }{%
\Psi\left( s,\cdot\right) }$ as $n\rightarrow\infty.$ We can then complete
the integral to the domain $\Omega_{L},$ adding a term which gives an error
of order $C_{L}\varepsilon_{0}.$ Making them $\varepsilon_{0}$ small and
then $L$ large, and returning to the original set of variables we obtain:%
\begin{equation*}
I_{3,n}\rightarrow\frac{8\pi^{2}}{\sqrt{2}}\int_{0}^{t}\int_{0}^{\infty}%
\int_{0}^{\infty}\int_{0}^{\infty}\frac{\Psi\left( t,\epsilon_{1}\right) }{%
\Psi\left( s,\epsilon_{1}\right) }\mu_{2}\mu_{3}\mu_{4}W\varphi\left(
t,\epsilon_{1}\right) d\epsilon_{1}d\epsilon_{3}d\epsilon_{4}ds\, ,\ \
n\rightarrow\infty. 
\end{equation*}

This concludes the proof of the continuity of $\mathcal{T}$ in the topology $%
\Theta$.
\end{proof}

\begin{lemma}
\label{Tcomp}Suppose that $R\geq B_{0}$ and $T<T_{\ast}\left( R\right) $ are
as in Lemma \ref{cont}. The operator $\mathcal{T}$, defined by means of (\ref%
{F3E8}) and restricted to the metric space $\mathcal{X}_{R,T}^{\gamma}$ is
compact.
\end{lemma}

\begin{proof}
By Proposition \ref{compCrit} it is enough to show that the functions $%
\psi_{\varphi}\left( t;\mathcal{T}\left( f\right) \right) =\int _{\mathbb{R}%
^{+}}\!f\left( t,\epsilon\right) \varphi\left( \epsilon\right) d\epsilon$
are uniformly continuous in $t\in\left[ 0,T\right] $ for $f\in\mathcal{X}%
_{R,T}^{\gamma}.$ Therefore, we need to prove that the functions 
\begin{equation*}
\Psi_{\varphi}^{\left( 1\right) }\left( t;f\right) =\int_{\mathbb{R}%
^{+}}f_{0}\left( \epsilon_{1}\right) \Psi\left( t,\epsilon_{1}\right)
\varphi\left( \epsilon_{1}\right) d\epsilon_{1}, 
\end{equation*}%
\begin{equation*}
\Psi_{\varphi}^{\left( 2\right) }\left( t;f\right) =\frac{8\pi^{2}}{\sqrt{2}}%
\int_{0}^{t}\frac{\Psi\left( t,\epsilon_{1}\right) }{\Psi\left(
s,\epsilon_{1}\right) }\int_{0}^{\infty}\int_{0}^{\infty}\int_{0}^{\infty
}f_{3}f_{4}\left( 1+f_{1}+f_{2}\right) \varphi\left( \epsilon_{1}\right)
Wd\epsilon_{1}d\epsilon_{3}d\epsilon_{4}ds, 
\end{equation*}
are uniformly continuous for $f\in\mathcal{X}_{R,T}^{\gamma},\ t\in\left[ 0,T%
\right] $. We prove first that the family of functions $\left\{
\Psi_{\varphi}^{\left( 1\right) }\left( \cdot;f\right) :f\in \mathcal{X}%
_{R,T}^{\gamma}\right\} $ is uniformly Lischitz in $t\in\left[ 0,T\right] .$
To this end, we just differentiate these functions with respect to $t.$ This
can be made $a.e.$ $t\in\left[ 0,T\right] $ due to the uniform boundedness
of $a\left( t,\epsilon_{1}\right) $ in compact sets of $\epsilon_{1}\in\left[
0,\infty\right) .$ Then:%
\begin{align*}
&\frac{\partial\Psi_{\varphi}^{\left( 1\right) }\left( t;f\right) }{\partial t%
}  = -\int_{\mathbb{R}^{+}}f_{0}\left( \epsilon_{1}\right) a\left(
t,\epsilon_{1}\right) \Psi\left( t,\epsilon_{1}\right) \varphi\left(
\epsilon_{1}\right) d\epsilon_{1} \\
&\frac{\partial\Psi_{\varphi}^{\left( 2\right) }\left( t;f\right) }{\partial t%
}  = \frac{8\pi^{2}}{\sqrt{2}}\int_{0}^{\infty}\int_{0}^{\infty}\int_{0}^{%
\infty}f_{3}f_{4}\left( 1+f_{1}+f_{2}\right) \varphi\left(
\epsilon_{1}\right) Wd\epsilon_{1}d\epsilon_{3}d\epsilon _{4}- \\
& -\frac{8\pi^{2}}{\sqrt{2}}\int_{0}^{t}\frac{\Psi\left( t,\epsilon
_{1}\right) a\left( t,\epsilon_{1}\right) }{\Psi\left( s,\epsilon
_{1}\right) }\int_{0}^{\infty}\int_{0}^{\infty}\int_{0}^{\infty}f_{3}f_{4}%
\left( 1+f_{1}+f_{2}\right) \varphi\left( \epsilon_{1}\right)
Wd\epsilon_{1}d\epsilon_{3}d\epsilon_{4}ds\ 
\end{align*}
$a.e.\ t\in\left[ 0,T\right] .$ The right-hand side of these formulas can be
bounded easily using that $f\in\mathcal{X}_{R,T}^{\gamma}.$ It then follows
that the family of functions $\mathcal{T}\left( \mathcal{X}_{R,T}^{\gamma
}\right) $ is equicontinuous and then Proposition \ref{AsAr} implies that it
is compact.
\end{proof}

\begin{proposition}
\label{LE}Let $f_{0}\in L^{\infty}\left( \mathbb{R}^{+};\left(
1+\epsilon\right) ^{\gamma}\right) $ with $\gamma>3.$ There exists $T>0$
depending only on $\left\Vert f_{0}\left( \cdot\right) \right\Vert
_{L^{\infty}\left( \mathbb{R}^{+};\left( 1+\epsilon\right) ^{\gamma }\right)
}$ and at least one mild solution of (\ref{F3E2}), (\ref{F3E3}) in the sense
of Definition \ref{mild}.
\end{proposition}

\begin{proof}
Given $\left\Vert f_{0}\left( \cdot\right) \right\Vert _{L^{\infty}\left( 
\mathbb{R}^{+};\left( 1+\epsilon\right) ^{\gamma}\right) }$ we choose $R\geq
B_{0}$ and $T<T_{\ast}\left( R\right) $ are as in Lemma \ref{cont}. Then the
operator $\mathcal{T}$ transforms the space $\mathcal{X}_{R,T}^{\gamma}$
into itself due to Lemma \ref{cont}. Moreover, this operator is continuous
due to Lemma \ref{weakC} if we endow $\mathcal{X}_{R,T}^{\gamma }$ with the
topology $\Theta$ and due to Lemmas \ref{topol}, \ref{compCrit} and \ref%
{Tcomp} the operator $\mathcal{T}:\mathcal{X}_{R,T}^{\gamma }\rightarrow%
\mathcal{X}_{R,T}^{\gamma}$ is compact. Then Schauder-Tikhonov Theorem (cf. 
\cite{DS}, p. 456) implies the existence of a fixed point of $\mathcal{T}$
in $f\in\mathcal{X}_{R,T}^{\gamma}$. Due to the definition of $\mathcal{T}$\
it follows that $f$ satisfies (\ref{F3E6}).
\end{proof}

Proposition \ref{LE} yields a solution of (\ref{F3E2}), (\ref{F3E3}) in the
sense of Definition \ref{mild}. In order to check that the total number of
particles and the total energy are constant in time it is convenient to
prove that the derived solution is a weak solution in some suitable sense.

\begin{lemma}
\label{der17}Suppose that $\gamma>3$ and $f\in L_{loc}^{\infty}\left( \left[
0,T\right) ;L^{\infty}\left( \mathbb{R}^{+};\left( 1+\epsilon\right)
^{\gamma}\right) \right) $ is a mild solution of (\ref{F3E2}), (\ref{F3E3})
in the sense of Definition \ref{mild}. Let $g$ be as in (\ref{F3E3a}).
Suppose that $\varphi\in C_{0}^{2}\left( \left[ 0,T\right) \!\!\times\!\!%
\left[ 0,\infty\right) \right) $ Then, the function $\psi_{\varphi}\left(
t\right) =\int_{\mathbb{R}^{+}}g\left( t,\epsilon\right) \varphi\left(
t,\epsilon\right) d\epsilon$ is Lipschitz continuous in $t\in\left[ 0,T%
\right] $ and the following identity holds:%
\begin{align}
\partial_{t}\left( \int_{\mathbb{R}^{+}}g\varphi d\epsilon\right) & =\int_{%
\mathbb{R}^{+}}g\partial_{t}\varphi d\epsilon+\frac{1}{2^{\frac{5}{2}}}\int_{%
\mathbb{R}^{+}}\int_{\mathbb{R}^{+}}\int_{\mathbb{R}^{+}}\frac {%
g_{1}g_{2}g_{3}\Phi}{\sqrt{\epsilon_{1}\epsilon_{2}\epsilon_{3}}}Q_{\varphi
}d\epsilon_{1}d\epsilon_{2}d\epsilon_{3}+ & &  \notag \\
& \hskip2cm+\frac{\pi}{2}\int_{\mathbb{R}^{+}}\int_{\mathbb{R}^{+}}\int_{%
\mathbb{R}^{+}}\frac{g_{1}g_{2}\Phi}{\sqrt{\epsilon_{1}\epsilon_{2}}}%
Q_{\varphi}d\epsilon_{1}d\epsilon_{2}d\epsilon_{3}\,,\ \ a.e.\ t & & \in
\left[ 0,T\right] ,   \label{F5E1a}
\end{align}
where:%
\begin{equation}
\Phi=\min\left\{ \sqrt{\epsilon_{1}},\sqrt{\epsilon_{2}},\sqrt{\epsilon_{3}},%
\sqrt{\left( \epsilon_{1}+\epsilon_{2}-\epsilon_{3}\right) _{+}}\right\} \ 
\label{F5E1b}
\end{equation}%
\begin{equation}
Q_{\varphi}=\varphi\left( \epsilon_{3}\right) +\varphi\left( \epsilon
_{1}+\epsilon_{2}-\epsilon_{3}\right) -2\varphi\left( \epsilon_{1}\right) . 
\label{F5E1c}
\end{equation}
\end{lemma}

\begin{proof}
Using (\ref{F3E6}) we obtain that $g$ satisfies:%
\begin{align*}
g\left( t,\epsilon_{1}\right)& =g_{0}\left( \epsilon_{1}\right) \Psi\left(
t,\epsilon_{1}\right)+\\ 
&+\pi\int_{0}^{t}\frac{\Psi\left( t,\epsilon
_{1}\right) }{\Psi\left( s,\epsilon_{1}\right) }\int_{0}^{\infty}\int
_{0}^{\infty}\frac{g_{3}g_{4}}{\sqrt{\epsilon_{3}\epsilon_{4}}}\left( 1+%
\frac{g_{1}}{4\pi\sqrt{2\epsilon_{1}}}+\frac{g_{2}}{4\pi\sqrt{2\epsilon_{2}}}%
\right) \Phi d\epsilon_{3}d\epsilon_{4}ds 
\end{align*}
a.e. $t\in\left[ 0,T\right) $. Multiplying by a test function $\varphi\in
C_{0}\left( \left[ 0,T\right) \times\left[ 0,\infty\right) \right) $ and
integrating in $\epsilon\in\left[ 0,\infty\right) $ we obtain:%
\begin{align*}
& \int_{0}^{\infty}g\left( t,\epsilon_{1}\right) \varphi\left(
t,\epsilon_{1}\right) d\epsilon_{1} =\int_{0}^{\infty}g_{0}\left(
\epsilon_{1}\right) \varphi\left( t,\epsilon_{1}\right) \Psi\left(
t,\epsilon_{1}\right) d\epsilon_{1}+ \\
& +\pi\int_{0}^{t}\int_{0}^{\infty}\int_{0}^{\infty}\int_{0}^{\infty}\frac{%
g_{3}g_{4}\Phi}{\sqrt{\epsilon_{3}\epsilon_{4}}}\left( 1+\frac{g_{1}}{4\pi%
\sqrt{2\epsilon_{1}}}+\frac{g_{2}}{4\pi\sqrt{2\epsilon_{2}}}\right)
\varphi\left( t,\epsilon_{1}\right) \frac{\Psi\left( t,\epsilon_{1}\right) }{%
\Psi\left( s,\epsilon_{1}\right) }d\epsilon_{1}d\epsilon_{3}d\epsilon _{4}ds.
\end{align*}

Notice that the integral $\psi_{\varphi}\left( t\right) =\int_{0}^{\infty
}g\left( t,\epsilon_{1}\right) \varphi\left( t,\epsilon_{1}\right)
d\epsilon_{1}$ is Lipschitz continuous with respect to $t$ in $t\in\left[
0,T\right) ,$ since $a\left( t,\epsilon_{1}\right) $ in (\ref{F3E6a}) is
uniformly bounded in compact sets of $\left[ 0,T\right) \times\left[
0,\infty\right) .$ Therefore $\psi_{\varphi}\left( t\right) $ is
differentiable $a.e.$ $t\in\left[ 0,T\right) $ and its derivative is given
by:%
\begin{align*}
&\partial_{t}\left( \int_{0}^{\infty}g\left( t,\epsilon_{1}\right)
\varphi\left( t,\epsilon_{1}\right) d\epsilon_{1}\right) 
=\int_{0}^{\infty}g\left( t,\epsilon_{1}\right) \partial_{t}\varphi\left(
t,\epsilon_{1}\right) d\epsilon_{1}+  \notag \\
&+\pi\int_{0}^{\infty}\!\! \int_{0}^{\infty}\!\!\int_{0}^{\infty}\!\!\frac{%
g_{3}g_{4}\Phi}{\sqrt{\epsilon _{3}\epsilon_{4}}}\left( 1+\frac{g_{1}}{4\pi%
\sqrt{2\epsilon_{1}}}+\frac {g_{2}}{4\pi\sqrt{2\epsilon_{2}}}\right)
\varphi\left( t,\epsilon _{1}\right) d\epsilon_{1}d\epsilon_{3}d\epsilon_{4}-
\\
&\hskip 7cm  -\int_{0}^{\infty}\varphi\left( t,\epsilon_{1}\right) a\left(
t,\epsilon_{1}\right) g\left( t,\epsilon_{1}\right) d\epsilon_{1}.
\end{align*}

Using (\ref{F3E6a}) we obtain:%
\begin{align*}
&\partial_{t}\left( \int_{0}^{\infty}g\left( t,\epsilon_{1}\right)
\varphi\left( t,\epsilon_{1}\right) d\epsilon_{1}\right) 
=\int_{0}^{\infty}g\left( t,\epsilon_{1}\right) \partial_{t}\varphi\left(
t,\epsilon_{1}\right) d\epsilon_{1}+  \notag \\
&+\pi\int_{0}^{\infty}\!\! \int_{0}^{\infty}\!\!\int_{0}^{\infty}\!\!\frac{%
g_{3}g_{4}\Phi}{\sqrt{\epsilon _{3}\epsilon_{4}}}\left( 1+\frac{g_{1}}{4\pi%
\sqrt{2\epsilon_{1}}}+\frac {g_{2}}{4\pi\sqrt{2\epsilon_{2}}}\right)
\varphi\left( t,\epsilon _{1}\right) d\epsilon_{1}d\epsilon_{3}d\epsilon_{4}-
\\
&\hskip 3.5cm  -\pi\int_{0}^{\infty}\!\! \int_{0}^{\infty}\!\!\int_{0}^{\infty}\!\!\frac{%
g_{1}g_{2}\Phi}{\sqrt{\epsilon_{1}\epsilon_{2}}}\left( 1+\frac{g_{3}}{4\pi 
\sqrt{2\epsilon_{3}}}+\frac{g_{4}}{4\pi\sqrt{2\epsilon_{4}}}\right)
\varphi\left( t,\epsilon_{1}\right) d\epsilon_{1}.
\end{align*}

In order to obtain (\ref{F5E1a}) we perform the two following simple
operations. We relabel the integration variables in the cubic terms in order
to write them as integrals with respect to the variables $\left( \epsilon
_{1},\epsilon_{2},\epsilon_{3}\right) $. In the quadratic terms we
symmetrize the integrals with respect to variables that appear in the
functions $g$.
\end{proof}

\begin{remark} Lemma \ref{der17} shows that any function $f\in L_{loc}^{\infty}\left(
\left[ 0,T\right);L^{\infty}\left( \mathbb{R}^{+};\left( 1+\epsilon \right)
^{\gamma}\right) \right) $ with $\gamma >3$ that is a mild solution of (\ref{F3E2}), (\ref{F3E3})
in the sense of Definition \ref{mild} is also a weak solution of  (\ref{F3E2}), (\ref{F3E3}) on $(0, T)$ in the sense of definition \ref{weakf}.
\end{remark}

We can now prove that mass and energy are preserved for the obtained
solution.

\begin{lemma}
\label{massenergy}Suppose that $\gamma>3$ and $f\in L_{loc}^{\infty}\left( %
\left[ 0,T\right) ;L^{\infty}\left( \mathbb{R}^{+};\left( 1+\epsilon \right)
^{\gamma}\right) \right) $ is a mild solution of (\ref{F3E2}), (\ref{F3E3})
in the sense of Definition \ref{mild}. Let $g$ be as in (\ref{F3E3a}). Then:%
\begin{equation*}
\partial_{t}\left( \int_{\mathbb{R}^{+}}g\epsilon d\epsilon\right)
=\partial_{t}\left( \int_{\mathbb{R}^{+}}gd\epsilon\right) =0\, ,\ \ a.e. \
t\in\left[ 0,T\right) . 
\end{equation*}
\end{lemma}

\begin{proof}
We apply (\ref{F5E1a}) with the test functions:%
\begin{equation*}
\varphi_{n}=\zeta_{n}\left( \epsilon\right) \epsilon 
\end{equation*}
where $\zeta_{n}$ is a cutoff function satisfying $\zeta_{n}\left(
\epsilon\right) =1\ $if $\epsilon\leq n,\ \zeta_{n}\left( \epsilon\right) =0$
if $\epsilon\geq n+1,\ \zeta_{n}\geq0,\ \zeta_{n}^{\prime}\leq 0,\
\zeta_{n}\in C^{\infty}\left( \mathbb{R}\right) .$ Notice that, since $%
\gamma>3,$ it is possible to pass to the limit in the cubic term:%
\begin{equation*}
\lim_{n\rightarrow\infty}\int_{\mathbb{R}^{+}}\int_{\mathbb{R}^{+}}\int_{%
\mathbb{R}^{+}}\frac{g_{1}g_{2}g_{3}\Phi}{\sqrt{\epsilon_{1}\epsilon
_{2}\epsilon_{3}}}Q_{\varphi_{n}}d\epsilon_{1}d\epsilon_{2}d\epsilon_{3}=%
\int_{\mathbb{R}^{+}}\int_{\mathbb{R}^{+}}\int_{\mathbb{R}^{+}}\frac {%
g_{1}g_{2}g_{3}\Phi}{\sqrt{\epsilon_{1}\epsilon_{2}\epsilon_{3}}}%
Q_{\varphi_{\infty}}d\epsilon_{1}d\epsilon_{2}d\epsilon_{3}, 
\end{equation*}
where $\varphi_{\infty}\left( \epsilon\right) =\epsilon.$ In order to pass
to the limit in the quadratic term a more careful argument is needed. We
first estimate the integral:%
\begin{align*}
\left\vert \int_{\mathbb{R}^{+}}\Phi Q_{\varphi_{n}}d\epsilon_{3}\right\vert
&\leq C\min\left\{ \sqrt{\epsilon_{1}},\sqrt{\epsilon_{2}}\right\} \left(
1+\epsilon_{1}+\epsilon_{2}\right) ^{2}\\
&\leq C\min\left\{ \sqrt{\epsilon_{1}},%
\sqrt{\epsilon_{2}}\right\} \left( 1+\left( \epsilon_{1}\right) ^{2}+\left(
\epsilon_{2}\right) ^{2}\right) . 
\end{align*}

Then, the quadratic integral can then estimated uniformly in $n$ as:%
\begin{align*}
& C\int_{\mathbb{R}^{+}}\int_{\mathbb{R}^{+}}\frac{g_{1}g_{2}}{\sqrt {%
\epsilon_{1}\epsilon_{2}}}\min\left\{ \sqrt{\epsilon_{1}},\sqrt{\epsilon _{2}%
}\right\} \left( 1+\left( \epsilon_{1}\right) ^{2}+\left(
\epsilon_{2}\right) ^{2}\right) d\epsilon_{1}d\epsilon_{2} \leq\\
&\hskip 1.5cm  \leq C\int_{\mathbb{R}^{+}}d\epsilon_{1}\int_{0}^{\epsilon_{1}}d%
\epsilon_{2}\frac{g_{1}g_{2}}{\sqrt{\epsilon_{1}}}\left( 1+\left(
\epsilon_{1}\right) ^{2}\right) \leq C\int_{\mathbb{R}^{+}}f_{1}\left(
\epsilon_{1}\right) ^{2}d\epsilon_{1},
\end{align*}
and since $\gamma>3$ this integral is convergent and we can take the limit $%
n\rightarrow\infty.$ Then:%
\begin{align}
\partial_{t}\left( \int_{\mathbb{R}^{+}}g\varphi d\epsilon\right) = & \int_{%
\mathbb{R}^{+}}g\partial_{t}\varphi d\epsilon+\frac{1}{2^{\frac{5}{2}}}\int_{%
\mathbb{R}^{+}}\int_{\mathbb{R}^{+}}\int_{\mathbb{R}^{+}}\frac {%
g_{1}g_{2}g_{3}\Phi}{\sqrt{\epsilon_{1}\epsilon_{2}\epsilon_{3}}}Q_{\varphi
}d\epsilon_{1}d\epsilon_{2}d\epsilon_{3}+  \notag \\
& +\frac{\pi}{2}\int_{\mathbb{R}^{+}}\int_{\mathbb{R}^{+}}\int_{\mathbb{R}%
^{+}}\frac{g_{1}g_{2}\Phi}{\sqrt{\epsilon_{1}\epsilon_{2}}}Q_{\varphi
}d\epsilon_{1}d\epsilon_{2}d\epsilon_{3}\, ,\ \ a.e.\ t \in\left[ 0,T\right)
,
\end{align}
with $\varphi\left( \epsilon\right) =\epsilon.$ Since $Q_{\varphi}=0$ we
obtain $\partial_{t}\left( \int_{\mathbb{R}^{+}}g\epsilon d\epsilon\right)
=0,\ $ $a.e.\ t\in\left[ 0,T\right) .$ A similar argument using the sequence
of test functions $\varphi_{n}=\zeta_{n}\left( \epsilon\right) $ yields $%
\partial_{t}\left( \int_{\mathbb{R}^{+}}gd\epsilon\right) =0,\ $ $a.e.\ t\in%
\left[ 0,T\right) .$
\end{proof}

As a next step we prove uniqueness of the mild solutions of (\ref{F3E2}), (%
\ref{F3E3}). Our goal is to prove the following:

\begin{proposition}
\label{uniq}Suppose that $\gamma>3$ and $f,\ \tilde{f}\in L_{loc}^{\infty
}\left( \left[ 0,T\right) ;L^{\infty}\left( \mathbb{R}^{+};\left(
1+\epsilon\right) ^{\gamma}\right) \right) $ are mild solutions of (\ref%
{F3E2}), (\ref{F3E3}) in the sense of Definition \ref{mild} with initial
data $f\left( 0,\cdot\right) =\tilde{f}\left( 0,\cdot\right) =f_{0}\left(
\cdot\right) \in L^{\infty}\left( \mathbb{R}^{+};\left( 1+\epsilon\right)
^{\gamma}\right) .$ Then $f=\tilde{f}.$
\end{proposition}

In order to prove Proposition \ref{uniq} we begin with a preliminary
computation.

\begin{lemma}
\label{ident}Suppose that $f,\ \tilde{f}$ are as in Proposition \ref{uniq}.
Then, the following identity holds:%
\begin{align}
& \left( f\left( t,\epsilon_{1}\right) -\tilde{f}\left( t,\epsilon
_{1}\right) \right) =f_{0}\left( \epsilon_{1}\right) \Omega\left(
t,\epsilon_{1}\right) \left[ \exp\left( -\int_{0}^{t}S\left[ f\right] \left(
s,\epsilon_{1}\right) ds\right) -\right.  \notag \\
& \hskip 8cm \left. -\exp\left( -\int_{0}^{t}S\left[ \tilde{f}\right]
\left( s,\epsilon_{1}\right) ds\right) \right] +  \notag \\
& +\frac{8\pi^{2}}{\sqrt{2}}\int_{0}^{t}\Omega\left( t-s,\epsilon
_{1}\right)\times \notag \\
&\times  \left[ \exp\left( -\int_{s}^{t}S\left[ f\right] \left(
s,\epsilon_{1}\right) ds\right) -\exp\left( -\int_{s}^{t}S\left[ \tilde {f}%
\right] \left( s,\epsilon_{1}\right) ds\right) \right] J\left[ f\right]
\left( s,\epsilon_{1}\right) ds+  \notag \\
& +\frac{8\pi^{2}}{\sqrt{2}}\int_{0}^{t}\Omega\left( t-s,\epsilon
_{1}\right) \exp\left( -\int_{s}^{t}S\left[ \tilde{f}\right] \left(
s,\epsilon_{1}\right) ds\right) \left[ J\left[ f\right] \left(
s,\epsilon_{1}\right) -J\left[ \tilde{f}\right] \left( s,\epsilon
_{1}\right) \right] ds,\   \label{G3E3}
\end{align}
where $M_{0}=4\pi\int_{0}^{\infty}f_{0}\left( \epsilon\right) \sqrt {%
2\epsilon}d\epsilon,$ $S\left[ \cdot\right] $ is defined as in (\ref{F8E5a}%
), (\ref{F8E5}) and:%
\begin{equation}
J\left[ f\right] =\int_{0}^{\infty}\int_{0}^{\infty}f_{3}f_{4}\left(
1+f_{1}+f_{2}\right) Wd\epsilon_{3}d\epsilon_{4},   \label{G3E3a}
\end{equation}%
\begin{equation*}
\Omega\left( t,\epsilon_{1}\right) =\exp\left( -\pi M_{0}t\sqrt {\epsilon_{1}%
}\right) . 
\end{equation*}
\end{lemma}

\begin{proof}
By Lemma \ref{massenergy} we have:%
\begin{equation}
4\pi\int_{0}^{\infty}f\left( t,\epsilon\right) \sqrt{2\epsilon}%
d\epsilon=4\pi\int_{0}^{\infty}\tilde{f}\left( t,\epsilon\right) \sqrt{%
2\epsilon}d\epsilon=M_{0}\ \ ,\ \ a.e.\ \ t\in\left[ 0,T\right) . 
\label{G3E1}
\end{equation}

Let $\Psi\left( t,\epsilon_{1}\right) $ be as in (\ref{F3E6a}) and let us
denote as $\tilde{\Psi}\left( t,\epsilon_{1}\right) $ the corresponding
function associated to $\tilde{f}.$ Due to (\ref{G3E1}) we have:%
\begin{align}
\Psi\left( t,\epsilon_{1}\right) & =\Omega\left( t,\epsilon_{1}\right)
\exp\!\left( -\int_{0}^{t}S\left[ f\right] \left( s,\epsilon_{1}\right)
ds\right) \label{G3E2a}\\
 \tilde{\Psi}\left( t,\epsilon_{1}\right) & =\Omega\left(
t,\epsilon_{1}\right) \exp\left( -\int_{0}^{t}S\left[ \tilde{f}\right]
\left( s,\epsilon_{1}\right) ds\right) .   \label{G3E2ab}
\end{align}

By definition $f$ and $\tilde{f}$ satisfy (\ref{F3E6}). Taking the
difference of these equations we obtain:%
\begin{align*}
 \left( f\left( t,\epsilon_{1}\right) -\tilde{f}\left( t,\epsilon
_{1}\right) \right) &=f_{0}\left( \epsilon_{1}\right) \left( \Psi\left(
t,\epsilon_{1}\right) -\tilde{\Psi}\left( t,\epsilon_{1}\right) \right) + \\
& +\frac{8\pi^{2}}{\sqrt{2}}\int_{0}^{t}\left[ \frac{\Psi\left(
t,\epsilon_{1}\right) }{\Psi\left( s,\epsilon_{1}\right) }-\frac {\tilde{\Psi%
}\left( t,\epsilon_{1}\right) }{\tilde{\Psi}\left( s,\epsilon_{1}\right) }%
\right] J\left[ f\right] \left( s,\epsilon _{1}\right) ds+\\
&+\frac{8\pi^{2}}{%
\sqrt{2}}\int_{0}^{t}\frac{\tilde{\Psi}\left( t,\epsilon_{1}\right) }{\tilde{%
\Psi}\left( s,\epsilon_{1}\right) }\left[ J\left[ f\right] \left(
s,\epsilon_{1}\right) -J\left[ \tilde{f}\right] \left( s,\epsilon_{1}\right) %
\right] ds.
\end{align*}
Using (\ref{G3E2a}) and (\ref{G3E2ab})we deduce (\ref{G3E3})
\end{proof}

In the next Lemmas we estimate the differences of the terms containing $S%
\left[ f\right]$, $S\left[ \tilde{f}\right] $ and $J\left[ f\right] ,\ J%
\left[ \tilde{f}\right] $ respectively.

\begin{lemma}
\label{LemmaS}Suppose that $f$ and $\tilde{f}$ are as in Proposition \ref%
{uniq}. Then, we have the following estimates:%
\begin{align}
& \left\vert \exp\left( -\int_{s}^{t}S\left[ f\right] \left( \xi
,\epsilon_{1}\right) d\xi\right) -\exp\left( -\int_{s}^{t}S\left[ \tilde{f}%
\right] \left( \xi,\epsilon_{1}\right) d\xi\right) \right\vert  \label{G3E4}
\\
& \leq\int_{s}^{t}\left\vert S\left[ f\right] \left( \xi,\epsilon
_{1}\right) -S\left[ \tilde{f}\right] \left( \xi,\epsilon_{1}\right)
\right\vert d\xi\ \ ,\ \ 0\leq s\leq t\leq T\, ,\ \epsilon_{1}>0\, ,\ \
0\leq s\leq t\leq T,  \notag
\end{align}%
\begin{equation}
\left\vert S\left[ f\right] \left( t,\epsilon_{1}\right) -S\left[ \tilde{f}%
\right] \left( t,\epsilon_{1}\right) \right\vert \leq\frac {C}{1+\sqrt{%
\epsilon_{1}}}\left\Vert f-\tilde{f}\right\Vert _{L^{\infty }\left( \mathbb{R%
}^{+},\left( 1+\epsilon\right) ^{\gamma}\right) }, \epsilon_{1}>0\ ,\
0\leq t\leq T,   \label{G3E5}
\end{equation}
for some suitable constant $C>0$ depending only on $\gamma.$
\end{lemma}

\begin{proof}
The estimate (\ref{G3E4}) is a consequence of the inequality $\left\vert
e^{-A}-e^{-B}\right\vert \leq\left\vert A-B\right\vert , \
A\geq0,\hfill\break B\geq0$. In order to prove (\ref{G3E5}) we first notice
that:%
\begin{align*}
\left\vert S\left[ f\right] \left( s,\epsilon_{1}\right) -S\left[ \tilde{f}%
\right] \left( s,\epsilon_{1}\right) \right\vert &\leq\left\vert S_{1}\left[ f%
\right] \left( s,\epsilon_{1}\right) -S_{1}\left[ \tilde {f}\right] \left(
s,\epsilon_{1}\right) \right\vert +\\
&+\left\vert S_{2}\left[ f\right] \left(
s,\epsilon_{1}\right) -S_{2}\left[ \tilde {f}\right] \left(
s,\epsilon_{1}\right) \right\vert 
\end{align*}

The definition of $S_{1}\left[ \cdot\right] $ in (\ref{F8E5}) yields:%
\begin{align*}
\left\vert S_{1}\left[ f\right] \left( s,\epsilon_{1}\right) -S_{1}\left[ 
\tilde{f}\right] \left( s,\epsilon_{1}\right) \right\vert &\leq\frac {C}{%
\left( 1+\sqrt{\epsilon_{1}}\right) }\int_{0}^{\infty}\left\vert f_{2}-%
\tilde{f}_{2}\right\vert \left( \epsilon_{2}\right) ^{\frac{3}{2}%
}d\epsilon_{2}+\\
&+C\int_{\epsilon_{1}}^{\infty}\epsilon_{2}\left\vert f_{2}-%
\tilde{f}_{2}\right\vert d\epsilon_{2}. 
\end{align*}

Then, using that $\gamma>3:$%
\begin{align}
\left\vert S_{1}\left[ f\right] \left( s,\epsilon_{1}\right) -S_{1}\left[ 
\tilde{f}\right] \left( s,\epsilon_{1}\right) \right\vert & \leq\frac {C}{1+%
\sqrt{\epsilon_{1}}}\left\Vert f_{2}-\tilde{f}_{2}\right\Vert
_{L^{\infty}\left( \mathbb{R}^{+},\left( 1+\epsilon\right) ^{\gamma }\right)
}+  \notag \\
& \hskip 2cm +\frac{C}{\left( 1+\epsilon_{1}\right) ^{\gamma-2}}\left\Vert
f_{2}-\tilde{f}_{2}\right\Vert _{L^{\infty}\left( \mathbb{R}^{+},\left(
1+\epsilon\right) ^{\gamma}\right) }  \notag \\
& \leq\frac{C}{1+\sqrt{\epsilon_{1}}}\left\Vert f_{2}-\tilde{f}%
_{2}\right\Vert _{L^{\infty}\left( \mathbb{R}^{+},\left( 1+\epsilon\right)
^{\gamma}\right) }.   \label{G3E6}
\end{align}

On the other hand, using (\ref{F8E5}): 
\begin{equation*}
\left\vert S_{2}\left[ f\right] \left( s,\epsilon_{1}\right) -S_{2}\left[ 
\tilde{f}\right] \left( s,\epsilon_{1}\right) \right\vert \leq C\int
_{0}^{\infty}\int_{0}^{\infty}\left\vert f_{2}\left( f_{3}+f_{4}\right) -%
\tilde{f}_{2}\left( \tilde{f}_{3}+\tilde{f}_{4}\right) \right\vert
Wd\epsilon_{3}d\epsilon_{4}. 
\end{equation*}

Then, using also the symmetry we obtain the following estimate for $%
\epsilon_{1}\geq1$:%
\begin{align*}
\left\vert S_{2}\left[ f\right] \left( s,\epsilon_{1}\right) -S_{2}\left[ 
\tilde{f}\right] \left( s,\epsilon_{1}\right) \right\vert & \leq
C\int_{0}^{\infty}\int_{0}^{\infty}\left\vert f_{2}-\tilde{f}_{2}\right\vert
f_{3}Wd\epsilon_{3}d\epsilon_{4}+  \notag \\
&\hskip 2cm +C\int_{0}^{\infty}\int_{0}^{\infty}\tilde {f}_{2}\left\vert
f_{3}-\tilde{f}_{3}\right\vert Wd\epsilon_{3}d\epsilon_{4} \\
& \leq\frac{C}{\sqrt{\epsilon_{1}}}\int_{0}^{\infty}\left\vert f_{2}-\tilde{f%
}_{2}\right\vert d\epsilon_{2}\int_{0}^{\infty}f_{3}\sqrt {\epsilon_{3}}%
d\epsilon_{3}+  \notag \\
&\hskip 2cm+\frac{C}{\sqrt{\epsilon_{1}}}\int_{0}^{\infty }\tilde{f}%
_{2}d\epsilon_{2}\int_{0}^{\infty}\left\vert f_{3}-\tilde{f}_{3}\right\vert 
\sqrt{\epsilon_{3}}d\epsilon_{3},
\end{align*}
whence%
\begin{equation*}
\left\vert S_{2}\left[ f\right] \left( s,\epsilon_{1}\right) -S_{2}\left[ 
\tilde{f}\right] \left( s,\epsilon_{1}\right) \right\vert \leq\frac {C}{%
\sqrt{\epsilon_{1}}}\int_{0}^{\infty}\left\vert f-\tilde{f}\right\vert
\left( 1+\sqrt{\epsilon}\right) d\epsilon. 
\end{equation*}

We can obtain also estimates for $\epsilon_{1}\leq1$ using $W\leq1.$ Then,
combining the estimates for $\epsilon_{1}\leq1$ and $\epsilon_{1}\geq1$%
\begin{align}
\left\vert S_{2}\left[ f\right] \left( s,\epsilon_{1}\right) -S_{2}\left[ 
\tilde{f}\right] \left( s,\epsilon_{1}\right) \right\vert \leq\frac {C}{%
\left( 1+\sqrt{\epsilon_{1}}\right) }\int_{0}^{\infty}\left\vert f-\tilde{f}%
\right\vert \left( 1+\sqrt{\epsilon}\right) d\epsilon  \notag \\
\leq\frac{C}{\left( 1+\sqrt{\epsilon_{1}}\right) }\left\Vert f_{2}-\tilde {f}%
_{2}\right\Vert _{L^{\infty}\left( \mathbb{R}^{+},\left( 1+\epsilon \right)
^{\gamma}\right) }.   \label{G3E7}
\end{align}

Using (\ref{G3E6}) and (\ref{G3E7}) we obtain (\ref{G3E5}).
\end{proof}

We now estimate the difference $\left[ J\left[ f\right] \left(
s,\epsilon_{1}\right) -J\left[ \tilde{f}\right] \left( s,\epsilon
_{1}\right) \right] $ in (\ref{G3E3}) as well as the functions $J\left[ f%
\right] \left( s,\epsilon_{1}\right) ,\ J\left[ \tilde{f}\right] \left(
s,\epsilon_{1}\right) $. The following result requires to use the same type
of detailed estimates for $a\left( \epsilon,t\right) $ used in the Proof of
Lemma \ref{LinfEst}.

\begin{lemma}
\label{LemmaJ}There exists a constant $C>0$, depending only on $\gamma$,
such that, for any $f,\ \tilde{f}$ as in Proposition \ref{uniq}: 
\begin{equation}
0\leq\max\left\{ J\left[ f\right] \left( t,\epsilon_{1}\right) ,J\left[ 
\tilde{f}\right] \left( t,\epsilon_{1}\right) \right\} \leq\frac {C}{\left(
1+\epsilon_{1}\right) ^{\gamma-\frac{1}{2}}}\ \ ,\ \ \epsilon _{1}>0\ \ ,\
0\leq t\leq T,   \label{G3E8}
\end{equation}%
\begin{align}
& \left\vert J\left[ f\right] \left( t,\epsilon_{1}\right) -J\left[ \tilde{f}%
\right] \left( t,\epsilon_{1}\right) \right\vert\leq
\frac{\sqrt{2}}{8\pi^{2}}\frac{M_{0}\theta\sqrt{\epsilon_{1}}}{\left(
1+\epsilon_{1}\right) ^{\gamma}}\left\Vert f-\tilde{f}\right\Vert
_{L^{\infty}\left( \mathbb{R}^{+},\left( 1+\epsilon\right) ^{\gamma }\right)
}+\notag \\
&\hskip 3cm +\frac{C\sqrt{\epsilon_{1}}}{\left( 1+\epsilon_{1}\right) ^{\gamma}}%
\int_{0}^{\infty}\left\vert f-\tilde{f}\right\vert \sqrt{\epsilon }d\epsilon+
\frac{C\left\Vert f-\tilde{f}\right\Vert _{L^{\infty}\left( \mathbb{R}%
^{+},\left( 1+\epsilon\right) ^{\gamma}\right) }}{\left(
1+\epsilon_{1}\right) ^{\gamma+\frac{1}{2}}},  \label{G3E9}
\end{align}
for $\epsilon_{1}>0\ \ ,\ 0\leq t\leq T.$
\end{lemma}

\begin{proof}
Using (\ref{G3E3a}) and the fact that $W\leq\min\left\{ \frac{\sqrt {%
\epsilon_{3}}}{\sqrt{\epsilon_{1}}},\frac{\sqrt{\epsilon_{4}}}{\sqrt {%
\epsilon_{1}}}\right\} $ and $W=0$ if $\epsilon_{3}+\epsilon_{4}\leq
\epsilon_{1}$ we obtain (\ref{G3E8}).

In order to estimate $\left[ J\left[ f\right] \left( s,\epsilon _{1}\right)
-J\left[ \tilde{f}\right] \left( s,\epsilon_{1}\right) \right] $ we first
use (\ref{G3E3a}) to obtain:%
\begin{equation*}
\left\vert J\left[ f\right] \left( s,\epsilon_{1}\right) -J\left[ \tilde{f}%
\right] \left( s,\epsilon_{1}\right) \right\vert \leq K_{1}+K_{2}+K_{3}, 
\end{equation*}%
\begin{align*}
K_{1} & =\int_{0}^{\infty}\int_{0}^{\infty}\left\vert f_{3}f_{4}-\tilde {f}%
_{3}\tilde{f}_{4}\right\vert Wd\epsilon_{3}d\epsilon_{4}\, ,\ \
K_{2}=\int_{0}^{\infty}\int_{0}^{\infty}\left\vert f_{3}f_{4}f_{1}-\tilde{f}%
_{3}\tilde{f}_{4}\tilde{f}_{1}\right\vert Wd\epsilon_{3}d\epsilon_{4}, \\
K_{3} & =\int_{0}^{\infty}\int_{0}^{\infty}\left\vert f_{3}f_{4}f_{2}-\tilde{%
f}_{3}\tilde{f}_{4}\tilde{f}_{2}\right\vert Wd\epsilon_{3}d\epsilon_{4}.
\end{align*}

The term $K_{1}$ must be carefully estimated, with the methods used in the
Proof of Lemma \ref{LinfEst}. Using the symmetry with respect to the
variables $\epsilon_{3},\ \epsilon_{4}$ we obtain:%
\begin{align*}
K_{1} & \leq\int_{0}^{\infty}\int_{0}^{\infty}f_{3}\left\vert f_{4}-\tilde{f}%
_{4}\right\vert
Wd\epsilon_{3}d\epsilon_{4}+\int_{0}^{\infty}\int_{0}^{\infty}\tilde{f}%
_{4}\left\vert f_{3}-\tilde{f}_{3}\right\vert Wd\epsilon_{3}d\epsilon_{4} \\
& =\int_{0}^{\infty}\int_{0}^{\infty}\left( f_{3}+\tilde{f}_{3}\right)
\left\vert f_{4}-\tilde{f}_{4}\right\vert Wd\epsilon_{3}d\epsilon_{4}.
\end{align*}

We now introduce numbers $L>0$ and $0<\mu<1$ as in the Proof of Lemma \ref%
{LinfEst}. We then estimate $K_{1}$ for $\epsilon_{1}<L$ as:%
\begin{align}
K_{1}\leq\frac{C}{\left( 1+\sqrt{\epsilon_{1}}\right) }\int_{0}^{\infty
}\left( f_{3}+\tilde{f}_{3}\right) d\epsilon_{3}\int_{0}^{\infty}\left\vert
f_{4}-\tilde{f}_{4}\right\vert \sqrt{\epsilon_{4}}d\epsilon_{4}  \notag \\
\leq\frac{C}{\left( 1+\sqrt{\epsilon_{1}}\right) }\int_{0}^{\infty
}\left\vert f_{4}-\tilde{f}_{4}\right\vert \sqrt{\epsilon_{4}}d\epsilon
_{4}\ \ \ ,\ \ \ \epsilon_{1}<L.   \label{G4E1}
\end{align}

On the other hand, in order to estimate $K_{1}$ for $\epsilon_{1}\geq L$ we
introduce an auxiliary parameter $\mu<1$ as in Lemma \ref{LinfEst} and whose
precise value will be determined later. We then have, using also (\ref{F3E3}%
):%
\begin{align*}
K_{1} & \leq\frac{1}{\sqrt{\epsilon_{1}}}\left( \int_{\left( 1-\mu\right)
\epsilon_{1}}^{\infty}\left( f_{3}+\tilde{f}_{3}\right) d\epsilon
_{3}\right) \left( \int_{\left( 1-\mu\right) \epsilon_{1}}^{\infty
}\left\vert f_{4}-\tilde{f}_{4}\right\vert \sqrt{\epsilon_{4}}d\epsilon
_{4}\right) + \\
& +\frac{1}{\sqrt{\epsilon_{1}}}\left( \int_{\mu\epsilon_{1}}^{\infty
}\left( f_{3}+\tilde{f}_{3}\right) d\epsilon_{3}\right) \left( \int
_{0}^{\infty}\left\vert f_{4}-\tilde{f}_{4}\right\vert \sqrt{\epsilon_{4}}%
d\epsilon_{4}\right) + \\
& +\frac{1}{\sqrt{\epsilon_{1}}}\left( \int_{0}^{\infty}\left( f_{3}+\tilde{f%
}_{3}\right) \sqrt{\epsilon_{3}}d\epsilon_{3}\right) \left(
\int_{\mu\epsilon_{1}}^{\infty}\left\vert f_{4}-\tilde{f}_{4}\right\vert
d\epsilon_{4}\right) ,
\end{align*}
if $\epsilon_{1}\geq L.$ Taking into account the definitions of $\left\Vert
\cdot\right\Vert _{L^{\infty}\left( \mathbb{R}^{+},\left( 1+\epsilon\right)
^{\gamma}\right) }$ and $M_{0}$ we obtain:%
\begin{align*}
K_{1} & \leq\frac{C_{\mu}}{\sqrt{\epsilon_{1}}}\frac{1}{\left( \epsilon
_{1}\right) ^{2\gamma-2}}\left\Vert f-\tilde{f}\right\Vert _{L^{\infty
}\left( \mathbb{R}^{+},\left( 1+\epsilon\right) ^{\gamma}\right) }+\frac{%
C_{\mu}}{\left( \epsilon_{1}\right) ^{\gamma-\frac{1}{2}}}\int
_{0}^{\infty}\left\vert f_{4}-\tilde{f}_{4}\right\vert \sqrt{\epsilon_{4}}%
d\epsilon_{4}+ \\
& +\frac{\sqrt{2}}{8\pi^{2}}\frac{2M_{0}}{\left( \gamma-1\right) }\frac{%
\left\Vert f-\tilde{f}\right\Vert _{L^{\infty}\left( \mathbb{R}^{+},\left(
1+\epsilon\right) ^{\gamma}\right) }}{\left( \mu\epsilon _{1}\right)
^{\gamma-\frac{1}{2}}},
\end{align*}
with $C_{\mu}>0$ depending on $\mu.$ Choosing $\mu$ sufficiently close to
one, assuming that $L$ is large, and using the fact that $\gamma>3$ we
obtain:%
\begin{equation*}
\frac{2}{\left( \gamma-1\right) }\frac{1}{\left( \mu\epsilon_{1}\right)
^{\gamma}}+\frac{C}{\left( \epsilon_{1}\right) ^{2\gamma-1}}\leq\frac {\theta%
}{\left( 1+\epsilon_{1}\right) ^{\gamma}}\ \ ,\ \ \epsilon_{1}\geq L\ \ 
\text{ }
\end{equation*}
for some $\theta<1.$ Then:%
\begin{equation*}
K_{1}\leq\frac{\sqrt{2}}{8\pi^{2}}\frac{M_{0}\theta\sqrt{\epsilon_{1}}}{%
\left( 1+\epsilon_{1}\right) ^{\gamma}}\left\Vert f-\tilde{f}\right\Vert
_{L^{\infty}\left( \mathbb{R}^{+},\left( 1+\epsilon\right) ^{\gamma }\right)
}+\frac{C\sqrt{\epsilon_{1}}}{\left( 1+\epsilon_{1}\right) ^{\gamma}}%
\int_{0}^{\infty}\left\vert f_{4}-\tilde{f}_{4}\right\vert \sqrt{\epsilon_{4}%
}d\epsilon_{4}\ \ ,\ \ \epsilon_{1}\geq L. 
\end{equation*}

Combining this estimate with (\ref{G1E1}) in order to include also the
contribution of the region $\left\{ \epsilon_{1}<L\right\} $ we obtain:%
\begin{equation}
K_{1}\leq\frac{\sqrt{2}}{8\pi^{2}}\frac{M_{0}\theta\sqrt{\epsilon_{1}}}{%
\left( 1+\epsilon_{1}\right) ^{\gamma}}\left\Vert f-\tilde{f}\right\Vert
_{L^{\infty}\left( \mathbb{R}^{+},\left( 1+\epsilon\right) ^{\gamma }\right)
}+\frac{C\sqrt{\epsilon_{1}}}{\left( 1+\epsilon_{1}\right) ^{\gamma}}%
\int_{0}^{\infty}\left\vert f-\tilde{f}\right\vert \sqrt{\epsilon }%
d\epsilon\ \ ,\ \ \epsilon_{1}\geq0.   \label{G4E2}
\end{equation}

Notice that $C$ depends only in $\gamma$ and $\theta.$ We now estimate $%
K_{2} $:%
\begin{equation*}
K_{2}\leq\int_{0}^{\infty}\int_{0}^{\infty}\left[ \left\vert f_{3}-\tilde {f}%
_{3}\right\vert f_{4}f_{1}+\tilde{f}_{3}\left\vert f_{4}-\tilde{f}%
_{4}\right\vert f_{1}+\tilde{f}_{3}\tilde{f}_{4}\left\vert f_{1}-\tilde{f}%
_{1}\right\vert \right] Wd\epsilon_{3}d\epsilon_{4}. 
\end{equation*}

The first two terms on the right-hand side can be easily estimated using the
boundedness of $f:$%
\begin{align*}
&\int_{0}^{\infty}\int_{0}^{\infty}\left[ \left\vert f_{3}-\tilde{f}%
_{3}\right\vert f_{4}f_{1}+\tilde{f}_{3}\left\vert f_{4}-\tilde{f}%
_{4}\right\vert f_{1}\right] Wd\epsilon_{3}d\epsilon_{4}\leq \\
&\hskip 6cm \leq \frac{C}{\left(
1+\epsilon_{1}\right) ^{\gamma+\frac{1}{2}}}\int_{0}^{\infty}\left\vert
f_{3}-\tilde{f}_{3}\right\vert \sqrt{\epsilon_{3}}d\epsilon_{3}. 
\end{align*}

On the other hand:%
\begin{equation*}
\int_{0}^{\infty}\int_{0}^{\infty}\tilde{f}_{3}\tilde{f}_{4}\left\vert f_{1}-%
\tilde{f}_{1}\right\vert Wd\epsilon_{3}d\epsilon_{4}\leq\frac {C\left\Vert f-%
\tilde{f}\right\Vert _{L^{\infty}\left( \mathbb{R}^{+},\left(
1+\epsilon\right) ^{\gamma}\right) }}{\left( 1+\epsilon_{1}\right) ^{\gamma+%
\frac{1}{2}}}. 
\end{equation*}

Then:%
\begin{align}
K_{2}&\leq\frac{C}{\left( 1+\epsilon_{1}\right) ^{\gamma+\frac{1}{2}}}%
\int_{0}^{\infty}\left\vert f_{3}-\tilde{f}_{3}\right\vert \sqrt{\epsilon_{3}%
}d\epsilon_{3}+\frac{C\left\Vert f-\tilde{f}\right\Vert _{L^{\infty}\left( 
\mathbb{R}^{+},\left( 1+\epsilon\right) ^{\gamma}\right) }}{\left(
1+\epsilon_{1}\right) ^{\gamma+\frac{1}{2}}}\notag\\
&\leq\frac{C\left\Vert f-\tilde {f%
}\right\Vert _{L^{\infty}\left( \mathbb{R}^{+},\left( 1+\epsilon\right)
^{\gamma}\right) }}{\left( 1+\epsilon_{1}\right) ^{\gamma+\frac{1}{2}}}. 
\label{G4E3}
\end{align}

We now estimate $K_{3}.$ Notice that: 
\begin{align*}
K_{3}&\leq\int_{0}^{\infty}\int_{0}^{\infty}\left\vert f_{3}f_{4}-\tilde{f}%
_{3}\tilde{f}_{4}\right\vert
f_{2}Wd\epsilon_{3}d\epsilon_{4}+\int_{0}^{\infty}\int_{0}^{\infty}\tilde{f}%
_{3}\tilde{f}_{4}\left\vert f_{2}-\tilde {f}_{2}\right\vert
Wd\epsilon_{3}d\epsilon_{4}\\
&=K_{3,1}+K_{3,2}. 
\end{align*}
If $\epsilon_{1}\leq1$ both terms on the right-hand side can be estimated by 
$C\int_{0}^{\infty}\left\vert f_{4}-\tilde{f}_{4}\right\vert d\epsilon_{4}$.
If $\epsilon_{1}>1$ we use the fact that at least one of the integration
variables $\epsilon_{3}$ or $\epsilon_{4}$ is larger than $\frac{\epsilon_{1}%
}{2}.$ Then:%
\begin{equation*}
K_{3,2}\leq\frac{C}{\left( 1+\epsilon_{1}\right) ^{\gamma+\frac{1}{2}}}%
\int_{0}^{\infty}\int_{0}^{\infty}\tilde{f}_{4}\left\vert f_{2}-\tilde{f}%
_{2}\right\vert \sqrt{\epsilon_{4}}d\epsilon_{3}d\epsilon_{4}\leq\frac {C}{%
\left( 1+\epsilon_{1}\right) ^{\gamma+\frac{1}{2}}}\int_{0}^{\infty
}\left\vert f-\tilde{f}\right\vert d\epsilon. 
\end{equation*}

Using also the fact that at least one of the integration variables is larger
than $\frac{\epsilon_{1}}{2}$ we obtain:%
\begin{align*}
& K_{3,1}\leq\int_{\frac{\epsilon_{1}}{2}}^{\infty}d\epsilon_{3}\int
_{0}^{\infty}d\epsilon_{4}f_{3}\left\vert f_{4}-\tilde{f}_{4}\right\vert
f_{2}W+\int_{\frac{\epsilon_{1}}{2}}^{\infty}d\epsilon_{3}\int_{0}^{\infty
}d\epsilon_{4}\left\vert f_{3}-\tilde{f}_{3}\right\vert f_{4}f_{2}W+ \\
& +\int_{0}^{\infty}d\epsilon_{3}\int_{\frac{\epsilon_{1}}{2}}^{\infty
}d\epsilon_{4}f_{3}\left\vert f_{4}-\tilde{f}_{4}\right\vert f_{2}W+\int
_{0}^{\infty}d\epsilon_{3}\int_{\frac{\epsilon_{1}}{2}}^{\infty}d\epsilon
_{4}\left\vert f_{3}-\tilde{f}_{3}\right\vert f_{4}f_{2}W.
\end{align*}

Therefore, using the definition of the norms $\left\Vert \cdot\right\Vert
_{L^{\infty}\left( \mathbb{R}^{+},\left( 1+\epsilon\right) ^{\gamma }\right)
}$:%
\begin{align*}
 K_{3,1}&\leq\frac{C}{\left( 1+\epsilon_{1}\right) ^{\gamma}}\int _{\frac{%
\epsilon_{1}}{2}}^{\infty}d\epsilon_{3}\int_{0}^{\infty}d\epsilon
_{4}\left\vert f_{4}-\tilde{f}_{4}\right\vert f_{2}W+ \\
& +\frac{C\left\Vert f-\tilde{f}\right\Vert _{L^{\infty}\left( \mathbb{R}%
^{+},\left( 1+\epsilon\right) ^{\gamma}\right) }}{\left(
1+\epsilon_{1}\right) ^{\gamma}}\int_{\frac{\epsilon_{1}}{2}}^{\infty
}d\epsilon_{3}\int_{0}^{\infty}d\epsilon_{4}f_{4}f_{2}W+ \\
& +\frac{C\left\Vert f-\tilde{f}\right\Vert _{L^{\infty}\left( \mathbb{R}%
^{+},\left( 1+\epsilon\right) ^{\gamma}\right) }}{\left(
1+\epsilon_{1}\right) ^{\gamma}}\int_{0}^{\infty}d\epsilon_{3}\int _{\frac{%
\epsilon_{1}}{2}}^{\infty}d\epsilon_{4}f_{3}f_{2}W+\\
&+\frac{C}{\left(
1+\epsilon_{1}\right) ^{\gamma}}\int_{0}^{\infty}d\epsilon_{3}\int _{\frac{%
\epsilon_{1}}{2}}^{\infty}d\epsilon_{4}\left\vert f_{3}-\tilde{f}%
_{3}\right\vert f_{2}W
\end{align*}
and, relabelling variables:%
\begin{align*}
K_{3,1}&\leq\frac{C}{\left( 1+\epsilon_{1}\right) ^{\gamma}}\int _{\frac{%
\epsilon_{1}}{2}}^{\infty}d\epsilon_{3}\int_{0}^{\infty}d\epsilon
_{4}\left\vert f_{4}-\tilde{f}_{4}\right\vert f_{2}W+\\
&+\frac{C\left\Vert f-%
\tilde{f}\right\Vert _{L^{\infty}\left( \mathbb{R}^{+},\left(
1+\epsilon\right) ^{\gamma}\right) }}{\left( 1+\epsilon_{1}\right) ^{\gamma}}%
\int_{\frac{\epsilon_{1}}{2}}^{\infty}d\epsilon_{3}\int_{0}^{\infty
}d\epsilon_{4}f_{4}f_{2}W, 
\end{align*}
whence:%
\begin{align*}
K_{3,1}&\leq\frac{C}{\left( 1+\epsilon_{1}\right) ^{\gamma+\frac{1}{2}}}%
\int_{0}^{\infty}d\epsilon_{4}\left\vert f_{4}-\tilde{f}_{4}\right\vert 
\sqrt{\epsilon_{4}}+\frac{C\left\Vert f-\tilde{f}\right\Vert _{L^{\infty
}\left( \mathbb{R}^{+},\left( 1+\epsilon\right) ^{\gamma}\right) }}{\left(
1+\epsilon_{1}\right) ^{\gamma+\frac{1}{2}}}\\
&\leq\frac{C\left\Vert f-\tilde{f}%
\right\Vert _{L^{\infty}\left( \mathbb{R}^{+},\left( 1+\epsilon\right)
^{\gamma}\right) }}{\left( 1+\epsilon_{1}\right) ^{\gamma+\frac{1}{2}}}. 
\end{align*}

Then:%
\begin{align}
K_{3} & \leq\frac{C}{\left( 1+\epsilon_{1}\right) ^{\gamma+\frac{1}{2}}}%
\int_{0}^{\infty}\left\vert f-\tilde{f}\right\vert d\epsilon+\frac {%
C\left\Vert f-\tilde{f}\right\Vert _{L^{\infty}\left( \mathbb{R}^{+},\left(
1+\epsilon\right) ^{\gamma}\right) }}{\left( 1+\epsilon_{1}\right) ^{\gamma+%
\frac{1}{2}}}  \notag \\
& \leq\frac{C\left\Vert f-\tilde{f}\right\Vert _{L^{\infty}\left( \mathbb{R}%
^{+},\left( 1+\epsilon\right) ^{\gamma}\right) }}{\left(
1+\epsilon_{1}\right) ^{\gamma+\frac{1}{2}}}.   \label{G4E4}
\end{align}

Combining (\ref{G4E1}), (\ref{G4E2}), (\ref{G4E3}), (\ref{G4E4}) we obtain (%
\ref{G3E9}).
\end{proof}

\begin{proof}[Proof of Proposition \protect\ref{uniq}]
We estimate the differences $\left( f\left( t,\epsilon_{1}\right) -%
\tilde{f}\left( t,\epsilon_{1}\right) \right)$. We have (cf. (\ref{G3E3})):%
\begin{align*}
& \left\vert f\left( t,\epsilon_{1}\right) -\tilde{f}\left( t,\epsilon
_{1}\right) \right\vert \leq f_{0}\left( \epsilon_{1}\right) \Omega\left(
t,\epsilon_{1}\right) \int_{0}^{t}ds\left\vert S\left[ f\right] \left(
s,\epsilon_{1}\right) -S\left[ \tilde{f}\right] \left( s,\epsilon
_{1}\right) \right\vert + \\
& +\frac{8\pi^{2}}{\sqrt{2}}\int_{0}^{t}\Omega\left( t-s,\epsilon
_{1}\right) \int_{s}^{t}d\xi\left\vert S\left[ f\right] \left(
\xi,\epsilon_{1}\right) -S\left[ \tilde{f}\right] \left( \xi,\epsilon
_{1}\right) \right\vert J\left[ f\right] \left( s,\epsilon_{1}\right) ds+ \\
& +\frac{8\pi^{2}}{\sqrt{2}}\int_{0}^{t}\Omega\left( t-s,\epsilon
_{1}\right) \exp\left( -\int_{s}^{t}S\left[ \tilde{f}\right] \left(
\xi,\epsilon_{1}\right) d\xi\right) \left\vert J\left[ f\right] \left(
s,\epsilon_{1}\right) -J\left[ \tilde{f}\right] \left( s,\epsilon
_{1}\right) \right\vert ds.
\end{align*}

Then, using the estimates for $f_{0}\in L^{\infty}\left( \mathbb{R}%
^{+};\left( 1+\epsilon\right) ^{\gamma}\right) $ as well as Lemma \ref%
{LemmaS} and (\ref{G3E8})$:$%
\begin{align*}
\left\vert f\left( t,\epsilon_{1}\right) -\tilde{f}\left( t,\epsilon
_{1}\right) \right\vert & \leq\frac{C}{\left( 1+\epsilon_{1}\right) ^{\gamma+%
\frac{1}{2}}}\int_{0}^{t}ds\left\Vert f(s)-\tilde{f}(s)\right\Vert
_{L^{\infty}\left( \mathbb{R}^{+},\left( 1+\epsilon\right) ^{\gamma }\right)
}+ \\
& +\frac{C}{\left( 1+\epsilon_{1}\right) ^{\gamma}}\int_{0}^{t}ds\int
_{s}^{t}d\xi\left\Vert f(\xi)-\tilde{f}(\xi)\right\Vert _{L^{\infty}\left( 
\mathbb{R}^{+},\left( 1+\epsilon\right) ^{\gamma}\right) } + \\
&\hskip -3cm  +\frac{8\pi^{2}}{\sqrt{2}}\int_{0}^{t}\Omega\left( t-s,\epsilon
_{1}\right) \exp\left( \!-\!\int_{s}^{t}\!S\left[ \tilde{f}\right] \left(
\xi,\epsilon_{1}\right) d\xi\right) \! \left\vert J\left[ f\right] \left(
s,\epsilon_{1}\right) -J\left[ \tilde{f}\right] \left( s,\epsilon
_{1}\right) \right\vert ds.
\end{align*}

By Lemma \ref{LemmaJ}: 
\begin{align}
\left\vert f\left( t,\epsilon_{1}\right) -\tilde{f}\left( t,\epsilon
_{1}\right) \right\vert & \leq\frac{C}{\left( 1+\epsilon_{1}\right) ^{\gamma+%
\frac{1}{2}}}\int_{0}^{t}ds\left\Vert f(s)-\tilde{f}(s)\right\Vert
_{L^{\infty}\left( \mathbb{R}^{+},\left( 1+\epsilon\right) ^{\gamma }\right)
}+  \notag \\
& +\frac{C}{\left( 1+\epsilon_{1}\right) ^{\gamma}}\int_{0}^{t}ds\int
_{s}^{t}d\xi\left\Vert f-\tilde{f}\right\Vert _{L^{\infty}\left( \mathbb{R}%
^{+},\left( 1+\epsilon\right) ^{\gamma}\right) }\left( \xi\right) +  \notag
\\
& +\int_{0}^{t}\Omega\left( t-s,\epsilon_{1}\right) \frac{M_{0}\theta \sqrt{%
\epsilon_{1}}}{\left( 1+\epsilon_{1}\right) ^{\gamma}}\left\Vert f-\tilde{f}%
\right\Vert _{L^{\infty}\left( \mathbb{R}^{+},\left( 1+\epsilon\right)
^{\gamma}\right) }+  \notag \\
& +\frac{C}{\left( 1+\epsilon_{1}\right) ^{\gamma}}\int_{0}^{t}ds\Omega%
\left( t-s,\epsilon_{1}\right) \sqrt{\epsilon_{1}}\int_{0}^{\infty
}\left\vert f-\tilde{f}\right\vert \sqrt{\epsilon}d\epsilon+  \notag \\
& +\frac{C}{\left( 1+\epsilon_{1}\right) ^{\gamma+\frac{1}{2}}}%
\int_{0}^{t}\left\Vert f-\tilde{f}\right\Vert _{L^{\infty}\left( \mathbb{R}%
^{+},\left( 1+\epsilon\right) ^{\gamma}\right) }\left( s\right) ds. 
\label{F5E2}
\end{align}

Then, using $\int_{0}^{t}\Omega\left( t-s,\epsilon_{1}\right) M_{0}\sqrt{%
\epsilon_{1}}\leq1:$%
\begin{align*}
\sup_{0\leq s\leq t}\left\Vert f-\tilde{f}\right\Vert _{L^{\infty}\left( 
\mathbb{R}^{+},\left( 1+\epsilon\right) ^{\gamma}\right) }\left( s\right) &
\leq CT\sup_{0\leq s\leq t}\left\Vert f-\tilde{f}\right\Vert _{L^{\infty
}\left( \mathbb{R}^{+},\left( 1+\epsilon\right) ^{\gamma}\right) }\left(
s\right) + \\
& +\theta\sup_{0\leq s\leq t}\left\Vert f-\tilde{f}\right\Vert _{L^{\infty
}\left( \mathbb{R}^{+},\left( 1+\epsilon\right) ^{\gamma}\right) }\left(
s\right) + \\
& +C\sup_{0\leq s\leq t}\left( \int_{0}^{\infty}\left\vert f-\tilde {f}%
\right\vert \left( 1+\sqrt{\epsilon}\right) d\epsilon\right) \text{\ \ \ if
\ }0\leq t\leq T.
\end{align*}

Since $\theta<1$ we have:%
\begin{align*}
\sup_{0\leq s\leq t}\left\Vert f-\tilde{f}\right\Vert _{L^{\infty}\left( 
\mathbb{R}^{+},\left( 1+\epsilon\right) ^{\gamma}\right) }\left( s\right)
\leq CT\sup_{0\leq s\leq t}\left\Vert f-\tilde{f}\right\Vert _{L^{\infty
}\left( \mathbb{R}^{+},\left( 1+\epsilon\right) ^{\gamma}\right) }\left(
s\right) + \\
+C\sup_{0\leq s\leq t}\left( \int_{0}^{\infty}\left\vert f-\tilde {f}%
\right\vert \left( 1+\sqrt{\epsilon}\right) d\epsilon\right) .
\end{align*}

On the other hand, multiplying (\ref{F5E2}) by $\left( 1+\sqrt{\epsilon_{1}}%
\right) $\ and integrating we obtain:%
\begin{align*}
\sup_{0\leq s\leq t}\left( \int_{0}^{\infty}\left\vert f-\tilde{f}%
\right\vert \left( 1+\sqrt{\epsilon}\right) d\epsilon\right) \leq
CT\sup_{0\leq s\leq t}\left\Vert f-\tilde{f}\right\Vert _{L^{\infty}\left( 
\mathbb{R}^{+},\left( 1+\epsilon\right) ^{\gamma}\right) }\left( s\right) +
\\
+CT\sup_{0\leq s\leq t}\left( \int_{0}^{\infty}\left\vert f-\tilde {f}%
\right\vert \left( 1+\sqrt{\epsilon}\right) d\epsilon\right) .
\end{align*}

Then, assuming that $T$ is small we obtain:%
\begin{equation*}
\sup_{0\leq s\leq t}\left( \int_{0}^{\infty}\left\vert f-\tilde{f}%
\right\vert \left( 1+\sqrt{\epsilon}\right) d\epsilon\right) \leq
CT\sup_{0\leq s\leq t}\left\Vert f-\tilde{f}\right\Vert _{L^{\infty}\left( 
\mathbb{R}^{+},\left( 1+\epsilon\right) ^{\gamma}\right) }\left( s\right) . 
\end{equation*}

Then:%
\begin{equation*}
\sup_{0\leq s\leq t}\left\Vert f-\tilde{f}\right\Vert _{L^{\infty}\left( 
\mathbb{R}^{+},\left( 1+\epsilon\right) ^{\gamma}\right) }\left( s\right)
\leq CT\sup_{0\leq s\leq t}\left\Vert f-\tilde{f}\right\Vert _{L^{\infty
}\left( \mathbb{R}^{+},\left( 1+\epsilon\right) ^{\gamma}\right) }\left(
s\right) \, ,\ \ 0\leq t\leq T, 
\end{equation*}
and choosing $T$ small we obtain $f=\tilde{f}$ for $0\leq t\leq T.$ This
gives the uniqueness of solutions for short times. Uniqueness for
arbitrarily long times can be obtained with a similar argument using the
fact that a solution defined in an interval $\left[ 0,T\right] ,$ with $T>0,$
is also a mild solution in any interval $\left[ T^{\ast},T\right] $ with $%
0<T^{\ast}<T$ and initial datum $f\left( \cdot,T^{\ast}\right) $ at time $%
t=T^{\ast}.$
\end{proof}

In order to conclude the Proof of Theorem \ref{localExistence} it only
remains to show that the solutions can be extended as long as $\sup_{0\leq
t\leq T}\left\Vert f\left( t,\cdot\right) \right\Vert _{L^{\infty}\left( 
\mathbb{R}^{+}\right) }$ remains bounded. To this end we prove the following:

\begin{lemma}
\label{prolong}Suppose that $\gamma>3$ and $f\in L_{loc}^{\infty}\left( %
\left[ 0,T\right) ;L^{\infty}\left( \mathbb{R}^{+};\left( 1+\epsilon \right)
^{\gamma}\right) \right) $ with $T>0$ is a mild solutions of (\ref{F3E2}), (%
\ref{F3E3}) in the sense of Definition \ref{mild}. Suppose that $\sup_{0\leq
t\leq T}\left\Vert f\left( t,\cdot\right) \right\Vert _{L^{\infty}\left( 
\mathbb{R}^{+}\right) }<\infty.$ Then, it is possible to extend the solution
to a larger time interval $\left[ 0,T+\delta\right) $ for some $\delta>0.$
\end{lemma}

\begin{proof}
We recall that $f=\mathcal{T}\left( f\right) $ for $t\in(0, T)$. Using Lemma %
\ref{LinfEst} we then obtain the estimate:%
\begin{align*}
\left\Vert f\right\Vert _{L^{\infty}\left( \mathbb{R}^{+};\left(
1+\epsilon\right) ^{\gamma}\right) }\left( t\right) & \leq\left\Vert
f_{0}\right\Vert _{L^{\infty}\left( \mathbb{R}^{+};\left( 1+\epsilon\right)
^{\gamma}\right) }+t\sup_{0\leq s\leq t}\left\Vert f\left( s,\cdot\right)
\right\Vert _{L^{\infty}\left( \mathbb{R}^{+};\left( 1+\epsilon\right)
^{\gamma}\right) }\times \\
& \hskip 4.5cm \times\left( \sup_{0\leq s\leq t}\int\left( 1+\epsilon ^{\frac{3%
}{2}}\right) f\left( s,\epsilon\right) d\epsilon\right) ^{2}+ \\
&\hskip -0.7cm  +Ct\left( 1+\sup_{0\leq s\leq t}\left\Vert f\left( s,\cdot\right)
\right\Vert _{L^{\infty}\left( \mathbb{R}^{+};\left( 1+\epsilon\right)
^{\gamma}\right) }\right) \left( \sup_{0\leq s\leq t}\int_{0}^{\infty
}f\left( s,\epsilon\right) d\epsilon\right) ^{2}+\  \\
& +Ct\sup_{0\leq s\leq t}\left\Vert f\left( s,\cdot\right) \right\Vert
_{L^{\infty}\left( \mathbb{R}^{+};\left( 1+\epsilon\right) ^{\gamma }\right)
}\left( \sup_{0\leq s\leq t}\int_{0}^{\infty}\left( \epsilon \right) ^{\frac{%
3}{2}}f\left( s,\epsilon\right) d\epsilon\right) + \\
& +\theta\sup_{0\leq s\leq t}\left\Vert f\left( s,\cdot\right) \right\Vert
_{L^{\infty}\left( \mathbb{R}^{+};\left( 1+\epsilon\right) ^{\gamma }\right)
}.
\end{align*}

We now use that the energy $\int_{0}^{\infty} \epsilon^{\frac{3}{2}}f\left(
t,\epsilon\right) d\epsilon$ remains constant in time for mild solutions
(cf. Lemma \ref{massenergy}). Then, splitting the domain of integration in
the regions $\left\{ \epsilon\geq1\right\} $ and $\left\{ \epsilon<1\right\} 
$ we derive the estimate:%
\begin{equation*}
\sup_{0\leq s\leq t}\int_{0}^{\infty}f\left( s,\epsilon\right) d\epsilon
\leq\sup_{0\leq s\leq t}\left\Vert f\left( s,\cdot\right) \right\Vert
_{L^{\infty}\left( \mathbb{R}^{+}\right) }+\sup_{0\leq s\leq t}\int
_{0}^{\infty}\left( \epsilon\right) ^{\frac{3}{2}}f\left( s,\epsilon \right)
d\epsilon. 
\end{equation*}

By assumption $\sup_{0\leq s\leq t}\left\Vert f\left( s,\cdot\right)
\right\Vert _{L^{\infty}\left( \mathbb{R}^{+}\right) }\leq C$ for $0\leq
t\leq T.$ Then:%
\begin{equation*}
\sup_{0\leq s\leq t}\int_{0}^{\infty}f\left( s,\epsilon\right) d\epsilon
\leq, C 
\end{equation*}
whence:%
\begin{align*}
\left\Vert f\right\Vert _{L^{\infty}\left( \mathbb{R}^{+};\left(
1+\epsilon\right) ^{\gamma}\right) }\left( t\right) & \leq\left\Vert
f_{0}\right\Vert _{L^{\infty}\left( \mathbb{R}^{+};\left( 1+\epsilon\right)
^{\gamma}\right) }+Ct\sup_{0\leq s\leq t}\left\Vert f\left( s,\cdot\right)
\right\Vert _{L^{\infty}\left( \mathbb{R}^{+};\left( 1+\epsilon\right)
^{\gamma}\right) }+ \\
& +Ct\left( 1+\sup_{0\leq s\leq t}\left\Vert f\left( s,\cdot\right)
\right\Vert _{L^{\infty}\left( \mathbb{R}^{+};\left( 1+\epsilon\right)
^{\gamma}\right) }\right) +\  \\
& +Ct\sup_{0\leq s\leq t}\left\Vert f\left( s,\cdot\right) \right\Vert
_{L^{\infty}\left( \mathbb{R}^{+};\left( 1+\epsilon\right) ^{\gamma }\right)
}+\theta\sup_{0\leq s\leq t}\left\Vert f\left( s,\cdot\right) \right\Vert
_{L^{\infty}\left( \mathbb{R}^{+};\left( 1+\epsilon\right) ^{\gamma}\right)
}.
\end{align*}

Since $\theta<1$ it then follows that there exists $t^{\ast}>0$ such that:%
\begin{equation*}
\sup_{0\leq s\leq t^{\ast}}\left\Vert f\left( s,\cdot\right) \right\Vert
_{L^{\infty}\left( \mathbb{R}^{+};\left( 1+\epsilon\right) ^{\gamma }\right)
}\leq\left( 1+a\right) \left\Vert f_{0}\right\Vert _{L^{\infty }\left( 
\mathbb{R}^{+};\left( 1+\epsilon\right) ^{\gamma}\right) }+Ct^{\ast}, 
\end{equation*}
for some $a>0.$ This estimate can be iterated starting at $t=t^{\ast}.$ It
then follows, after a number of iterations that:%
\begin{equation*}
\sup_{0\leq s\leq T}\left\Vert f\left( s,\cdot\right) \right\Vert
_{L^{\infty}\left( \mathbb{R}^{+};\left( 1+\epsilon\right) ^{\gamma }\right)
}\leq C. 
\end{equation*}

Then, applying Proposition \ref{LE} we obtain that it is possible to extend
the solution to a larger time interval $\left[ 0,T+\delta\right) .$ Indeed,
suppose that $f$ is defined as $f=f^{\left( 1\right) }$ for $0\leq t\leq
t^{\ast}$ and $f=f^{\left( 2\right) }$ for $t^{\ast}\leq t\leq t^{\ast\ast}$
where $f^{\left( 1\right) }$ is a mild solution of (\ref{F3E2}), (\ref{F3E3}%
) with initial data $f_{0}$ in the interval $0\leq t\leq t^{\ast}$ and $%
f_{2} $ is a mild solution of (\ref{F3E2}), (\ref{F3E3}) defined for $%
t^{\ast}\leq t\leq t^{\ast\ast}$ such that $f^{\left( 2\right) }=f^{\left(
1\right) }$ for $t=t^{\ast}.$ It follows from (\ref{F3E6}) that $f$ is a
mild solution of (\ref{F3E2}), (\ref{F3E3}) in the interval $0\leq t\leq
t^{\ast\ast}.$
\end{proof}

\begin{proof}[Proof of Theorem \protect\ref{localExistence}]
It is a consequence of Proposition \ref{LE}, Lemma \ref{massenergy},
Proposition \ref{uniq} and Lemma \ref{prolong}.
\end{proof}

\section{Monotonicity properties of the kernel $Q_{3}\left[ f\right] $.}

\setcounter{equation}{0} \setcounter{theorem}{0}

In this Section we recall a crucial monotonicity property of the kernel $%
Q_{3}\left[ f\right] $ that captures in a precise way the fact that the
cubic terms in (\ref{F3E2}) have some tendency to yield concentration of $%
f\left( \epsilon,t\right) $ to concentrate towards regions with smaller
values of $\epsilon.$ This property has been obtained in \cite{Lu3}.

\begin{proposition}
\label{atractiveness} Let $q_{3}\left( \cdot\right) $ as in (\ref{Q1E1}).
Let us denote as $\mathcal{S}^{3}$ the group of permutations of the three
elements $\left\{ 1,2,3\right\} .$ Suppose that $\varphi\in C\left( \mathbb{R%
}^{+}\right) $ is a test function. The following identity holds for any for any $f$ such that $h=\sqrt{\epsilon}f\left(\epsilon\right)\in \mathcal{M}
_{+}\left( \mathbb{R}^{+} \right)$:
\begin{equation}
\int_{\left( \mathbb{R}^{+}\right) ^{3}}d\mathcal{\epsilon}_{1}d\mathcal{%
\epsilon}_{3}d\mathcal{\epsilon}_{4}\Phi \ q_{3}\left( f\right) \left(
\epsilon_{1}\right) \sqrt{\epsilon_{1}}\varphi\left( \epsilon_{1}\right)
=\int_{\left( \mathbb{R}^{+}\right) ^{3}}d\mathcal{\epsilon}_{1}d\mathcal{%
\epsilon}_{2}d\mathcal{\epsilon}_{3}\,f_{1}\,f_{2}\,f_{3}\mathcal{G}%
_{\varphi}(\mathcal{\epsilon}_{1}\,\mathcal{\epsilon}_{2}\,\mathcal{\epsilon}%
_{3}),   \label{S1E12a}
\end{equation}
where:%
\begin{align}
\mathcal{G}_{\varphi}\left( \mathcal{\epsilon}_{1},\mathcal{\epsilon}_{2},%
\mathcal{\epsilon}_{3}\right) & =\frac{1}{6}\sum_{\sigma\in \mathcal{S}%
^{3}}H_{\varphi}\left( \mathcal{\epsilon}_{\sigma(1)},\mathcal{\epsilon}%
_{\sigma(2)},\mathcal{\epsilon}_{\sigma(3)}\right) \Phi\left( \mathcal{%
\epsilon}_{\sigma(1)},\mathcal{\epsilon}_{\sigma (2)};\mathcal{\epsilon}%
_{\sigma(3)}\right) ,  \label{S1E12bis} \\
& H_{\varphi}(x,y,z)=\varphi\left( z\right) +\varphi\left( x+y-z\right)
-\varphi\left( x\right) -\varphi\left( y\right) ,   \label{S1E12ter}
\end{align}
with $\Phi$ as in (\ref{F3E5}) and:%
\begin{equation}
\mathcal{G}_{\varphi}\left( \mathcal{\epsilon}_{1},\mathcal{\epsilon}_{2},%
\mathcal{\epsilon}_{3}\right) =\mathcal{G}_{\varphi}\left( \mathcal{\epsilon}%
_{\sigma(1)},\mathcal{\epsilon}_{\sigma(2)},\mathcal{\epsilon}%
_{\sigma(3)}\right) \ \ \ \text{for any\ }\sigma \in\mathcal{S}^{3}. 
\label{S1E12four}
\end{equation}

Moreover, if the function $\varphi$ is convex we have $\mathcal{G}_{\varphi
}\left( \mathcal{\epsilon}_{1},\mathcal{\epsilon}_{2},\mathcal{\epsilon}%
_{3}\right) \geq0$ and if $\varphi$ is concave we have $\mathcal{G}_{\varphi
}\left( \mathcal{\epsilon}_{1},\mathcal{\epsilon}_{2},\mathcal{\epsilon}%
_{3}\right) \leq0.$ For any test function $\varphi$ the function $\mathcal{G}%
_{\varphi}\left( \mathcal{\epsilon}_{1},\mathcal{\epsilon}_{2},\mathcal{%
\epsilon}_{3}\right) $ vanishes along the diagonal $\left\{ \left( \mathcal{%
\epsilon}_{1},\mathcal{\epsilon}_{2},\mathcal{\epsilon}_{3}\right) \in\left( 
\mathbb{R}^{+}\right) ^{3}:\mathcal{\epsilon}_{1}=\mathcal{\epsilon}_{2}=%
\mathcal{\epsilon}_{3}\right\} .$
\end{proposition}

\begin{remark}
The interpretation of this Theorem is simple if we think the process in
terms of particles whose dynamics is driven by means of the collision terms $%
Q_{3}\left[ f\right] .$ Notice that such dynamics can be thought as a
classical dynamics in which given three particles two of them are selected
as incoming particles and the last one is one of the outgoing particles. The
energy is the fourth one is then determined by means of the conservation of
energy. Given three particles with energies $\mathcal{\epsilon}_{1},\mathcal{%
\epsilon}_{2},\mathcal{\epsilon}_{3}$ we consider all the processes in which
they can be involved, either as initial or final particles. The
probabilities of these processes depend on the specific choice made of
incoming and outgoing particles. We then compute the average change of $%
\Delta=\sum_{k=1}^{3}\varphi\left( \mathcal{\epsilon}_{k}\right) $ in these
processes. If $\varphi$ is concave the change of $\Delta$ is nonnegative,
and if $\varphi$ is convex, such a change is nonpositive. If we take, for
instance the convex function $\varphi\left( \epsilon\right) =\epsilon^{-r}$
with $r>0,$ the monotonicity property states that particles tend to move on
average towards smallest values of $\epsilon.$
\end{remark}

\begin{remark}
Notice that $\mathcal{G}_{\varphi}\left( \mathcal{\epsilon}_{1},\mathcal{%
\epsilon}_{2},\mathcal{\epsilon}_{3}\right) =0$ if $\varphi=1$ or $%
\varphi=\epsilon.$ This could be expected due to the fact that the kinetic
equation (or the particle interpretation of this process) formally conserves
the number of particles and the energy. Moreover, we have $\mathcal{G}%
_{\varphi}\left( \mathcal{\epsilon}_{1},\mathcal{\epsilon}_{1},\mathcal{%
\epsilon}_{1}\right) =0.$ The meaning of this identity is that the
distribution of particles is not modified by the cubic terms of the equation
if there is only one type of them. Notice that this implies also that Dirac
masses $g\left( \epsilon\right) =\delta_{\epsilon=\epsilon^{\ast}}$ with $%
\epsilon^{\ast}>0$ are stationary solutions of the kinetic equation
containing only cubic terms (cf. (\ref{St3})). This stationarity is the
source of many of the technical difficulties in the forthcoming analysis.
\end{remark}

\begin{definition}
\label{aux}We will use repeatedly the auxiliary functions $\epsilon _{+},\
\epsilon_{0},\ \epsilon_{-}$ defined from $\mathbb{R}^{+}\times\mathbb{R}%
^{+}\times\mathbb{R}^{+}$ to $\mathbb{R}^{+}$ as follows:%
\begin{align*}
\epsilon_{+}\left( \mathcal{\epsilon}_{1},\mathcal{\epsilon}_{2},\mathcal{%
\epsilon}_{3}\right) & =\max\left\{ \mathcal{\epsilon}_{1},\mathcal{\epsilon}%
_{2},\mathcal{\epsilon}_{3}\right\} , \\
\epsilon_{-}\left( \mathcal{\epsilon}_{1},\mathcal{\epsilon}_{2},\mathcal{%
\epsilon}_{3}\right) & =\min\left\{ \mathcal{\epsilon}_{1},\mathcal{\epsilon}%
_{2},\mathcal{\epsilon}_{3}\right\} , \\
\epsilon_{0}\left( \mathcal{\epsilon}_{1},\mathcal{\epsilon}_{2},\mathcal{%
\epsilon}_{3}\right) & =\mathcal{\epsilon}_{k}\in\left\{ \mathcal{\epsilon}%
_{1},\mathcal{\epsilon}_{2},\mathcal{\epsilon}_{3}\right\} \text{ such that }%
\epsilon_{-}\left( \mathcal{\epsilon}_{1},\mathcal{\epsilon }_{2},\mathcal{%
\epsilon}_{3}\right) \leq\mathcal{\epsilon}_{k}\leq \epsilon_{+}\left( 
\mathcal{\epsilon}_{1},\mathcal{\epsilon}_{2},\mathcal{\epsilon}_{3}\right) ,
\end{align*}
with $k\in\left\{ 1,2,3\right\} .$
\end{definition}

\begin{proof}[Proof of Proposition \protect\ref{atractiveness}]
Notice that (\ref{S1E12a})-(\ref{S1E12ter}) are just a consequence of the
identity 
\begin{equation}
\int_{{\mathbb{R}}^{+}}Q_{3}\left[ f\right] \left( \epsilon_{1}\right) \sqrt{%
\epsilon_{1}}\varphi_{1}d\epsilon_{1}=\frac{1}{2}\int_{\left( {\mathbb{R}}%
^{+}\right) ^{3}}d\epsilon_{1}d\epsilon_{2}d\epsilon_{3}\Phi
f_{1}f_{2}f_{3}\left( \varphi_{3}+\varphi_{4}-\varphi_{1}-\varphi_{2}\right)
,   \label{Qtest1}
\end{equation}
combined with a symmetrization argument. In order to prove that $\mathcal{G}%
_{\varphi}\left( \mathcal{\epsilon}_{1},\mathcal{\epsilon}_{2},\mathcal{%
\epsilon}_{3}\right) $ has the indicated signs for convex or concave
functions $\varphi,$ we use the fact that the symmetry of $\mathcal{G}%
_{\varphi}$ under perturbations yields:%
\begin{equation}
\mathcal{G}_{\varphi}\left( \mathcal{\epsilon}_{1},\mathcal{\epsilon}_{2},%
\mathcal{\epsilon}_{3}\right) =\mathcal{G}_{\varphi}\left( \epsilon
_{+},\epsilon_{-},\epsilon_{0}\right) .   \label{sym}
\end{equation}

Using (\ref{F3E5}) and Definition \ref{aux} we obtain:%
\begin{equation}
\Phi\left( \epsilon_{+},\epsilon_{-};\epsilon_{0}\right) =\Phi\left(
\epsilon_{0},\epsilon_{+};\epsilon_{-}\right) =\sqrt{\mathcal{\epsilon}_{-}}%
\ \ ,\ \ \ \Phi\left( \epsilon_{0},\epsilon_{-};\epsilon_{+}\right) =\sqrt{%
\left( \mathcal{\epsilon}_{0}+\mathcal{\epsilon}_{-}-\mathcal{\epsilon }%
_{+}\right) _{+}}.   \label{G1E1}
\end{equation}
We have also the symmetry property $\Phi\left( \epsilon_{j},\epsilon_{\ell
};\epsilon_{k}\right) =\Phi\left( \epsilon_{\ell},\epsilon_{j};\epsilon
_{k}\right) \ \ ,\ \ \ j,\ell,k\in\left\{ 1,2,3\right\} .$ Then, using (\ref%
{S1E12bis}) we obtain: 
\begin{align}
\mathcal{G}_{\varphi}\left( \epsilon_{+},\epsilon_{-},\epsilon_{0}\right) =%
\frac{1}{3}\left[ H_{\varphi}\left( \epsilon_{+},\epsilon_{-};\epsilon
_{0}\right) \sqrt{\mathcal{\epsilon}_{-}}+H_{\varphi}\left( \epsilon
_{0},\epsilon_{+};\epsilon_{-}\right) \sqrt{\mathcal{\epsilon}_{-}}+\right. \notag\\
\left.+H_{\varphi}\left( \epsilon_{0},\epsilon_{-};\epsilon_{+}\right) \sqrt{%
\left( \mathcal{\epsilon}_{0}+\mathcal{\epsilon}_{-}-\mathcal{\epsilon }%
_{+}\right) _{+}}\right] ,   \label{G1E2}
\end{align}
and using (\ref{S1E12ter}):%
\begin{align}
\mathcal{G}_{\varphi}\left( \epsilon_{+},\epsilon_{-},\epsilon_{0}\right) & =%
\frac{1}{3}\left[ \sqrt{\mathcal{\epsilon}_{-}}\left[ \varphi\left(
\epsilon_{+}+\epsilon_{-}-\epsilon_{0}\right) +\varphi\left( \epsilon
_{+}+\epsilon_{0}-\epsilon_{-}\right) -2\varphi\left( \epsilon_{+}\right) %
\right] \right. +  \notag \\
& +\left. \sqrt{\left( \mathcal{\epsilon}_{0}+\mathcal{\epsilon}_{-}-%
\mathcal{\epsilon}_{+}\right) _{+}}\left[ \varphi\left( \epsilon _{+}\right)
+\varphi\left( \epsilon_{0}+\epsilon_{-}-\epsilon_{+}\right) -\varphi\left(
\epsilon_{0}\right) -\varphi\left( \epsilon_{-}\right) \right] \right] . 
\label{G1E3}
\end{align}

Suppose now that $\varphi=\varphi\left( \epsilon\right) $ is a convex
function for $\epsilon>0.$ Then:%
\begin{equation}
\frac{1}{2}\left[ \varphi\left( \epsilon+z\right) +\varphi\left(
\epsilon-z\right) \right] \geq\varphi\left( \epsilon\right) \ \ \ ,\ \ \
\epsilon>0,\ \ \ z\geq0,\ \ \epsilon-z>0.   \label{G1E4}
\end{equation}

On the other hand we can prove the following property for convex functions.
Suppose that $\psi$ is a convex function in $\epsilon>0.$ Then for any $%
0<\epsilon_{1}\leq\epsilon_{2}\leq\epsilon_{3}\leq\epsilon_{4}$ satisfying $%
\epsilon_{1}+\epsilon_{4}=\epsilon_{2}+\epsilon_{3}$ we have:%
\begin{equation}
\psi\left( \epsilon_{1}\right) +\psi\left( \epsilon_{4}\right) \geq
\psi\left( \epsilon_{2}\right) +\psi\left( \epsilon_{3}\right) . 
\label{G1E5}
\end{equation}

To prove (\ref{G1E5}) we define the function $W\left( z\right) =\psi\left( 
\frac{\epsilon_{1}+\epsilon_{4}}{2}+z\right) +\psi\left( \frac{\epsilon
_{1}+\epsilon_{4}}{2}-z\right) $ for $z\geq0,\ \frac{\epsilon_{1}+%
\epsilon_{4}}{2}-z>0.$ If $\psi\in C^{2}$ we would have:%
\begin{align*}
W^{\prime}\left( 0\right) & =\psi^{\prime}\left( \frac{\epsilon
_{1}+\epsilon_{4}}{2}\right) -\psi^{\prime}\left( \frac{\epsilon
_{1}+\epsilon_{4}}{2}\right) =0, \\
W^{\prime\prime}\left( z\right) & =\psi^{\prime\prime}\left( \frac{%
\epsilon_{1}+\epsilon_{4}}{2}+z\right) +\psi^{\prime\prime}\left( \frac{%
\epsilon_{1}+\epsilon_{4}}{2}-z\right) \geq0.
\end{align*}

It then follows that $W^{\prime}\left( z\right) \geq0$ if $z\geq0,$ whence:%
\begin{equation*}
W\left( z_{2}\right) \geq W\left( z_{1}\right) \ \ \text{if\ \ \thinspace
\thinspace\thinspace}0\leq z_{1}<z_{2}. 
\end{equation*}

Choosing $z_{1}=\epsilon_{3}-\frac{\epsilon_{1}+\epsilon_{4}}{2}=\epsilon
_{3}-\frac{\epsilon_{2}+\epsilon_{3}}{2}$ and $z_{2}=\epsilon_{4}-\frac{%
\epsilon_{1}+\epsilon_{4}}{2}$ we obtain (\ref{G1E5}). If $\psi$ does not
have two derivatives, the result can be proved extending $\psi$ as a linear
function for negative values, convolving the resulting function with a
mollifier and passing to the limit in the desired identity.

Using (\ref{G1E5}) with $f=\varphi$ and $\epsilon_{1}=\epsilon_{+}+%
\epsilon_{-}-\epsilon_{0},\ \epsilon_{2}=\epsilon_{-},\
\epsilon_{3}=\epsilon_{0},\ \epsilon_{4}=\epsilon_{+}$ we obtain:%
\begin{equation}
\varphi\left( \epsilon_{+}\right) +\varphi\left( \epsilon_{0}+\epsilon
_{-}-\epsilon_{+}\right) -\varphi\left( \epsilon_{0}\right) -\varphi\left(
\epsilon_{-}\right) \geq0.   \label{G1E6}
\end{equation}

Plugging (\ref{G1E4}), (\ref{G1E6}) into (\ref{G1E3}) we obtain $\mathcal{G}%
_{\varphi}\left( \epsilon_{+},\epsilon_{-},\epsilon_{0}\right) \geq0$ for
any convex function $\varphi.$ On the other hand, a similar argument shows
that $\mathcal{G}_{\varphi}\left( \epsilon_{+},\epsilon_{-},\epsilon
_{0}\right) \leq0$ for any concave function $\varphi.$ This concludes the
proof.
\end{proof}

\section{Estimating the number of collisions between small particles.}

\setcounter{equation}{0} \setcounter{theorem}{0} The main goal of this
Section is to derive an estimate for the number of vectors $\left( \epsilon
_{1},\epsilon_{2},\epsilon_{3}\right) \in\left[ 0,R\right] ^{3}$, that we
call triples, that are sufficiently separated from the diagonal $\left\{
\epsilon_{1}=\epsilon_{2}=\epsilon_{3}\right\} $ for $R\leq\frac{1}{2}$ (cf.
(\ref{H1})). The first step is to derive a precise estimate for the number
of "triple collisions" taking place in the system.

\begin{proposition}
\label{propositionfiveone}
Suppose that $f$ is a weak solution of (\ref{F3E2}), (\ref{F3E3}) on $(0, T)$ in the sense of Definition \ref{weakf}, with initial data $f_0$,  and let $g$ be as in (\ref{F3E3a}).
Then, there exists a numerical constant $B>0,$ independent on $f_{0}$ and $T,
$ such that, for any $R\in\left( 0,1\right) $ we have: 
\begin{align}
B\int_{0}^{T}dt\int_{\left[ 0,\frac{R}{2}\right] ^{3}}\left[ \prod
_{m=1}^{3}\,g_{m}d\mathcal{\epsilon}_{m}\right] \frac{\left( \epsilon
_{0}\right) ^{\frac{3}{2}}}{\left( \epsilon_{+}\right) ^{\frac{3}{2}}}\left( 
\frac{\epsilon_{0}-\epsilon_{-}}{\epsilon_{0}}\right) ^{2}
&\leq2\pi R^{\frac{3%
}{2}}\int_{0}^{T}dt\left( \int_{\left[ 0,R\right] }g\left( \epsilon\right)
d\epsilon\right) ^{2}+\notag\\
&\hskip 3cm +MR,   \label{W1E2}
\end{align}
where $M$ is as in (\ref{C1}) and the functions $\epsilon_{-},\ \epsilon
_{0},\ \epsilon_{+}$ are as in Definition \ref{aux}.
\end{proposition}

\begin{proof}
We use (\ref{F5E1a}), (\ref{F5E1b}), (\ref{F5E1c}) with test function $%
\varphi\left( \epsilon\right) =\psi\left( \frac{\epsilon}{R}\right), 
R>0, \epsilon>0$,  where:%
\begin{equation*}
\psi\left( s\right) =\left\{ 
\begin{array}{c}
s^{\theta}\ \ \ \ ,\ \ \ \ 0<s<1 \\ 
\hskip -0.57cm 1\ \ \ \ \ ,\ \ \ \ \ s\geq1%
\end{array}
\right\} \,,\ \ 0<\theta<1. 
\end{equation*}
Let $\mathcal{G}_{\varphi}\left( \mathcal{\epsilon}_{1},\mathcal{\epsilon }%
_{2},\mathcal{\epsilon}_{3}\right) $ as in (\ref{S1E12bis}). Then, since the
function $\varphi$ is concave, Proposition \ref{atractiveness} implies that $%
\mathcal{G}_{\varphi}\left( \mathcal{\epsilon}_{1},\mathcal{\epsilon}_{2},%
\mathcal{\epsilon}_{3}\right) \leq0.$ Using (\ref{S1E12a}) we then obtain:%
\begin{align*}
& \int_{\mathbb{R}^{+}}Q_{3}\left[ f\right] \left( \epsilon_{1}\right) \sqrt{%
\epsilon_{1}}\varphi\left( \epsilon_{1}\right) d\mathcal{\epsilon}_{1} \\
& =\int_{\left[ 0,\frac{R}{2}\right] ^{3}}d\mathcal{\epsilon}_{1}d\mathcal{%
\epsilon}_{2}d\mathcal{\epsilon}_{3}\,f_{1}\,f_{2}\,f_{3}\mathcal{G}%
_{\varphi}(\mathcal{\epsilon}_{1}\,\mathcal{\epsilon}_{2}\,\mathcal{\epsilon}%
_{3})+\int_{\left( \mathbb{R}^{+}\right) ^{3}\setminus\left[ 0,\frac{R}{2}%
\right] ^{3}}d\mathcal{\epsilon}_{1}d\mathcal{\epsilon}_{2}d\mathcal{\epsilon%
}_{3}\,f_{1}\,f_{2}\,f_{3}\mathcal{G}_{\varphi}(\mathcal{\epsilon}_{1}\,%
\mathcal{\epsilon}_{2}\,\mathcal{\epsilon}_{3}) \\
& \leq\int_{\left[ 0,\frac{R}{2}\right] ^{3}}d\mathcal{\epsilon}_{1}d%
\mathcal{\epsilon}_{2}d\mathcal{\epsilon}_{3}\,f_{1}\,f_{2}\,f_{3}\mathcal{G}%
_{\varphi}(\mathcal{\epsilon}_{1}\,\mathcal{\epsilon}_{2}\,\mathcal{\epsilon}%
_{3}).
\end{align*}

Using Definition \ref{aux} and (\ref{sym}) we obtain: 
\begin{equation}
\int_{\mathbb{R}^{+}}Q_{3}\left[ f\right] \left( \epsilon_{1}\right) \sqrt{%
\epsilon_{1}}\varphi\left( \epsilon_{1}\right) d\mathcal{\epsilon}%
_{1}\leq\int_{\left[ 0,\frac{R}{2}\right] ^{3}}d\mathcal{\epsilon}_{1}d%
\mathcal{\epsilon}_{2}d\mathcal{\epsilon}_{3}\,f_{1}\,f_{2}\,f_{3}\mathcal{G}%
_{\varphi}(\epsilon_{+},\epsilon_{-},\epsilon_{0}).   \label{G1E7}
\end{equation}

Using (\ref{G1E3}) we can compute $\mathcal{G}_{\varphi}(\epsilon_{+},%
\epsilon_{-},\epsilon_{0})$ for $(\mathcal{\epsilon}_{1}\,\mathcal{\epsilon }%
_{2}\,\mathcal{\epsilon}_{3})\in\left[ 0,\frac{R}{2}\right] ^{3}$:%
\begin{align*}
\mathcal{G}_{\varphi}\left( \epsilon_{+},\epsilon_{-},\epsilon_{0}\right) & =%
\frac{1}{3}\left[ \sqrt{\mathcal{\epsilon}_{-}}\left[ \left( \frac{%
\epsilon_{+}+\epsilon_{-}-\epsilon_{0}}{R}\right) ^{\theta}+\left( \frac{%
\epsilon_{+}+\epsilon_{0}-\epsilon_{-}}{R}\right) ^{\theta}-2\left( \frac{%
\epsilon_{+}}{R}\right) ^{\theta}\right] \right. + \\
&\hskip -1cm  +\left. \sqrt{\left( \mathcal{\epsilon}_{0}+\mathcal{\epsilon}_{-}-%
\mathcal{\epsilon}_{+}\right) _{+}}\left[ \left( \frac{\epsilon_{+}}{R}%
\right) ^{\theta}+\left( \frac{\epsilon_{0}+\epsilon_{-}-\epsilon_{+}}{R}%
\right) ^{\theta}-\left( \frac{\epsilon_{0}}{R}\right) ^{\theta }-\left( 
\frac{\epsilon_{-}}{R}\right) ^{\theta}\right] \right] . 
\end{align*}

Integrating (\ref{F5E1a}) and using (\ref{G1E7}) as well as the
nonnegativity of $\varphi$ we deduce:%
\begin{align}
&\frac{-1}{2^{\frac{5}{2}}}\int_{0}^{T}dt\int_{\left[ 0,\frac{R}{2}\right]
^{3}}d\mathcal{\epsilon}_{1}d\mathcal{\epsilon}_{2}d\mathcal{\epsilon}%
_{3}\,g_{1}\,g_{2}\,g_{3}\frac{\mathcal{G}_{\varphi}(\epsilon_{+},\epsilon
_{-},\epsilon_{0})}{\sqrt{\epsilon_{+}\epsilon_{-}\epsilon_{0}}}\leq\notag \\
&\hskip 3cm \leq\frac {%
\pi}{2}\int_{0}^{T}dt\int_{\left( {\mathbb{R}}^{+}\right) ^{3}}\frac {%
g_{1}g_{2}\Phi}{\sqrt{\epsilon_{1}\epsilon_{2}}}Q_{\varphi}d\epsilon
_{1}d\epsilon_{2}d\epsilon_{3}
+\int_{{\mathbb{R}}^{+}}g(\epsilon_{1},0)\varphi_{1}d\epsilon_{1}, 
\label{G1E9}
\end{align}
where $T>0$ is otherwise arbitrary. Notice that, since $\varphi\leq1$ we
have:%
\begin{equation}
\int_{{\mathbb{R}}^{+}}g(\epsilon_{1},0)\varphi_{1}d\epsilon_{1}\leq M. 
\label{G2E3}
\end{equation}

We now estimate the first term on the right-hand side of (\ref{G1E9}) as
follows. We split the integral as:%
\begin{equation*}
\int_{\left( {\mathbb{R}}^{+}\right) ^{3}}\frac{g_{1}g_{2}\Phi}{\sqrt{%
\epsilon_{1}\epsilon_{2}}}Q_{\varphi}d\epsilon_{1}d\epsilon
_{2}d\epsilon_{3}=\int_{\left( {\mathbb{R}}^{+}\right) ^{3}\setminus\left[
0,R\right] ^{3}}\left[ \cdot\cdot\cdot\right] +\int_{\left[ 0,R\right] ^{3}}%
\left[ \cdot\cdot\cdot\right] . 
\end{equation*}

If $\left( \epsilon_{1},\epsilon_{2},\epsilon_{3}\right) \in\left( {\mathbb{R%
}}^{+}\right) ^{3}\setminus\left[ 0,R\right] ^{3}$ we have $\left(
\varphi_{3}+\varphi_{4}-\varphi_{1}-\varphi_{2}\right) =\left(
1+\varphi_{4}-1-1\right) =\left( \varphi_{4}-1\right) \leq0,$ whence $%
\int_{\left( {\mathbb{R}}^{+}\right) ^{3}\setminus\left[ 0,R\right] ^{3}}%
\left[ \cdot\cdot\cdot\right] \leq0.$ Therefore, using that $\varphi\leq1,$
as well as the fact that $\int_{\left[ 0,R\right] }\Phi d\epsilon_{3}\leq2%
\sqrt{R}\min\left\{ \sqrt{\epsilon_{1}},\sqrt{\epsilon_{2}}\right\}
\max\left\{ \sqrt{\epsilon_{1}},\sqrt{\epsilon_{2}}\right\} $\ if $\left(
\epsilon_{1},\epsilon_{2}\right) \in\left[ 0,R\right] ^{2}:$%
\begin{align}
\int_{\left( {\mathbb{R}}^{+}\right) ^{3}}\frac{g_{1}g_{2}\Phi}{\sqrt{%
\epsilon_{1}\epsilon_{2}}}Q_{\varphi}d\epsilon_{1}d\epsilon
_{2}d\epsilon_{3} & \leq2\int_{\left[ 0,R\right] ^{3}}\frac{g_{1}g_{2}\Phi}{%
\sqrt{\epsilon_{1}\epsilon_{2}}}d\epsilon_{1}d\epsilon_{2}d\epsilon _{3} 
\notag \\
& \hskip -1cm \leq4\sqrt{R}\int_{\left[ 0,R\right] ^{2}}\frac{g_{1}g_{2}}{\sqrt{%
\epsilon_{1}\epsilon_{2}}}\min\left\{ \sqrt{\epsilon_{1}},\sqrt{\epsilon_{2}}%
\right\} \max\left\{ \sqrt{\epsilon_{1}},\sqrt {\epsilon_{2}}\right\}
d\epsilon_{1}d\epsilon_{2},
\end{align}
whence:%
\begin{equation}
\int_{\left( {\mathbb{R}}^{+}\right) ^{3}}\frac{g_{1}g_{2}\Phi}{\sqrt{%
\epsilon_{1}\epsilon_{2}}}Q_{\varphi}d\epsilon_{1}d\epsilon
_{2}d\epsilon_{3}\leq4\sqrt{R}\int_{\left[ 0,R\right] ^{2}}g_{1}g_{2}d%
\epsilon_{1}d\epsilon_{2}=4\sqrt{R}\left( \int_{\left[ 0,R\right]
}gd\epsilon\right) ^{2}.   \label{G2E4}
\end{equation}

Plugging (\ref{G2E3}), (\ref{G2E4}) into (\ref{G1E9}) we obtain:%
\begin{equation}
-\frac{1}{2^{\frac{5}{2}}}\int_{0}^{T}dt\int_{\left[ 0,\frac{R}{2}\right]
^{3}}\,g_{1}\,g_{2}\,g_{3}\frac{\mathcal{G}_{\varphi}(\epsilon_{+},%
\epsilon_{-},\epsilon_{0})}{\sqrt{\epsilon_{+}\epsilon_{-}\epsilon_{0}}}d%
\mathcal{\epsilon}_{1}d\mathcal{\epsilon}_{2}d\mathcal{\epsilon}_{3}\leq 2\pi%
\sqrt{R}\left( \int_{\left[ 0,R\right] }gd\epsilon\right) ^{2}+M. 
\label{G2E5}
\end{equation}

In order to derive a lower estimate of $\mathcal{G}_{\varphi}\left(
\epsilon_{+},\epsilon_{-},\epsilon_{0}\right) $ we need some calculus
inequalities. To this end we define:%
\begin{equation*}
\sigma\left( X_{-},X_{0},X_{+}\right) =\left( X_{+}+X_{-}-X_{0}\right)
^{\theta}+\left( X_{+}+X_{0}-X_{-}\right) ^{\theta}-2\left( X_{+}\right)
^{\theta}, 
\end{equation*}
with $0\leq X_{-}\leq X_{0}\leq X_{+}$. Then, we write 
\begin{equation*}
\sigma\left( X_{-},X_{0},X_{+}\right) =\left( X_{+}\right)
^{\theta}\sigma\left( \frac{X_{-}}{X_{+}},\frac{X_{0}}{X_{+}},1\right)
=\left( X_{+}\right) ^{\theta}\sigma\left( 0,\frac{X_{0}-X_{-}}{X_{+}}%
,1\right) . 
\end{equation*}
The function $\sigma\left( 0,Z,1\right) $ is decreasing on $Z$ if $Z>0.$
Suppose that $\left( X_{0}-X_{-}\right) >\frac{X_{+}}{2}.$ Then $%
\sigma\left( 0,\frac{X_{0}-X_{-}}{X_{+}},1\right) \leq\sigma\left( 0,\frac{1%
}{2},1\right) .$ The same convexity argument that was used in the Proof of (%
\ref{G1E4}) yields $\sigma\left( 0,\frac{1}{2},1\right) <0.$ Then\ $%
\sigma\left( \frac{X_{-}}{X_{+}},\frac{X_{0}}{X_{+}},1\right)
\leq-A_{1,\theta}$ for some $A_{1,\theta}>0$ if $\left( X_{0}-X_{-}\right) >%
\frac{X_{+}}{2}.$ Suppose now that $X_{-}\leq X_{0},\ \left(
X_{0}-X_{-}\right) \leq\frac{X_{+}}{2}$ we can use Taylor's Theorem to
obtain $\sigma\left( \frac{X_{-}}{X_{+}},\frac{X_{0}}{X_{+}},1\right)
\leq-A_{2,\theta}\left( \frac{X_{0}-X_{-}}{X_{+}}\right) ^{2},$ with $%
A_{2,\theta}>0.$ Then: 
\begin{equation*}
\sigma\left( X_{-},X_{0},X_{+}\right) \leq-A_{\theta}\left( X_{+}\right)
^{\theta}\left( \frac{X_{0}-X_{-}}{X_{+}}\right) ^{2}\ ,\ 0\leq X_{-}\leq
X_{0}\leq X_{+}, 
\end{equation*}
with $A_{\theta}>0.$ Since $\left( \epsilon_{1},\epsilon_{2},\epsilon
_{3}\right) \in\left[ 0,\frac{R}{2}\right] ^{3}$ we have that the function $%
\varphi\left( s\right) $ is evaluated by means of $\left( \frac{s}{R}\right)
^{\theta}$ in all the terms we then have:%
\begin{equation*}
\mathcal{G}_{\varphi}\left( \epsilon_{+},\epsilon_{-},\epsilon_{0}\right)
\leq-\frac{A_{\theta}}{3}\sqrt{\mathcal{\epsilon}_{-}}\left( \frac {%
\epsilon_{+}}{R}\right) ^{\theta}\left( \frac{\epsilon_{0}-\epsilon_{-}}{%
\epsilon_{+}}\right) ^{2}. 
\end{equation*}

We can assume by definiteness that $\theta=\frac{1}{2}.$ Using (\ref{G2E5})
we obtain:%
\begin{align*}
&B\int_{0}^{T}dt\int_{\left[ 0,\frac{R}{2}\right] ^{3}}d\mathcal{\epsilon }%
_{1}d\mathcal{\epsilon}_{2}d\mathcal{\epsilon}_{3}\left[ \prod_{k=1}^{3}%
\,g_{k}\right] \frac{1}{\sqrt{\epsilon_{+}\epsilon_{0}}}\left( \frac{%
\epsilon_{+}}{R}\right) ^{\theta}\left( \frac{\epsilon_{0}-\epsilon_{-}}{%
\epsilon_{+}}\right) ^{2}\leq \notag \\
&\hskip 6cm \leq2\pi\sqrt{R}\int_{0}^{T}dt\left( \int_{\left[
0,R\right] }g\left( \epsilon\right) d\epsilon\right) ^{2}
+M,   \label{G2E6}
\end{align*}
where $B>0$ is independent on $R,\ g_{0}$ and $T.$ After some computations
we arrive at:%
\begin{align*}
& B\int_{0}^{T}dt\int_{\left[ 0,\frac{R}{2}\right] ^{3}}d\mathcal{\epsilon }%
_{1}d\mathcal{\epsilon}_{2}d\mathcal{\epsilon}_{3}\left[ \prod_{k=1}^{3}%
\,g_{k}\right] \frac{\left( \epsilon_{+}\right) ^{\theta-1}\left(
\epsilon_{0}\right) ^{\frac{3}{2}}}{\left( \epsilon_{+}\right) ^{\frac {3}{2}%
}}\left( \frac{\epsilon_{0}-\epsilon_{-}}{\epsilon_{0}}\right) ^{2}\leq \\
& \hskip 6cm \leq2\pi\sqrt{R}R^{\theta}\int_{0}^{T}dt\left( \int_{\left[ 0,R%
\right] }g\left( \epsilon\right) d\epsilon\right) ^{2}+MR^{\theta}.
\end{align*}

Using the fact that $\theta<1$ and $\epsilon_{+}\leq R$ we obtain $\left(
\epsilon_{+}\right) ^{\theta-1}\geq\left( R\right) ^{\theta-1}$ whence (\ref%
{W1E2}) follows.
\end{proof}

\subsection{An estimate for nonnegative sequences.}

We will need the following technical Lemma.

\begin{lemma}
\label{seq2}Let us consider two sequences of nonnegative numbers $\left\{
I_{k}\right\} _{k=0}^{\infty}$ and  $\left\{ A_{k}\right\} _{k=0}^{\infty}$
satisfying the inequalities:%
\begin{equation}
I_{M}^{2}\sum_{M<k-1}I_{j}\leq A_{M}   \label{B1a}
\end{equation}
for any $M=0,1,2,...$ . Then for any $M_{0}\geq0$ we have:%
\begin{equation}
\sum_{M_{0}\leq j\leq\ell<k-1}I_{k}I_{\ell}I_{j}\leq\sum_{M_{0}\leq
j\leq\ell }\sqrt{A_{j}A_{\ell}}.   \label{B2b}
\end{equation}
\end{lemma}

\begin{proof}
The estimate (\ref{B1a}) implies:%
\begin{equation*}
I_{M}\left( \sum_{k>M+1}I_{k}\right) ^{\frac{1}{2}}\leq\sqrt{A_{M}}. 
\end{equation*}
Taking the values $M=j$ and $M=\ell$ and multiplying the resulting
inequalities we obtain:%
\begin{equation*}
I_{j}I_{\ell}\left( \sum_{k>j+1}I_{k}\right) ^{\frac{1}{2}}\left(
\sum_{k>\ell+1}I_{k}\right) ^{\frac{1}{2}}\leq\sqrt{A_{j}A_{\ell}}, 
\end{equation*}
for any $j,\ell=0,1,2,...$.Summing this inequality for $M_{0}\leq j\leq\ell$
we obtain 
\begin{equation*}
\sum_{M_{0}\leq j\leq\ell}I_{j}I_{\ell}\left( \sum_{k>j+1}I_{k}\right) ^{%
\frac{1}{2}}\left( \sum_{k>\ell+1}I_{k}\right) ^{\frac{1}{2}}\leq
\sum_{M_{0}\leq j\leq\ell}\sqrt{A_{j}A_{\ell}}. 
\end{equation*}

Using that, for $j\leq\ell$ we have $\sum_{k>j+1}I_{k}\geq\sum_{k>%
\ell+1}I_{k}$ we obtain:%
\begin{equation*}
\sum_{M_{0}\leq j\leq\ell}I_{j}I_{\ell}\left( \sum_{k>\ell+1}I_{k}\right)
=\sum_{M_{0}\leq j\leq\ell<k-1}I_{j}I_{\ell}I_{k}\leq\sum_{M_{0}\leq j\leq
\ell}\sqrt{A_{j}A_{\ell}}. 
\end{equation*}
\end{proof}

\subsection{From one estimate for the rate of collisions to an estimate for
the number of triples.}

\subsubsection{Notation and some geometrical results.}

As a next step we transform estimate (\ref{W1E2}) into a new one that does
not contain the power laws of $\epsilon$ and only contains the measures $g.$
The resulting formula is more convenient to derive estimates for the number
of particles concentrated near $\epsilon=0.$

We will need in all the following some suitable notation. Given $a>0,$ we
define a sequence of intervals $\left\{ \mathcal{I}_{k}\right\}
_{k=0}^{\infty}$ contained in the interval $\left[ 0,1\right] $ by means
of:\ 
\begin{equation}
\mathcal{I}_{k}\left( b\right) =b^{-k}\left( \frac{1}{b},1\right] \ \ ,\ \
k=0,1,2,...\, ,\ \ b>1.   \label{B3}
\end{equation}

Notice that $\bigcup_{k=0}^{\infty}\mathcal{I}_{k}\left( b\right) =\left( 0,1%
\right] ,\ \mathcal{I}_{k}\left( b\right) \cap\mathcal{I}_{j}\left( b\right)
=\varnothing$ if $k\neq j.$ Given a measure $g\in\mathcal{M}_{+}\left( \left[
0,1\right] \right) ,$ we remark that, if $\int_{\left\{ 0\right\} }g\left(
\epsilon\right) d\epsilon=0$ we have:%
\begin{equation}
\int_{\left[ 0,1\right] }g\left( \epsilon\right) d\epsilon=\sum
_{k=0}^{\infty}\int_{\mathcal{I}_{k}\left( b\right) }g\left( \epsilon
\right) d\epsilon.   \label{B4}
\end{equation}

We need to define also the ``extended'' intervals:%
\begin{equation}
\mathcal{I}_{k}^{\left( E\right) }\left( b\right) =\mathcal{I}_{k-1}\left(
b\right) \cup\mathcal{I}_{k}\left( b\right) \cup \mathcal{I}_{k+1}\left(
b\right) ,\ \ k=0,1,2,...\ \ \   \label{B5}
\end{equation}
where, by convenience, we assume that $\mathcal{I}_{-1}\left( b\right)
=\varnothing.$

We remark that each $\epsilon\in\left( 0,1\right] $ belongs to three sets $%
\mathcal{I}_{k}^{\left( E\right) }\left( b\right) .$ It also readily follows
that:%
\begin{equation*}
3\int_{\left[ 0,1\right] }g\left( \epsilon\right) d\epsilon=\sum
_{k=0}^{\infty}\int_{\mathcal{I}_{k}^{\left( E\right) }\left( b\right)
}g\left( \epsilon\right) d\epsilon. 
\end{equation*}

We will write $\mathcal{I}_{k}=\mathcal{I}_{k}\left( b\right) ,\ \mathcal{I}%
_{k}^{\left( E\right) }=\mathcal{I}_{k}^{\left( E\right) }\left( b\right) $
if the dependence of the intervals in $b$ is clear in the argument.

We also define, for further references, a set $\mathcal{P}_{b}$ of subsets
of $\left[ 0,1\right] :$%
\begin{equation}
\mathcal{P}_{b}\mathcal{=}\left\{ A\subset\left[ 0,1\right] :A=\bigcup _{j}%
\mathcal{I}_{k_{j}}\left( b\right) \text{ for some set of indexes }\left\{
k_{j}\right\} \subset\left\{ 1,2,...\right\} \right\} .   \label{B5c}
\end{equation}

Notice that the elements of $\mathcal{P}_{b}$ consists of unions of elements
of the family $\left\{ \mathcal{I}_{k}\left( b\right) \right\} .$ The set $%
\left\{ k_{j}\right\} $ can contain a finite or infinity number of elements.

Given $A\in\mathcal{P}_{b}$ we define an extended set $A^{\left( E\right) }$
as follows. Suppose that $A=\bigcup_{j=1}^{\infty}\mathcal{I}_{k_{j}}\left(
b\right) .$ We then define:%
\begin{equation}
A^{\left( E\right) }=\bigcup_{j=1}^{\infty}\mathcal{I}_{k_{j}}^{\left(
E\right) }\left( b\right) .   \label{B5d}
\end{equation}

We will also need the following family of rescaled intervals. Given $%
R\in\left( 0,1\right] $ and $b>1$ we define: $\left\{ \mathcal{I}_{k}\left(
b,R\right) \right\} ,\ \left\{ \mathcal{I}_{k}^{\left( E\right) }\left(
b,R\right) \right\} $ by means of:%
\begin{equation}
\mathcal{I}_{k}\left( b,R\right) =R\mathcal{I}_{k}\left( b\right) \ \ ,\ \ 
\mathcal{I}_{k}^{\left( E\right) }\left( b,R\right) =R\mathcal{I}%
_{k}^{\left( E\right) }\left( b\right) \ \ ,\ \ k=0,1,2,...\, , 
\label{Z1E1}
\end{equation}
with $\left\{ \mathcal{I}_{k}\left( b\right) \right\} ,\ \left\{ \mathcal{I}%
_{k}^{\left( E\right) }\left( b\right) \right\} $ as in (\ref{B3}), (\ref{B5}%
). We define also a class of sets $\mathcal{P}_{b}\left( R\right) $ as
follows:%
\begin{equation}
\mathcal{P}_{b}\left( R\right) =\left\{ A\subset\left[ 0,R\right] :A=RB,\ \
B\in\mathcal{P}_{b}\right\} , \   \label{Z1E2}
\end{equation}
where $\mathcal{P}_{b}$ is as in (\ref{B5d}). We can also define the concept
of extended sets. Given $A\in\mathcal{P}_{b}\left( R\right) ,$ with the form 
$A=RB,\ B\in\mathcal{P}_{b}$ we define:%
\begin{equation}
A^{\left( E\right) }=RB^{\left( E\right) }.   \label{Z1E3}
\end{equation}

We now define the following family of subsets of the cube $\left[ 0,1\right]
^{3}:$%
\begin{equation}
\mathcal{S}_{R,\rho}=\left\{ \left( \epsilon_{1},\epsilon_{2},\epsilon
_{3}\right) \in\left[ 0,R\right] ^{3}:\left\vert \epsilon_{0}-\epsilon
_{-}\right\vert >\rho\epsilon_{0}\right\} \, ,\ \ 0<R\leq1,\ \ 0<\rho<1\ . 
\label{B5b}
\end{equation}
where we will assume in the following that $\epsilon_{-},\ \epsilon _{0},\
\epsilon_{+}$ are as in Definition \ref{aux}. We finally define also the
sets:%
\begin{equation}
\mathcal{I}_{j,\ell,k}\left( b\right) =\mathcal{I}_{j}\left( b\right) \times%
\mathcal{I}_{\ell}\left( b\right) \times\mathcal{I}_{k}\left( b\right)
\subset\left[ 0,1\right] ^{3}\ \ ,\ \ j,k,\ell=0,1,2,...   \label{B5a}
\end{equation}

\begin{lemma}
\label{subsets}Suppose that $0<R\leq1,\ \ 0<\rho<1.$ Then:%
\begin{equation}
\mathcal{S}_{R,\rho}\subset\left( \bigcup_{\sigma\in S^{3}}\left[
\bigcup_{N\left( R,b\right) \leq j\leq\ell<k-1}\mathcal{I}_{\sigma\left(
j,\ell,k\right) }\left( b\right) \right] \right) \subset\mathcal{S}_{\left(
bR\right) \wedge1,\left( 1-\frac{1}{b}\right) }\, ,\ \ \   \label{B6}
\end{equation}
where $bR\wedge1=\min\left\{ bR,1\right\} ,$ $N\left( R,b\right) =\left[ 
\frac{\log\left( \frac{1}{R}\right) }{\log\left( b\right) }\right] $\ , $b=%
\frac{1}{1-\rho}.$
\end{lemma}

\begin{proof}
Suppose that $\left( \epsilon_{1},\epsilon_{2},\epsilon_{3}\right) \in%
\mathcal{S}_{R,\rho}.$ Due to the invariance of the result under
permutations of the indexes we can assume without loss of generality that $%
\epsilon_{3}=\epsilon_{-}<\epsilon_{2}=\epsilon_{0}\leq\epsilon_{1}=%
\epsilon_{+}.$ The choice of $N\left( R,b\right) $ implies that there exist $%
\ell,\ j$ such that $\epsilon_{2}\in\mathcal{I}_{\ell},\ \epsilon_{1}\in%
\mathcal{I}_{j}$ with $j\leq\ell.$ Due to the definition of $\mathcal{S}%
_{R,\rho}$ and $\mathcal{I}_{\ell}$ we have $\epsilon_{3}<\left(
1-\rho\right) \epsilon_{2}\leq\left( 1-\rho\right) b^{-\ell}.$ Since $\left(
1-\rho\right) =b^{-2}$ we then have $\epsilon_{3}<b^{-\left( \ell+2\right) }$
whence $\epsilon_{3}\in\bigcup_{k>\ell+1}\mathcal{I}_{k}$. Therefore:%
\begin{equation*}
\left( \epsilon_{1},\epsilon_{2},\epsilon_{3}\right) \in\bigcup_{N\left(
R,b\right) \leq j\leq\ell<k-1}\mathcal{I}_{j,\ell,k}\left( b\right) . 
\end{equation*}

This gives the first inclusion in (\ref{B6}). In order to prove the second
inclusion, we assume, without loss of generality, that $\epsilon_{3}=%
\epsilon_{-},\ \epsilon_{2}=\epsilon_{0},\ \epsilon_{1}=\epsilon_{+}.$
Suppose that $\epsilon_{2}\in\mathcal{I}_{\ell},\ \epsilon_{1}\in \mathcal{I}%
_{j}$ with $j\leq\ell.$ Then $\epsilon_{3}\leq b^{-\left( \ell+2\right) },$ $%
\epsilon_{2}>b^{-\left( \ell+1\right) }.$ Therefore $\left\vert
\epsilon_{0}-\epsilon_{-}\right\vert >\bar{\rho}\epsilon_{0}$ if $\bar{\rho}%
\leq\left( 1-\frac{1}{b}\right) $ and the result follows.
\end{proof}

\subsection{Estimating the number of triples not too close to the diagonal.}

We now prove the following result which provides a precise estimate for the
number of triples $\left( \epsilon_{1},\epsilon_{2},\epsilon_{3}\right) \in%
\left[ 0,R\right] ^{3}$ whose distance to the diagonal is comparable to
their distance to the origin.

\begin{lemma}
\label{estProd}Suppose that $g\in L^{\infty}\left( \left[ 0,T\right] ;%
\mathcal{M}_{+}\left( \left[ 0,1\right] \right) \right)$, satisfies (\ref%
{W1E2}) for any $0\leq R\leq1$ and $T>0.$ Suppose also that:
\begin{equation}
\label{Enodiract}
\int_{\left\{
0\right\} }g\left( \epsilon,t\right) d\epsilon=0,\,\,\,\hbox{ for any}\,\, t\in\left[ 0,T
\right] .
\end{equation}
Let $0<\rho<1$ and $\mathcal{S}_{R,\rho}$ as in (\ref{B5b}).
Then, for any $T>0$ we have:%
\begin{equation}
B\int_{0}^{T}dt\int_{\mathcal{S}_{R,\rho}}\left[ \prod_{m=1}^{3}\,g_{m}d%
\mathcal{\epsilon}_{m}\right] \leq\frac{2b^{\frac{7}{2}}R}{\rho ^{2}\left( 
\sqrt{b}-1\right) ^{2}}\left[ 2\pi\int_{0}^{T}dt\left( \int_{\left[ 0,1%
\right] }g\left( \epsilon\right) d\epsilon\right) ^{2}+M\right] , \ 
\label{H1}
\end{equation}
with $b$ as in (\ref{B6}) and $R\in\left[ 0,\frac{1}{2}\right] $ and $B$ as
in (\ref{W1E2}).
\end{lemma}

\begin{remark}
Notice that Lemma \ref{estProd} provides a general estimate for arbitrary measures $g$
satisfying (\ref{W1E2}) even if they are  completely unrelated to the equation (%
\ref{F3E2}), (\ref{F3E3}).
\end{remark}

\begin{proof}
We define: 
\begin{equation}
f_{R}\left( t\right) =\frac{1}{R}\int_{\left[ 0,\frac{R}{2}\right] ^{3}}%
\left[ \prod_{m=1}^{3}\,g_{m}d\mathcal{\epsilon}_{m}\right] \frac{\left(
\epsilon_{0}\right) ^{\frac{3}{2}}}{\left( \epsilon_{+}\right) ^{\frac {3}{2}%
}}\left( \frac{\epsilon_{0}-\epsilon_{-}}{\epsilon_{0}}\right) ^{2} 
\label{B7}
\end{equation}

Notice that (\ref{W1E2}) implies that:%
\begin{equation}
B\int_{0}^{T}f_{R}\left( t\right) dt\leq\left[ 2\pi\int_{0}^{T}dt\left(
\int_{\left[ 0,1\right] }g\left( \epsilon\right) d\epsilon\right) ^{2}+M%
\right] \ \ ,\ \ 0<R\leq1.   \label{B7a}
\end{equation}

A crucial point in the following is the fact that this estimate is uniform
in $R.$

Let us select $b>1$ as in (\ref{B6}). We define sets $\mathcal{I}_{j,\ell
,k}\left( b\right) $ by means of (\ref{B3}), (\ref{B5a}). Notice that $%
\mathcal{I}_{\left( j_{1},\ell_{1},k_{1}\right) }\left( b\right) \cap%
\mathcal{I}_{\left( j_{2},\ell_{2},k_{2}\right) }\left( b\right) =\varnothing
$ if $\left( j_{1},\ell_{1},k_{1}\right) \neq\left(
j_{2},\ell_{2},k_{2}\right) .$ Lemma \ref{subsets} as well as the definition
of $N\left( R,b\right) $ in (\ref{B6}) imply:%
\begin{equation}
\mathcal{S}_{\frac{R}{2},\rho}\subset\left( \bigcup_{\sigma\in S^{3}}\left[
\bigcup_{N\left( \frac{R}{2},b\right) \leq j\leq\ell<k-1}\mathcal{I}%
_{\sigma\left( j,\ell,k\right) }\left( b\right) \right] \right) \subset\left[
0,\frac{bR}{2}\right] ^{3}.   \label{B9}
\end{equation}

Notice that for each $j,\ell,k$ the sets in the family $\left\{ \mathcal{I}%
_{\sigma\left( j,\ell,k\right) }\left( b\right) :\sigma\in S^{3}\right\} $
are disjoint, except if $j=\ell.$ In such a case the permutation $\sigma$
that keeps $k$ constant and exchange the indexes $j$ and $\ell$ implies $%
\mathcal{I}_{\sigma\left( j,\ell,k\right) }\left( b\right) =\mathcal{I}%
_{j,\ell,k}\left( b\right) .$ Therefore, each point $\left(
\epsilon_{1},\epsilon_{2},\epsilon_{3}\right) $ is contained at most in two
of the sets $\mathcal{I}_{\sigma\left( j,\ell,k\right) }\left( b\right) $ in
(\ref{B9}) and we then have:%
\begin{align}
& \frac{1}{2}\sum_{\sigma\in S^{3}}\left[ \sum_{N\left( \frac{R}{2},b\right)
\leq j\leq\ell<k-1}\int_{\mathcal{I}_{\sigma\left( j,\ell,k\right) }\left(
b\right) }\left[ \prod_{m=1}^{3}\,g_{m}d\mathcal{\epsilon}_{m}\right]
\Psi\left( \epsilon_{1},\epsilon_{2},\epsilon_{3}\right) \right] \ 
\label{B10} \\
& \leq\int_{\mathcal{V}_{R,b}}\left[ \prod_{m=1}^{3}\,g_{m}d\mathcal{\epsilon%
}_{m}\right] \Psi\left( \epsilon_{1},\epsilon_{2},\epsilon_{3}\right) , 
\notag
\end{align}
where:%
\begin{equation}
\mathcal{V}_{R,b}=\bigcup_{\sigma\in S^{3}}\left[ \bigcup_{N\left( \frac {R}{%
2},b\right) \leq j\leq\ell<k-1}\mathcal{I}_{\sigma\left( j,\ell ,k\right)
}\left( b\right) \right] \ \ \ ,\ \ \ \Psi\left( \epsilon
_{1},\epsilon_{2},\epsilon_{3}\right) =\frac{\left( \epsilon_{0}\right) ^{%
\frac{3}{2}}}{\left( \epsilon_{+}\right) ^{\frac{3}{2}}}\left( \frac{%
\epsilon_{0}-\epsilon_{-}}{\epsilon_{0}}\right) ^{2}.   \label{F1}
\end{equation}

Moreover, since:%
\begin{equation}
\int_{\mathcal{I}_{\sigma\left( j,\ell,k\right) }\left( b\right) }\left[
\prod_{m=1}^{3}\,g_{m}d\mathcal{\epsilon}_{m}\right] \Psi\left( \epsilon
_{1},\epsilon_{2},\epsilon_{3}\right) =\int_{\mathcal{I}_{j,\ell,k}\left(
b\right) }\left[ \prod_{m=1}^{3}\,g_{m}d\mathcal{\epsilon}_{m}\right]
\Psi\left( \epsilon_{1},\epsilon_{2},\epsilon_{3}\right) \, ,\ \ \sigma\in
S^{3},   \label{F2}
\end{equation}
and the cardinal of $S^{3}$ is six, (\ref{B10}) implies:%
\begin{align}
&\left[ \sum_{N\left( \frac{R}{2},b\right) \leq j\leq\ell<k-1}\int_{\mathcal{I%
}_{j,\ell,k}\left( b\right) }\left[ \prod_{m=1}^{3}\,g_{m}d\mathcal{\epsilon}%
_{m}\right] \Psi\left( \epsilon_{1},\epsilon _{2},\epsilon_{3}\right) \right]\leq \notag\\
&\hskip 5cm \leq\frac{1}{3}\int_{\mathcal{V}_{R,b}}\left[ \prod_{m=1}^{3}\,g_{m}d%
\mathcal{\epsilon}_{m}\right] \Psi\left(
\epsilon_{1},\epsilon_{2},\epsilon_{3}\right) .   \label{F3a}
\end{align}

On the other hand, (\ref{B9}) and the definition of $f_{R}\left( t\right) $
in (\ref{B7}) imply:%
\begin{equation}
\int_{\mathcal{V}_{R,b}}\left[ \prod_{m=1}^{3}\,g_{m}d\mathcal{\epsilon}_{m}%
\right] \Psi\left( \epsilon_{1},\epsilon_{2},\epsilon_{3}\right) \leq
bRf_{bR}\left( t\right) .   \label{F4}
\end{equation}

The nonnegativity of $g,\ \Psi$ yields:%
\begin{align}
& \sum_{N\left( \frac{R}{2},b\right) \leq j=\ell<k-1}\int_{\mathcal{I}%
_{j,\ell,k}\left( b\right) }\left[ \prod_{m=1}^{3}\,g_{m}d\mathcal{\epsilon }%
_{m}\right] \Psi\left( \epsilon_{1},\epsilon_{2},\epsilon_{3}\right) \leq 
\notag \\
& \hskip 3cm \leq\sum_{N\left( \frac{R}{2},b\right) \leq j\leq\ell<k-1}\int_{%
\mathcal{I}_{j,\ell,k}\left( b\right) }\left[ \prod_{m=1}^{3}\,g_{m}d%
\mathcal{\epsilon}_{m}\right] \Psi\left( \epsilon_{1},\epsilon
_{2},\epsilon_{3}\right) .   \label{F5}
\end{align}

Using the definition of the sets $\mathcal{I}_{j,\ell,k}\left( b\right) $ in
(\ref{B5a}), as well as (\ref{B6}), we obtain that in the integral term on
the right-hand side of (\ref{F5}) we have, using that $j=\ell<k-1:$%
\begin{equation}
\epsilon_{-}=\epsilon_{3}\ \ ,\ \ \frac{1}{b^{\frac{3}{2}}}\leq\frac{\left(
\epsilon_{0}\right) ^{\frac{3}{2}}}{\left( \epsilon_{+}\right) ^{\frac {3}{2}%
}}\leq1\ \ ,\ \ \left( \frac{\epsilon_{0}-\epsilon_{-}}{\epsilon_{0}}\right)
^{2}\geq\left( 1-\frac{1}{b}\right) ^{2}=\rho^{2}.   \label{F6}
\end{equation}

We define:%
\begin{equation}
\int_{\mathcal{I}_{j}\left( b\right) }gd\epsilon=I_{j}\ \ ,\ \ j=0,1,2,... 
\label{F7}
\end{equation}

Combining (\ref{F1}), (\ref{F5}), (\ref{F6}), (\ref{F7}) we then obtain:%
\begin{equation*}
\frac{\rho^{2}}{b^{\frac{3}{2}}}\sum_{N\left( \frac{R}{2},b\right) \leq
j<k-1}I_{j}^{2}I_{k}\leq\sum_{N\left( \frac{R}{2},b\right) \leq j\leq
\ell<k-1}\int_{\mathcal{I}_{j,\ell,k}\left( b\right) }\left[ \prod
_{m=1}^{3}\,g_{m}d\mathcal{\epsilon}_{m}\right] \Psi\left( \epsilon
_{1},\epsilon_{2},\epsilon_{3}\right) . 
\end{equation*}

Using this inequality, as well as (\ref{F3a}), (\ref{F4}), (\ref{F5}) we
arrive at:%
\begin{equation*}
\frac{\rho^{2}}{b^{\frac{3}{2}}}\sum_{N\left( \frac{R}{2},b\right) \leq
j<k-1}I_{j}^{2}I_{k}\leq\frac{bR}{3}f_{bR}\left( t\right) . 
\end{equation*}

Suppose that we write $M=N\left( \frac{R}{2},b\right) .$ The definition of $%
N\left( \frac{R}{2},b\right) $ in (\ref{B6}) then implies:%
\begin{equation*}
\frac{\rho^{2}}{b^{\frac{3}{2}}}\sum_{M\leq j<k-1}I_{j}^{2}I_{k}\leq \frac{%
2b^{1-M}}{3}f_{bR}\left( t\right) , 
\end{equation*}
whence, keeping only the terms in the sum with $j=M$ and choosing $R_{M}$ as
the solution of $M=\left[ \frac{\log\left( \frac{2}{R}\right) }{\log\left(
b\right) }\right] $ we obtain:%
\begin{equation*}
I_{M}^{2}\sum_{M<k}I_{k}\leq\left[ \frac{2b^{\frac{5}{2}}}{3\rho^{2}}%
f_{bR_{M}}\left( t\right) \right] b^{-M}. 
\end{equation*}

By Lemma \ref{seq2}, we deduce: 
\begin{align*}
\sum_{M_{0}\leq j\leq\ell<k-1}I_{k}\left( b\right) I_{\ell}\left( b\right)
I_{j}\left( b\right) & \leq\sum_{M_{0}\leq j \leq\ell}\sqrt{\left[ \frac{2b^{%
\frac{5}{2}}}{3\rho^{2}}f_{bR_{j}}\left( t\right) \right] b^{-j}\left[ \frac{%
2b^{\frac{5}{2}}}{3\rho^{2}}f_{bR_{\ell}}\left( t\right) \right] b^{-\ell}}
\\
& =\frac{2b^{\frac{5}{2}}}{3\rho^{2}}\sum_{M_{0}\leq j\leq\ell}\frac {\sqrt{%
f_{bR_{j}}\left( t\right) f_{bR_{\ell}}\left( t\right) }}{\sqrt{b^{j+\ell}}},
\end{align*}
for any $M_{0}\geq\left[ \frac{\log\left( \frac{2}{R}\right) }{\log\left(
b\right) }\right] $. We choose then $M_{0}=\left[ \frac{\log\left( \frac{2}{R%
}\right) }{\log\left( b\right) }\right] .$ We now notice that, due to (\ref%
{B9}):%
\begin{equation*}
\int_{\mathcal{S}_{\frac{R}{2},\rho}}\left[ \prod_{m=1}^{3}\,g_{m}d\mathcal{%
\epsilon}_{m}\right] \leq6\sum_{M_{0}\leq j\leq\ell<k-1}I_{k}\left( b\right)
I_{\ell}\left( b\right) I_{j}\left( b\right) 
\end{equation*}

and therefore:%
\begin{equation}
\int_{\mathcal{S}_{\frac{R,}{2}\rho}}\left[ \prod_{m=1}^{3}\,g_{m}d\mathcal{%
\epsilon}_{m}\right] \leq\frac{4b^{\frac{5}{2}}}{\rho^{2}}\sum_{M_{0}\leq
j\leq\ell}\frac{\sqrt{f_{bR_{j}}\left( t\right) f_{bR_{\ell }}\left(
t\right) }}{\sqrt{b^{j+\ell}}}.   \label{F8}
\end{equation}

Integrating (\ref{F8}) in $\left[ 0,T\right] $ and using Young's inequality
we obtain:%
\begin{equation*}
\int_{0}^{T}dt\int_{\mathcal{S}_{\frac{R}{2},\rho}}\left[ \prod_{m=1}^{3}%
\,g_{m}d\mathcal{\epsilon}_{m}\right] \leq\frac{2b^{\frac{5}{2}}}{\rho^{2}}%
\sum_{M_{0}\leq j\leq\ell}\frac{1}{\sqrt{b^{j+\ell}}}\int_{0}^{T}\left[
f_{bR_{j}}\left( t\right) +f_{bR_{\ell}}\left( t\right) \right] dt. 
\end{equation*}

Using the estimate (\ref{B7a}) that is uniform in $R$ we obtain:%
\begin{equation*}
B\int_{0}^{T}dt\int_{\mathcal{S}_{\frac{R}{2},\rho}}\left[
\prod_{m=1}^{3}\,g_{m}d\mathcal{\epsilon}_{m}\right] \leq\frac{4b^{\frac{5}{2%
}}}{\rho^{2}}\left[ 2\pi\int_{0}^{T}dt\left( \int_{\left[ 0,1\right]
}g\left( \epsilon\right) d\epsilon\right) ^{2}+M\right] \sum_{M_{0}\leq
j\leq\ell}\frac{1}{\sqrt{b^{j+\ell}}}. 
\end{equation*}

Using then the inequalities%
\begin{equation*}
\sum_{M_{0}\leq j\leq\ell}\frac{1}{\sqrt{b^{j+\ell}}}\leq\frac{b^{-M_{0}}}{%
\left( 1-\frac{1}{\sqrt{b}}\right) ^{2}}\leq\frac{bR}{2\left( \sqrt {b}%
-1\right) ^{2}}, 
\end{equation*}
where we have used the definition of $M_{0}$ we obtain:%
\begin{equation*}
B\int_{0}^{T}dt\int_{\mathcal{S}_{\frac{R}{2},\rho}}\left[
\prod_{m=1}^{3}\,g_{m}d\mathcal{\epsilon}_{m}\right] \leq\frac{2b^{\frac{5}{2%
}}}{\rho^{2}}\left[ 2\pi\int_{0}^{T}dt\left( \int_{\left[ 0,1\right]
}g\left( \epsilon\right) d\epsilon\right) ^{2}+M\right] \frac{bR}{\left( 
\sqrt{b}-1\right) ^{2}}, 
\end{equation*}
and (\ref{H1}) follows.
\end{proof}

\section{A Measure Theory result.}

\setcounter{equation}{0} \setcounter{theorem}{0}

We now prove the following measure theory result which will play a crucial
role in all the remaining part of the argument.

\begin{lemma}
\label{alt}Suppose that $b>1$ and let us define the intervals $\left\{ 
\mathcal{I}_{k}\left( b\right) \right\} ,\ \left\{ \mathcal{I}_{k}^{\left(
E\right) }\left( b\right) \right\} $ as in (\ref{B3}), (\ref{B5}). Let $%
\mathcal{P}_{b}$ as in (\ref{B5d}) and $A^{\left( E\right) }$ as in (\ref%
{B5c}) for $A\in\mathcal{P}_{b}$. Given $0<\delta<\frac{2}{3},$ we define $%
\eta=\min\left\{ \left( \frac{1}{3}-\frac{\delta}{2}\right) ,\frac{\delta}{6}%
\right\} >0.$ Then, for any $g\in\mathcal{M}^{+}\left[ 0,1\right] $
satisfying 
\begin{equation}
\label{Enodirac}
\int_{\left\{ 0\right\} }gd\epsilon=0,
\end{equation}
 at least one of the
following statements is satisfied:

(i) There exists an interval $\mathcal{I}_{k}\left( b\right) $ such that:%
\begin{equation}
\int_{\mathcal{I}_{k}^{\left( E\right) }\left( b\right) }gd\epsilon
\geq\left( 1-\delta\right) \int_{\left[ 0,1\right] }gd\epsilon,   \label{Q1}
\end{equation}

(ii) There exists two sets $\mathcal{U}_{1},\mathcal{U}_{2}\in\mathcal{P}_{b}
$ such that $\mathcal{U}_{2}\cap\mathcal{U}_{1}^{\left( E\right)
}=\varnothing$ and:%
\begin{equation}
\min\left\{ \int_{\mathcal{U}_{1}}gd\epsilon,\int_{\mathcal{U}%
_{2}}gd\epsilon\right\} \geq\eta\int_{\left[ 0,1\right] }gd\epsilon. 
\label{Q2}
\end{equation}

Moreover, in the case (ii) the set $\mathcal{U}_{1}$ can be written in the
form:%
\begin{equation}
\mathcal{U}_{1}=\bigcup_{j=1}^{L}\mathcal{I}_{k_{j}}\left( b\right) , \ 
\label{Q2a}
\end{equation}
for some set of integers $\left\{ k_{j}\right\} \subset\left\{
1,2,3,...\right\} $ and some finite $L.$ We have:%
\begin{equation}
\mathcal{I}_{k_{m}}\left( b\right) \cap\left( \bigcup_{j=1}^{m-1}\mathcal{I}%
_{k_{j}}^{\left( E\right) }\left( b\right) \right) =\varnothing\ \ ,\ \
m=2,3,...L   \label{Q2b}
\end{equation}
and also:%
\begin{equation}
\sum_{j=1}^{L}\left( \int_{\mathcal{I}_{k_{j}}\left( b\right)
}gd\epsilon\right) ^{2}\leq\left( \int_{\mathcal{I}_{k_{1}}\left( b\right)
}gd\epsilon\right) ^{2}+\sum_{j=2}^{L}\int_{\mathcal{I}_{k_{1}}\left(
b\right) }gd\epsilon\int_{\mathcal{I}_{k_{j}}\left( b\right) }gd\epsilon, 
\label{Q2c}
\end{equation}%
\begin{equation}
\int_{\mathcal{I}_{k_{1}}\left( b\right) }gd\epsilon<\left( 1-\delta \right)
\int_{\left[ 0,1\right] }gd\epsilon.   \label{Q2d}
\end{equation}
\end{lemma}

\begin{remark}
The choice of the sets $\mathcal{U}_{1},\ \mathcal{U}_{2}$ is not entirely
symmetric. The property (\ref{Q2b}) holds for $\mathcal{U}_{1}$ but not for $%
\mathcal{U}_{2}.$
\end{remark}

\begin{proof}
Since the result is trivial if $\int_{\left[ 0,1\right] }gd\epsilon=0$ we
can assume without loss of generality that $\int_{\left[ 0,1\right]
}gd\epsilon=1$ replacing, if needed, $g$ by $\frac{g}{\int_{\left[ 0,1\right]
}gd\epsilon}.$ We will denote as $\mathcal{G}_{1}$ the family of intervals $%
\left\{ \mathcal{I}_{k}\left( b\right) \right\} .$ Using that $\int_{\left\{
0\right\} }gd\epsilon=0$ we have: 
\begin{equation}
\sum_{\mathcal{I}\in\mathcal{G}_{1}}\int_{\mathcal{I}}gd\epsilon=\int_{\left[
0,1\right] }gd\epsilon=1.   \label{S1}
\end{equation}

We define:%
\begin{equation*}
a_{1}=\max\left\{ \int_{\mathcal{I}}gd\epsilon:\mathcal{I}\in\mathcal{G}%
_{1}\right\} =\int_{\mathcal{I}_{\left( 1\right) }}gd\epsilon \ \ ,\ \ 
\mathcal{I}_{\left( 1\right) }\in\mathcal{G}_{1}. 
\end{equation*}

Since the sum on the left of (\ref{S1}) is finite, it follows that this
maximum exists. The interval $\mathcal{I}_{\left( 1\right) }$ does not need
to be unique. Since $\int_{\mathcal{I}_{\left( 1\right) }^{\left( E\right)
}}gd\epsilon\geq a_{1},$ if $a_{1}\geq\left( 1-\delta\right) ,$ we would
have (\ref{Q1}) with $\mathcal{I}_{k}\left( b\right) =\mathcal{I}_{\left(
1\right) }.$ Suppose then that $a_{1}<\left( 1-\delta\right) .$ We define $%
\mathcal{G}_{2}$ as:%
\begin{equation*}
\mathcal{G}_{2}=\left\{ \mathcal{I}\setminus\mathcal{I}_{\left( 1\right)
}^{\left( E\right) }:\mathcal{I\in G}_{1}\right\} . 
\end{equation*}

Notice that $\mathcal{G}_{2}\subset\mathcal{G}_{1}.$ We define now:%
\begin{equation*}
a_{2}=\max\left\{ \int_{\mathcal{I}}gd\epsilon:\mathcal{I}\in\mathcal{G}%
_{2}\right\} =\int_{\mathcal{I}_{\left( 2\right) }}gd\epsilon \ \ ,\ \ 
\mathcal{I}_{\left( 2\right) }\in\mathcal{G}_{2}. 
\end{equation*}

If $\int_{\mathcal{I}_{\left( 1\right) }^{\left( E\right) }\cup \mathcal{I}%
_{\left( 2\right) }^{\left( E\right) }}gd\epsilon<\left( 1-\delta\right) $
we continue the iteration procedure and define sequentially sets $\mathcal{G}%
_{k},$ values $a_{k}$ and intervals $\mathcal{I}_{\left( k\right) }$ as long
as we have $\int_{\bigcup_{j=1}^{k-1}\mathcal{I}_{\left( j\right) }^{\left(
E\right) }}gd\epsilon<\left( 1-\delta\right) :$%
\begin{equation*}
\mathcal{G}_{k}=\left\{ \mathcal{I}\setminus\mathcal{I}_{\left( k-1\right)
}^{\left( E\right) }:\mathcal{I\in G}_{k-1}\right\} \ \ ,\ \ k=2,3,... , 
\end{equation*}%
\begin{equation}
a_{k}=\max\left\{ \int_{\mathcal{I}}gd\epsilon:\mathcal{I}\in\mathcal{G}%
_{k}\right\} =\int_{\mathcal{I}_{\left( k\right) }}gd\epsilon \ \ ,\ \ 
\mathcal{I}_{\left( k\right) }\in\mathcal{G}_{k}.   \label{Q3a}
\end{equation}

Due to (\ref{S1}), iterating the procedure, we eventually arrive to some
integer value $M\geq2$ such that:%
\begin{equation}
\int_{\bigcup_{j=1}^{M}\mathcal{I}_{\left( j\right) }^{\left( E\right)
}}gd\epsilon\geq\left( 1-\delta\right) \ \ \text{\ and\ \ }\int
_{\bigcup_{j=1}^{M-1}\mathcal{I}_{\left( j\right) }^{\left( E\right)
}}gd\epsilon<\left( 1-\delta\right) .   \label{Q3}
\end{equation}

We define $\mathcal{U}_{1}=\bigcup_{j=1}^{M-1}\mathcal{I}_{\left( j\right)
}, $ $\mathcal{U}_{2}=\left[ 0,1\right] \setminus\mathcal{U}_{1}^{\left(
E\right) }.$ Notice that $\mathcal{U}_{1},\mathcal{U}_{2}\in\mathcal{P}_{b}$
and $\mathcal{U}_{2}\cap\mathcal{U}_{1}^{\left( E\right) }=\varnothing.\ $We
prove now that (\ref{Q2}) holds for some $\eta>0.$ To this end we consider
two different possibilities. Either $a_{M}\geq\frac{\delta}{6}$ or $a_{M}<%
\frac{\delta}{6}.$ Notice that in both cases the second inequality in (\ref%
{Q3}) implies:%
\begin{equation}
\int_{\mathcal{U}_{2}}gd\epsilon\geq\delta.   \label{Q4}
\end{equation}

Suppose first that $a_{M}\geq\frac{\delta}{6}.$ Then $\int_{\mathcal{U}%
_{1}}gd\epsilon\geq\int_{\mathcal{I}_{\left( 1\right) }}gd\epsilon=a_{1}\geq
a_{M}\geq\frac{\delta}{6}.$ Therefore (\ref{Q2}) holds with $\eta=\frac {%
\delta}{6}.$ Suppose now that $a_{M}<\frac{\delta}{6}.$ Then, using the
first inequality of (\ref{Q3}) we obtain: 
\begin{equation}
\left( 1-\delta\right) \leq\int_{\bigcup_{j=1}^{M}\mathcal{I}_{\left(
j\right) }^{\left( E\right) }}gd\epsilon\leq\int_{\mathcal{U}_{1}^{\left(
E\right) }}gd\epsilon+\int_{\mathcal{I}_{\left( M\right) }^{\left( E\right)
}\setminus\mathcal{U}_{1}^{\left( E\right) }}gd\epsilon.   \label{Q5}
\end{equation}

Notice that the definitions of $a_{M}$, $\mathcal{I}_{\left( M\right)
}^{\left( E\right) },\ \mathcal{U}_{1}$ and the families $\mathcal{G}_{k}$
imply:%
\begin{equation*}
\int_{\mathcal{I}_{\left( M\right) }^{\left( E\right) }\setminus \mathcal{U}%
_{1}^{\left( E\right) }}gd\epsilon\leq3a_{M}<\frac{\delta}{2}, 
\end{equation*}
because $\mathcal{I}_{\left( M\right) }^{\left( E\right) }\setminus \mathcal{%
U}_{1}^{\left( E\right) }$ contains at most three intervals of the family $%
\mathcal{G}_{M}.$ Due to the definition of $a_{M},$ the integral of $g$ over
these intervals is at most $a_{M}.$ Therefore, (\ref{Q5}) yields:%
\begin{equation}
\left( 1-\frac{3\delta}{2}\right) \leq\int_{\mathcal{U}_{1}^{\left( E\right)
}}gd\epsilon.   \label{Q6}
\end{equation}

We now remark that:%
\begin{equation}
\int_{\mathcal{U}_{1}^{\left( E\right) }}gd\epsilon=\sum_{k=1}^{M-1}\int_{%
\mathcal{I}_{\left( k\right) }^{\left( E\right) }\setminus
\bigcup_{j=1}^{k-1}\mathcal{I}_{\left( j\right) }^{\left( E\right)
}}gd\epsilon,   \label{Q6a}
\end{equation}
where we assume that $\bigcup_{j=1}^{0}\mathcal{I}_{\left( j\right)
}^{\left( E\right) }=\varnothing.$ Notice that, by definition of the
sequence $\left\{ a_{k}\right\} $ and the extended intervals $\mathcal{I}%
_{\left( k\right) }^{\left( E\right) }$ we have:%
\begin{equation}
\int_{\mathcal{I}_{\left( k\right) }^{\left( E\right) }\setminus
\bigcup_{j=1}^{k-1}\mathcal{I}_{\left( j\right) }^{\left( E\right)
}}gd\epsilon\leq3a_{k}=3\int_{\mathcal{I}_{\left( k\right) }}gd\epsilon. 
\label{Q7}
\end{equation}

Combining (\ref{Q6}) and (\ref{Q7}) as well as (\ref{Q3a}) and the fact that 
$\mathcal{I}_{\left( k\right) }\cap\mathcal{I}_{\left( j\right) }=\varnothing
$ if $k\neq j$ we obtain:%
\begin{equation}
\left( 1-\frac{3\delta}{2}\right) \leq3\sum_{k=1}^{M-1}\int_{\mathcal{I}%
_{\left( k\right) }}gd\epsilon=3\int_{\bigcup_{k=1}^{M-1}\mathcal{I}_{\left(
k\right) }}gd\epsilon=3\int_{\mathcal{U}_{1}}gd\epsilon.   \label{Q8}
\end{equation}

Combining (\ref{Q4}) and (\ref{Q8}) we obtain (\ref{Q2}) with $\eta
=\min\left\{ \left( \frac{1}{3}-\frac{\delta}{2}\right) ,\delta\right\} .$
Using the result obtained if $a_{M}\geq\frac{\delta}{6}$ we then obtain that
(\ref{Q2}) is valid in all cases with $\eta=\min\left\{ \left( \frac{1}{3}-%
\frac{\delta}{2}\right) ,\frac{\delta}{6}\right\} .$

It remains to prove (\ref{Q2a})-(\ref{Q2d}). Note that (\ref{Q2a}),\ (\ref%
{Q2b}) just follow from the definition of the set $\mathcal{U}_{1}.$ In
order to prove (\ref{Q2c}) we use the fact that the sequence $\left\{
a_{k}\right\} _{k=1}^{M}$ is nonincreasing and:%
\begin{equation*}
\int_{\mathcal{I}_{k_{j}}\left( b\right) }gd\epsilon=a_{j}\ ,\ \
j=1,2,...,\left( M-1\right) , 
\end{equation*}
whence (\ref{Q2c}) follows. Finally (\ref{Q2d}) is a consequence of the
construction of the sequence of intervals $\int_{\mathcal{I}_{k_{1}}\left(
b\right) }gd\epsilon=a_{1}$ as well as the fact that in this case $%
a_{1}<\left( 1-\delta\right) .$
\end{proof}

We will need a rescaled version of Lemma \ref{alt}:

\begin{lemma}
\label{altresc}Suppose that $b>1,\ 0<R\leq1.$ We define intervals $\left\{ 
\mathcal{I}_{k}\left( b,R\right) \right\}$, $\left\{ \mathcal{I}_{k}^{\left(
E\right) }\left( b,R\right) \right\} $ as in (\ref{Z1E1}). Let $\mathcal{P}%
_{b}\left( R\right) $ as in (\ref{Z1E2}) and $A^{\left( E\right) }$ as in (%
\ref{Z1E3}) for $A\in\mathcal{P}_{b}\left( R\right) $. Given $0<\delta<\frac{%
2}{3},$ we define $\eta=\min\left\{ \left( \frac{1}{3}-\frac{\delta}{2}%
\right) ,\frac{\delta}{6}\right\} >0.$ Then, for any $g\in\mathcal{M}^{+}%
\left[ 0,R\right] $ satisfying $\int_{\left\{ 0\right\} }gd\epsilon=0,$ at
least one of the following statements is satisfied:

(i) Either there exist an interval $\mathcal{I}_{k}\left( b,R\right) $ such
that:%
\begin{equation}
\int_{\mathcal{I}_{k}^{\left( E\right) }\left( b,R\right) }gd\epsilon
\geq\left( 1-\delta\right) \int_{\left[ 0,R\right] }gd\epsilon, 
\label{Z1E4}
\end{equation}

(ii) or, either there exist two sets $\mathcal{U}_{1},\mathcal{U}_{2}\in%
\mathcal{P}_{b}\left( R\right) $ such that $\mathcal{U}_{2}\cap\mathcal{U}%
_{1}^{\left( E\right) }=\varnothing$ and:%
\begin{equation}
\min\left\{ \int_{\mathcal{U}_{1}}gd\epsilon,\int_{\mathcal{U}%
_{2}}gd\epsilon\right\} \geq\eta\int_{\left[ 0,R\right] }gd\epsilon. 
\label{Z1E5}
\end{equation}
Moreover, in the case (ii) the set $\mathcal{U}_{1}$ can be written in the
form:%
\begin{equation}
\mathcal{U}_{1}=\bigcup_{j=1}^{L}\mathcal{I}_{k_{j}}\left( b,R\right) 
\label{Z1E5a}
\end{equation}
for some sequence $\left\{ k_{j}\right\} $ and some finite $L.$ We have:%
\begin{equation}
\mathcal{I}_{k_{m}}\left( b,R\right) \cap\left( \bigcup_{j=1}^{m-1}\mathcal{I%
}_{k_{j}}^{\left( E\right) }\left( b,R\right) \right) =\varnothing\ \ ,\ \
m=2,3,...L,\   \label{Z1E5b}
\end{equation}
and also:%
\begin{equation}
\sum_{j=1}^{L}\left( \int_{\mathcal{I}_{k_{j}}\left( b,R\right)
}gd\epsilon\right) ^{2}\leq\left( \int_{\mathcal{I}_{k_{1}}\left( b,R\right)
}gd\epsilon\right) ^{2}+\sum_{j=2}^{L}\int_{\mathcal{I}_{k_{1}}\left(
b,R\right) }gd\epsilon\int_{\mathcal{I}_{k_{j}}\left( b,R\right)
}gd\epsilon,   \label{Z1E5c}
\end{equation}%
\begin{equation}
\int_{\mathcal{I}_{k_{1}}\left( b,R\right) }gd\epsilon<\left(
1-\delta\right) \int_{\left[ 0,d\right] }gd\epsilon.   \label{Z1E5d}
\end{equation}
\end{lemma}

\begin{proof}
It is essentially a rescaling of Lemma \ref{alt}. It can be obtained just
defining a new set of variables $\tilde{\epsilon}=\frac{\epsilon}{R}.$ The
sets $\mathcal{I}_{k}\left( b\right) ,$ $\mathcal{U}_{1},\ \mathcal{U}_{2}$
in Lemma \ref{alt} are then transformed in the sets stated in this Lemma by
means of the inverse transform $\epsilon=R\tilde{\epsilon}.$
\end{proof}

\section{Lower estimate for the triples in $\left[ 0,R\right] ^{3}$ in terms
of those which are separated from the diagonal.}

\setcounter{equation}{0} \setcounter{theorem}{0}

We now prove that for the times in which the second alternative in Lemma \ref%
{alt} holds (cf. \ref{Q2}), it is possible to estimate the total number of
triples contained in $\left[ 0,R\right] ^{3}$ by means of the triples which
are separated from the diagonal $\left\{
\epsilon_{1}=\epsilon_{2}=\epsilon_{3}\right\} .$

\begin{lemma}
\label{Compl}Let $0<\delta<\frac{2}{3},$ $0<\rho<1.$ For any $R\in\left(
0,1\right) $ we define $\mathcal{S}_{R,\rho}$ as in (\ref{B5b}). Let us
assume also that $b=\frac{1}{\left( 1-\rho\right) }.$ Then there exists $%
\nu=\nu\left( \delta\right) >0$ independent on $R$ and $\rho$ such that, for
any $g\in\mathcal{M}^{+}\left[ 0,R\right] $ satisfying $\int_{\left\{
0\right\} }gd\epsilon=0$ if the alternative (ii) in Lemma \ref{altresc}
takes place we have:%
\begin{equation}
\int_{\mathcal{S}_{bR,\rho}}\left[ \prod_{m=1}^{3}\,g_{m}d\mathcal{\epsilon }%
_{m}\right] \geq\nu\left( \int_{\left[ 0,R\right] }gd\epsilon\right) ^{3}\ . 
\label{H1bis}
\end{equation}
\end{lemma}

\begin{proof}
Using the second inclusion in (\ref{B6}) of Lemma \ref{subsets} we obtain:%
\begin{equation}
\int_{\mathcal{S}_{bR,\rho}}\left[ \prod_{m=1}^{3}\,g_{m}d\mathcal{\epsilon }%
_{m}\right] \geq\sum_{\sigma\in S^{3}}\left[ \sum_{N\left( R,b\right) \leq
j\leq\ell<k-1}\int_{\mathcal{I}_{\sigma\left( j,\ell,k\right) }\left(
b,R\right) }\left[ \prod_{m=1}^{3}\,g_{m}d\mathcal{\epsilon}_{m}\right] %
\right] ,   \label{S2}
\end{equation}
where: 
\begin{equation*}
\mathcal{I}_{\sigma\left( j,\ell,k\right) }\left( b,R\right) =\mathcal{I}%
_{\sigma\left( j\right) }\left( b,R\right) \times \mathcal{I}_{\sigma\left(
\ell\right) }\left( b,R\right) \times \mathcal{I}_{\sigma\left( k\right)
}\left( b,R\right) . 
\end{equation*}

We define the action of the permutations semigroup $S^{3}$ in $\left[ 0,R%
\right] ^{3}$ by means of the mapping $\left( \epsilon_{1},\epsilon
_{2},\epsilon_{3}\right) \rightarrow\left( \epsilon_{\sigma\left( 1\right)
},\epsilon_{\sigma\left( 2\right) },\epsilon_{\sigma\left( 3\right) }\right)
.$\ Notice that the subsets of the family $\left\{ \mathcal{I}_{k}\left(
b,R\right) \right\} $ are ordered by means of the order relation which says
that $\mathcal{J}_{m}<\mathcal{J}_{k}$ if for any $\epsilon_{1}\in\mathcal{J}%
_{m}$ and $\epsilon_{2}\in\mathcal{J}_{k}$ we have that $\epsilon_{1}<%
\epsilon_{2}.$

We now restrict ourselves to the case in which at least two of the intervals
are contained in $\mathcal{U}_{1}$. The third one can be either in $\mathcal{%
U}_{1}$ or $\mathcal{U}_{2}.$ More precisely, we define the following sets:%
\begin{equation*}
\mathcal{Y}_{0}=\left[ \bigcup_{\left[ \mathcal{J}_{m_{1}}<\mathcal{J}%
_{m_{2}}\leq\mathcal{J}_{m_{3}};\mathcal{J}_{m_{1}}\subset\mathcal{U}_{1},%
\mathcal{J}_{m_{2}}\subset\mathcal{U}_{2},\mathcal{J}_{m_{3}}\subset\mathcal{%
U}_{2}\right] }\bigcup_{m_{1}}\bigcup_{m_{2}}\bigcup_{m_{3}}\sigma\left( 
\mathcal{J}_{m_{1}}\times\mathcal{J}_{m_{2}}\times \mathcal{J}%
_{m_{3}}\right) \right] 
\end{equation*}
\ 
\begin{equation*}
\mathcal{Y}_{1}=\left[ \bigcup_{\left[ \mathcal{J}_{m_{1}}<\mathcal{J}%
_{m_{2}}\leq\mathcal{J}_{m_{3}};\mathcal{J}_{m_{1}}\subset\mathcal{U}_{2},%
\mathcal{J}_{m_{2}}\subset\mathcal{U}_{1},\mathcal{J}_{m_{3}}\subset\mathcal{%
U}_{1}\right] }\bigcup_{m_{1}}\bigcup_{m_{2}}\bigcup_{m_{3}}\sigma\left( 
\mathcal{J}_{m_{1}}\times\mathcal{J}_{m_{2}}\times \mathcal{J}%
_{m_{3}}\right) \right] 
\end{equation*}%
\begin{equation*}
\mathcal{Y}_{2}=\left[ \bigcup_{\left[ \mathcal{J}_{m_{1}}<\mathcal{J}%
_{m_{2}}\leq\mathcal{J}_{m_{3}};\mathcal{J}_{m_{1}}\subset\mathcal{U}_{1},%
\mathcal{J}_{m_{2}}\subset\mathcal{U}_{2},\mathcal{J}_{m_{3}}\subset\mathcal{%
U}_{1}\right] }\bigcup_{m_{1}}\bigcup_{m_{2}}\bigcup_{m_{3}}\sigma\left( 
\mathcal{J}_{m_{1}}\times\mathcal{J}_{m_{2}}\times \mathcal{J}%
_{m_{3}}\right) \right] 
\end{equation*}%
\begin{equation*}
\mathcal{Y}_{3}=\left[ \bigcup_{\left[ \mathcal{J}_{m_{1}}<\mathcal{J}%
_{m_{2}}\leq\mathcal{J}_{m_{3}};\mathcal{J}_{m_{1}}\subset\mathcal{U}_{1},%
\mathcal{J}_{m_{2}}\subset\mathcal{U}_{1},\mathcal{J}_{m_{3}}\subset\mathcal{%
U}_{2}\right] }\bigcup_{m_{1}}\bigcup_{m_{2}}\bigcup_{m_{3}}\sigma\left( 
\mathcal{J}_{m_{1}}\times\mathcal{J}_{m_{2}}\times \mathcal{J}%
_{m_{3}}\right) \right] 
\end{equation*}

\begin{equation}
\mathcal{Y}=\mathcal{Y}_{0}\cup\mathcal{Y}_{1}\cup\mathcal{Y}_{2}\cup%
\mathcal{Y}_{3}.   \label{S1E7}
\end{equation}

We now claim that the set $\mathcal{Y}$ is contained in the union of sets $%
\mathcal{I}_{\sigma\left( j,\ell,k\right) }\left( b,R\right) $ appearing in
the right-hand side of (\ref{S2}). (Notice that we impose there that \hfill \break $%
N\left( R,b\right) \leq j\leq\ell<k-1$ in that union). Indeed, we have to
consider several possibilities. Suppose first that $\mathcal{J}_{m_{1}}\in%
\mathcal{U}_{1}$ and $\mathcal{J}_{m_{2}}\in\mathcal{U}_{2}.$ Then we have $%
\mathcal{J}_{m_{1}}\cap\mathcal{J}_{m_{2}}^{\left( E\right) }=\varnothing$
whence we would have $\ell<k-1$ as in (\ref{S2}). If $\mathcal{J}_{m_{1}}\in%
\mathcal{U}_{2}$ and $\mathcal{J}_{m_{2}}\in\mathcal{U}_{1}$ we argue
similarly. Suppose finally that $\mathcal{J}_{m_{1}}\in\mathcal{U}_{1}$ and $%
\mathcal{J}_{m_{2}}\in\mathcal{U}_{1}.$ Then, the property (\ref{Z1E5b})
yields also $\mathcal{J}_{m_{1}}\cap\mathcal{J}_{m_{2}}^{\left( E\right)
}=\varnothing.$

Therefore:%
\begin{equation*}
\int_{\mathcal{Y}}\left[ \prod_{m=1}^{3}\,g_{m}d\mathcal{\epsilon}_{m}\right]
\leq\sum_{\sigma\in S^{3}}\left[ \sum_{N\left( R,b\right) \leq
j\leq\ell<k-1}\int_{\mathcal{I}_{\sigma\left( j,\ell,k\right) }\left(
b,R\right) }\left[ \prod_{m=1}^{3}\,g_{m}d\mathcal{\epsilon}_{m}\right] %
\right] , 
\end{equation*}
whence, using also (\ref{S2}):%
\begin{equation}
\int_{\mathcal{Y}}\left[ \prod_{m=1}^{3}\,g_{m}d\mathcal{\epsilon}_{m}\right]
\leq\int_{\mathcal{S}_{bR,\rho}}\left[ \prod_{m=1}^{3}\,g_{m}d\mathcal{%
\epsilon}_{m}\right] .   \label{S1E6}
\end{equation}

Notice that, due to (\ref{Z1E5}), since alternative (ii) holds, we have:%
\begin{equation}
\left( \int_{\left[ 0,R\right] }gd\epsilon\right) ^{3}\leq\frac{1}{\eta^{3}}%
\left( \int_{\mathcal{U}_{1}}gd\epsilon\right) ^{2}\int _{\mathcal{U}%
_{2}}gd\epsilon, \   \label{Z1E6a}
\end{equation}
where $\eta$ is as in Lemma \ref{altresc}.

We then need to prove that it is possible to estimate the right-hand side of
(\ref{Z1E6a}) by means of the integral on the left-hand side of (\ref{S1E6}%
). In order to check this we just notice that the product $\left( \int _{%
\mathcal{U}_{1}}gd\epsilon\right) ^{2}\int_{\mathcal{U}_{2}}gd\epsilon$ can
be written as the sum:%
\begin{equation*}
\left[ \sum_{\mathcal{J}_{m_{1}}\in\mathcal{U}_{1},\mathcal{J}_{m_{2}}\in%
\mathcal{U}_{1}}\left( \int_{\mathcal{J}_{m_{1}}}gd\epsilon\right) \left(
\int_{\mathcal{J}_{m_{2}}}gd\epsilon\right) \right] \int _{\mathcal{U}%
_{2}}gd\epsilon. 
\end{equation*}

We decompose this sum in two types of terms, namely:%
\begin{equation*}
S_{1}=\left[ \sum_{\mathcal{J}_{m_{1}}\in\mathcal{U}_{1},\mathcal{J}%
_{m_{3}}\in\mathcal{U}_{2}}\left( \int_{\mathcal{J}_{m_{1}}}gd\epsilon%
\right) ^{2}\left( \int_{\mathcal{J}_{m_{3}}}gd\epsilon\right) \right]
=\sum_{\mathcal{J}_{m_{1}}\in\mathcal{U}_{1}}\left( \int_{\mathcal{J}%
_{m_{1}}}gd\epsilon\right) ^{2}\int_{\mathcal{U}_{2}}gd\epsilon, 
\end{equation*}%
\begin{equation*}
S_{2}=\left[ \sum_{\mathcal{J}_{m_{1}}\in\mathcal{U}_{1},\mathcal{J}%
_{m_{2}}\in\mathcal{U}_{1},\mathcal{J}_{m_{3}}\in\mathcal{U}_{2};\mathcal{J}%
_{m_{1}}\neq\mathcal{J}_{m_{2}}}\left( \int_{\mathcal{J}_{m_{1}}}gd\epsilon%
\right) \left( \int_{\mathcal{J}_{m_{2}}}gd\epsilon\right) \left( \int _{%
\mathcal{J}_{m_{3}}}gd\epsilon\right) \right] . 
\end{equation*}

The integral $S_{2}$ can be estimated by means of $\int_{\mathcal{Y}}\left[
\prod_{m=1}^{3}\,g_{m}d\mathcal{\epsilon}_{m}\right] .$ To check this we
just notice that all the sets of the form $\mathcal{J}_{m_{1}}\times\mathcal{%
J}_{m_{2}}\times\mathcal{J}_{m_{3}}$ with $\mathcal{J}_{m_{1}}\in\mathcal{U}%
_{1},\mathcal{J}_{m_{2}}\in\mathcal{U}_{1},\mathcal{J}_{m_{3}}\in \mathcal{U}%
_{2},$ $\mathcal{J}_{m_{1}}\neq\mathcal{J}_{m_{2}}$ are contained in $%
\mathcal{Y}_{1}\cup\mathcal{Y}_{2}\cup\mathcal{Y}_{3}.$ Therefore:%
\begin{equation}
S_{2}\leq\int_{\mathcal{Y}}\left[ \prod_{m=1}^{3}\,g_{m}d\mathcal{\epsilon }%
_{m}\right] .   \label{S3}
\end{equation}

It only remains to estimate $S_{1}.$ To this end we use (\ref{Z1E5c}). Then: 
\begin{align}
S_{1}&\leq\left( \int_{\mathcal{I}_{k_{1}}\left( b,R\right) }gd\epsilon
\right) ^{2}\int_{\mathcal{U}_{2}}gd\epsilon+\sum_{\mathcal{J}_{m_{3}}\in%
\mathcal{U}_{2}}\sum_{j=2}^{L}\int_{\mathcal{I}_{k_{1}}\left( b,R\right)
}gd\epsilon\int_{\mathcal{I}_{k_{j}}\left( b,R\right) }gd\epsilon \int_{%
\mathcal{J}_{m_{3}}}gd\epsilon \notag\\
&\equiv S_{1,1}+S_{1,2},   \label{S4}
\end{align}
where the meaning of the sets $\mathcal{I}_{k_{j}}\left( b\right) $ is the
same as in the Proof of Lemma \ref{alt}.

The term $S_{1,2}$ in (\ref{S4}) consists of the sum of integrals in sets of
the form $\mathcal{I}_{k_{1}}\left( b,R\right) \times\mathcal{I}%
_{k_{j}}\left( b,R\right) \times\mathcal{J}_{m_{3}}$ with $j=2,3,...$ and $%
\mathcal{J}_{m_{3}}\in\mathcal{U}_{2}.$ These sets are contained in $%
\mathcal{Y}_{1}\cup\mathcal{Y}_{2}\cup\mathcal{Y}_{3}.$ Then:%
\begin{equation}
S_{1,2}\leq\int_{\mathcal{Y}}\left[ \prod_{m=1}^{3}\,g_{m}d\mathcal{\epsilon 
}_{m}\right] .   \label{S5}
\end{equation}

We now estimate the term $S_{1,1}.$ We write $\mathcal{U}_{2}=\bigcup _{m}%
\mathcal{I}_{m}^{\ast}.$ Due to (\ref{Z1E5}), as well as the fact that $%
\mathcal{I}_{m}^{\ast}\cap\mathcal{I}_{k_{1}}^{\left( E\right) }\left(
b,R\right) =\varnothing$, we have at least one of the two following
possibilities:%
\begin{equation}
\int_{\bigcup_{m}\mathcal{I}_{m}^{\ast};\mathcal{I}_{m}^{\ast}<\mathcal{I}%
_{k_{1}}\left( b,R\right) }gd\epsilon\geq\frac{1}{2}\int_{\mathcal{U}%
_{2}}gd\epsilon\   \label{Z1E7}
\end{equation}
or:%
\begin{equation}
\int_{\bigcup_{m}\mathcal{I}_{m}^{\ast};\mathcal{I}_{m}^{\ast}>\mathcal{I}%
_{k_{1}}\left( b,R\right) }gd\epsilon\geq\frac{1}{2}\int_{\mathcal{U}%
_{2}}gd\epsilon.   \label{Z1E8}
\end{equation}

Suppose that (\ref{Z1E7}) takes place. Then:%
\begin{equation*}
S_{1,1}\leq2\left( \int_{\mathcal{I}_{k_{1}}\left( b,R\right)
}gd\epsilon\right) ^{2}\int_{\bigcup_{m}\mathcal{I}_{m}^{\ast};\mathcal{I}%
_{m}^{\ast}<\mathcal{I}_{k_{1}}\left( b,R\right) }gd\epsilon. 
\end{equation*}

The right-hand side of this inequality can be estimated by $2\int _{\mathcal{%
Y}_{1}}\left[ \prod_{m=1}^{3}\,g_{m}d\mathcal{\epsilon}_{m}\right] ,$ since
it is possible to write the term on the right as the sum of integrals on
sets with the form $\mathcal{I}_{m}^{\ast}\times \mathcal{I}_{k_{1}}\left(
b,R\right) \times\mathcal{I}_{k_{1}}\left( b,R\right) .$

Suppose now that we have (\ref{Z1E8}). Combining this formula with (\ref%
{Z1E5d}) and (\ref{Z1E5}) we obtain:%
\begin{equation*}
\int_{\mathcal{I}_{k_{1}}\left( b,R\right) }gd\epsilon\leq\left(
1-\delta\right) \int_{\left[ 0,R\right] }gd\epsilon\leq\frac{\left(
1-\delta\right) }{\eta}\int_{\mathcal{U}_{2}}gd\epsilon\leq\frac{2\left(
1-\delta\right) }{\eta}\int_{\bigcup_{m}\mathcal{I}_{m}^{\ast};\mathcal{I}%
_{m}^{\ast}>\mathcal{I}_{k_{1}}\left( b,R\right) }gd\epsilon, 
\end{equation*}
whence:%
\begin{align*}
S_{1,1} & \leq\frac{2\left( 1-\delta\right) }{\eta}\left( \int _{\mathcal{I}%
_{k_{1}}\left( b,R\right) }gd\epsilon\right) \left( \int_{\mathcal{U}%
_{2}}gd\epsilon\right) \left( \int_{\bigcup_{m}\mathcal{I}_{m}^{\ast};%
\mathcal{I}_{m}^{\ast}>\mathcal{I}_{k_{1}}\left( b,R\right)
}gd\epsilon\right) \\
& \leq\frac{4\left( 1-\delta\right) }{\eta}\left( \int_{\mathcal{I}%
_{k_{1}}\left( b,R\right) }gd\epsilon\right) \left( \int_{\bigcup _{m}%
\mathcal{I}_{m}^{\ast};\mathcal{I}_{m}^{\ast}>\mathcal{I}_{k_{1}}\left(
b,R\right) }gd\epsilon\right) ^{2}.
\end{align*}
The right hand side of this inequality can be estimated by $\frac{4\left(
1-\delta\right) }{\eta}\int_{\mathcal{Y}_{0}}\left[ \prod_{m=1}^{3}\,g_{m}d%
\mathcal{\epsilon}_{m}\right] .$ We have then obtained that:%
\begin{equation*}
S_{1,1}\leq\max\left\{ 2,\frac{4\left( 1-\delta\right) }{\eta}\right\} \int_{%
\mathcal{Y}}\left[ \prod_{m=1}^{3}\,g_{m}d\mathcal{\epsilon}_{m}\right] , 
\end{equation*}
which combined with (\ref{S5}) yields:%
\begin{equation*}
S_{1}\leq\max\left\{ 3,\frac{4\left( 1-\delta\right) }{\eta}\right\} \int_{%
\mathcal{Y}}\left[ \prod_{m=1}^{3}\,g_{m}d\mathcal{\epsilon}_{m}\right] . 
\end{equation*}

Combining this formula with (\ref{S3}) \ and using (\ref{S1E6}) we conclude
the Proof of the Lemma.
\end{proof}

\section{Estimating the rate of formation of particles with small energy.}

\setcounter{equation}{0} \setcounter{theorem}{0}

In this Section we prove several estimates whose meaning is the following.
If we have a weak solution $f$ of (\ref{F3E2}), (\ref{F3E3})   on $0\leq t\leq T_{0}$ in the sense of
Definition \ref{weakf} such that $g(t)=4\pi\sqrt{2\epsilon} f(t)$ satisfies condition (\ref{Enodirac}) for all $0\le t\le T_0$ for some
suitable $T_{0}$, then either the alternative (ii) in Lemma \ref{altresc} takes
place during most of the time for small $R,$ something that contradicts (\ref%
{H1}), or the alternative (i) in Lemma \ref{altresc} takes place for most times with $b$
sufficiently close to one. In this second case, if the initial density of
particles is not too small near $\epsilon=0,$ there would be a large
transfer of particles towards small energies and this would contradict the
conservation of the total number of particles. The consequence of this
contradiction is that the maximal time of existence for the solution $f$
must be smaller than $T_{0}.$

The precise way of obtaining this contradiction is to estimate the measure
of some subsets of $\left[ 0,T_{0}\right] $ for which precise information
about the concentration properties of $g$ over them are available. We will
then prove that the total measure of these sets, which cover the whole
interval $\left[ 0,T_{0}\right] ,$ is strictly smaller than $T_{0}.$

\subsection{Defining some subsets of $\left[ 0,T_{0}\right] .$}

In the remaining of this Section we assume that $f$ is a weak solution of (\ref{F3E2}) (\ref%
{F3E3}) on $(0, T)$ in the sense of Definition \ref{weakf}, with
initial data $f_{0}$ satisfying (\ref{C1}) and  consider $g$ as defined by (\ref{F3E3a}).

For all $n=0,1,2,...$ let $R_{n}=2^{-n}$. For any $\theta_{1}>0,\ \theta
_{2}>0$ and $0\le T_{0}<T_{max}$ we define the following sets:

\begin{equation}
B_{\ell}=\left\{ t\in\left[ 0,T_{0}\right] :\int_{\left[ 0,R_{\ell }\right]
}g\left( \epsilon,t\right) d\epsilon\geq\left( R_{\ell}\right)
^{\theta_{1}}\right\} \ \ ,\ \ell=0,1,2,...   \label{F1E2}
\end{equation}

We also define the sequence $b_{\ell}=1+\left( R_{\ell}\right)
^{\theta_{2}},\ \ell=0,1,2,...$ and the sets:%
\begin{equation}
A_{n,\ell}=\left\{ t\in\left[ 0,T_{0}\right] :\text{ such that }\int_{%
\mathcal{I}_{n}^{\left( E\right) }\left( b_{\ell},R_{\ell}\right) }g\left(
t,\epsilon\right) d\epsilon\geq\left( R_{\ell+1}\right)
^{\theta_{1}}\right\} , \   \label{F1E3}
\end{equation}
for $\ell=0,1,2,..,$ $n=1,2,3,...$ We recall that $\mathcal{I}_{n}\left(
b_{\ell},R_{\ell}\right) $ has been defined in (\ref{Z1E1}).

Notice that we have $\mathcal{I}_{n}^{\left( E\right) }\left( b_{\ell
},R_{\ell}\right) \subset\left[ 0,R_{\ell}\right] $ for all $n=1,2,...$ .
This is the motivation of the definitions of the sets above.

The following result is basically a consequence of Lemma \ref{estProd} and
Lemma \ref{alt}.

\begin{lemma}
\label{MeasureOmega} Let   $f$  be a weak solution of (\ref{F3E2}), (\ref{F3E3}) on $[0, T]$ with initial data $f_0$ such that
$g_0=4\pi\sqrt{2\epsilon} f_{0}\in%
\mathcal{M}_{+}\left( \mathbb{R}^{+};1+\epsilon\right)$ and satisfying (\ref{C1}). Suppose also that $g(t)=4\pi\sqrt{2\epsilon} f(t)$ satisfies  condition (\ref{Enodirac}) for all $t\in [0, T]$.  Given $%
\theta_{1}>0,\ \theta_{2}>0,$ let us define the sets $B_{n},\ A_{n,\ell}$ as
in (\ref{F1E2}), (\ref{F1E3}) and $\Omega_{\ell}$ as:%
\begin{equation*}
\Omega_{\ell}=B_{\ell}\setminus\bigcup_{n\geq1}A_{n,\ell}\subset\left[
0,T\right] \ \ ,\ \ \ell=0,1,2,... 
\end{equation*}

Then, there exists $\theta_{0}>0$ such that, if $\min\left\{ \theta
_{1},\theta_{2}\right\} <\theta_{0}$, we have: 
\begin{equation*}
\left\vert \Omega_{\ell}\right\vert \leq K\left( 1+T\right) R_{\ell
}^{1-3\theta_{1}-4\theta_{2}}, 
\end{equation*}
for some $K=K\left( M,\theta_{1}\right) $ and for any $\ell=0,1,2,..\ $.
\end{lemma}

\begin{proof}
We apply Lemma \ref{altresc} with $\left( 1-\delta\right) =2^{-\theta_{1}}$
and $b=b_{\ell}.$ The definitions of the sets $B_{\ell}$ and $A_{n,\ell}$
show that $\Omega_{\ell}$ is the set of times $t$ in $\left[ 0,T
\right] $ for which the alternative (i) in Lemma \ref{altresc} does not take
place. Therefore, the alternative (ii) takes place. We can then apply, for
such times, Lemma \ref{Compl} which combined with Lemma \ref{estProd} (cf.
also (\ref{F1E2}))\ gives the following estimate:%
\begin{equation*}
\left( R_{\ell}\right)
^{3\theta_{1}}\int_{0}^{T_{\max}}\chi_{\Omega_{\ell}}dt\leq\frac{2b_{\ell}^{%
\frac{7}{2}}b_{\ell}R_{\ell}}{B\nu\rho_{\ell}^{2}\left( \sqrt{b_{\ell}}%
-1\right) ^{2}}\left[ 2\pi\int_{0}^{T_{\max}}dt\left( \int_{\left[ 0,1\right]
}g\left( \epsilon\right) d\epsilon \right) ^{2}+M\right] , 
\end{equation*}
where $\rho_{\ell}$ is related with $b_{\ell}$ as in Lemma \ref{Compl},
whence $\rho_{\ell}=1-\frac{1}{b_{\ell}}.$ We estimate the terms between
brackets in the right-hand side in terms of the total number of particles.
Therefore, using Taylor's expansion, it follows that there exists $K=K\left(
M,\theta_{1}\right) $ such that:%
\begin{equation*}
\int_{0}^{T_{\max}}\chi_{\Omega_{\ell}}dt\leq
KR_{\ell}^{1-3\theta_{1}-4\theta_{2}}\left( 1+T\right) 
\end{equation*}
whence the result follows.
\end{proof}

We define a new family of sets $\mathcal{A}_{\ell}$ by means of:%
\begin{equation*}
\mathcal{A}_{\ell}=\bigcup_{n=1}^{\left[ \frac{\log\left( 2\right) }{%
\log\left( b_{\ell}\right) }\right] +1}A_{n,\ell}. 
\end{equation*}

\begin{lemma}
\label{Inters}Under the assumptions of Lemma \ref{MeasureOmega} we have:%
\begin{equation*}
\left( B_{\ell}\setminus B_{\ell+1}\right) \cap\left( \bigcup_{n\geq
1}A_{n,\ell}\diagdown\mathcal{A}_{\ell}\right) =\varnothing, 
\end{equation*}
for $\ell=0,1,2,...$.
\end{lemma}

\begin{proof}
Notice that for $n\geq\left[ \frac{\log\left( 2\right) }{\log\left(
b_{\ell}\right) }\right] +2>\frac{\log\left( 2\right) }{\log\left(
b_{\ell}\right) }+1$ the extended intervals $\mathcal{I}_{n}^{\left(
E\right) }\left( b_{\ell},R_{\ell}\right) $ which appear in the definition
of the sets $A_{n,\ell}$ are contained in 
$$\left\{ \epsilon\leq b_{\ell
}^{-\left( \frac{\log\left( 2\right) }{\log\left( b_{\ell}\right) }+1\right)
}R_{\ell}=\frac{b_{\ell}^{-1}}{2}R_{\ell}<R_{\ell+1}\right\}.$$ Then, if $%
t\in\bigcup_{n\geq1}A_{n,\ell}\diagdown\mathcal{A}_{\ell}$ we have:%
\begin{equation*}
\int_{\left[ 0,R_{\ell+1}\right] }g\left( \epsilon,t\right) d\epsilon
\geq\int_{\mathcal{I}_{n_{0}}^{\left( E\right) }\left( b_{\ell},R_{\ell
}\right) }g\left( \epsilon,t\right) d\epsilon, 
\end{equation*}
for some $n_{0}\geq\left[ \frac{\log\left( 2\right) }{\log\left( b_{\ell
}\right) }\right] +2.$ Therefore, due to the definition of $A_{n_{0},\ell}:$%
\begin{equation*}
\int_{\left[ 0,R_{\ell+1}\right] }g\left( \epsilon,t\right) d\epsilon
\geq\left( R_{\ell+1}\right) ^{\theta_{1}}. 
\end{equation*}
On the other hand, if $t\in\left( B_{\ell}\setminus B_{\ell+1}\right) $ we
have:%
\begin{equation*}
\int_{\left[ 0,R_{\ell}\right] }g\left( \epsilon,t\right) d\epsilon
\geq\left( R_{\ell}\right) ^{\theta_{1}}\ \ ,\ \ \int_{\left[ 0,R_{\ell +1}%
\right] }g\left( \epsilon,t\right) d\epsilon<\left( R_{\ell+1}\right)
^{\theta_{1}}, 
\end{equation*}
but this gives a contradiction unless $\left( B_{\ell}\setminus B_{\ell
+1}\right) \cap\left( \bigcup_{n\geq1}A_{n,\ell}\diagdown\mathcal{A}_{\ell
}\right) =\varnothing.$
\end{proof}

\begin{lemma}
\label{MeasDiff} Under the assumptions of Lemma \ref{MeasureOmega} we have:%
\begin{equation*}
\left\vert \left( B_{\ell}\setminus B_{\ell+1}\right) \setminus \mathcal{A}%
_{\ell}\right\vert \leq K\left( R_{\ell}\right) ^{\alpha}, 
\end{equation*}
where $K=K\left( E,M,\theta_{1}\right) $ and $\alpha$ are as in Lemma \ref%
{MeasureOmega}.
\end{lemma}

\begin{proof}
Due to Lemma \ref{Inters} we have:%
\begin{equation*}
\left( B_{\ell}\setminus B_{\ell+1}\right) \setminus\mathcal{A}_{\ell
}=\left( B_{\ell}\setminus B_{\ell+1}\right)
\setminus\bigcup_{n=1}A_{n,\ell}. 
\end{equation*}

Then, since $\left( B_{\ell}\setminus B_{\ell+1}\right) \setminus
\bigcup_{n\geq1}A_{n,\ell}\subset
B_{\ell}\setminus\bigcup_{n\geq1}A_{n,\ell}, $ we obtain:%
\begin{equation*}
\left( B_{\ell}\setminus B_{\ell+1}\right) \setminus\mathcal{A}_{\ell
}\subset B_{\ell}\setminus\bigcup_{n\geq\ell}A_{n,\ell}=\Omega_{\ell}. 
\end{equation*}

Using Lemma \ref{MeasureOmega} the result follows.
\end{proof}

We now proceed to estimate $\left\vert \mathcal{A}_{\ell}\right\vert .$ This
is the crucial step where the properties of the kinetic equation are used.
More precisely, we derive some detailed estimates for the lifetime of the
possible concentrations of mass of $g$ at regions of order $R_{\ell}.$ These
estimates will be obtained using suitable test functions that solve some
kind of adjoint equation of (\ref{F3E2}), (\ref{F3E3}). The choice of these
test functions is made in order to show that, if the measure $g$ is very
concentrated, then the particles transported towards smaller sizes remain
there for sufficiently long times. As a preliminary step we describe the
construction of the test function.

We need to introduce some additional notation. Given $t\in\mathcal{A}_{\ell}$
there exists at least one integer $N=N\left( t\right) \in\left\{ 1,...,\left[
\frac{\log\left( 2\right) }{\log\left( b_{\ell}\right) }\right] +1\right\} $
such that $\int_{\mathcal{I}_{N\left( t\right) }^{\left( E\right) }\left(
b_{\ell},R_{\ell}\right) }g\left( t,\epsilon\right) d\epsilon\geq\left(
R_{\ell+1}\right) ^{\theta_{1}}.$ If different possible choices exist, it is
possible to define a measurable function $N\left( t\right) $ with this
property.

We then have the following result:

\begin{lemma}
\label{phi}Suppose that the assumptions of Lemma \ref{MeasureOmega} hold.
Given $\theta_{1}>0,\ \theta_{2}>0$ such that $\left(
1-2\theta_{1}-\theta_{2}\right) >0$ we define the sets $B_{n},\ A_{n,\ell}$
as in (\ref{F1E2}), (\ref{F1E3}). Let us assume that there exists $\tilde{T}%
_{0}\in\left[ 0,T\right] $ such that%
\begin{equation*}
\int_{0}^{\tilde{T}_{0}}\chi_{\mathcal{A}_{\ell}}\left( t\right) \left(
\int_{\mathcal{I}_{N\left( t\right) }^{\left( E\right) }\left( b_{\ell
},R_{\ell}\right) }g\left( t,\epsilon\right) d\epsilon\right)
^{2}dt=K_{2}\left( R_{\ell}\right) ^{1-\theta_{2}}, 
\end{equation*}
with 
\begin{equation}
K_{2}=\left( \frac{\sqrt{2}-1}{2}\right) .   \label{F2E2a}
\end{equation}

Then, there exists a function $\varphi\in L^{\infty}\left( \left[ 0,\tilde{T}%
_{0}\right] ,C^{1}\left( \mathbb{R}^{+}\right) \right) $ satisfying the
following properties:

(i) $0\leq\varphi\left( t,\epsilon\right) \leq1$ for $\left( t,\epsilon
\right) \in\left[ 0,\tilde{T}_{0}\right] \times\mathbb{R}^{+}.$

(ii) $\varphi\left( t,\cdot\right) $ is convex in $\mathbb{R}^{+}$ for each $%
t\in\left[ 0,\tilde{T}_{0}\right] .$

(iii) $\operatorname*{supp}\left( \varphi\left( t,\cdot\right) \right) \subset%
\left[ 0,\frac{R_{\ell}}{4}\right] $ for each $t\in\left[ 0,\tilde{T}_{0}%
\right] .$

(iv) $\varphi\left( \epsilon,t\right) \geq\frac{1}{2}$ \ for \ $0\leq
\epsilon\leq\left( \frac{\sqrt{2}-1}{\sqrt{2}}\right) \frac{R_{\ell}}{8}\ \
,\ \ 0\leq t\leq\tilde{T}_{0}.$

(v) The following inequality holds for $0\leq t\leq\tilde{T}_{0},\
\epsilon\geq0:$%
\begin{equation}
\partial_{t}\varphi+\frac{\chi_{\mathcal{A}_{\ell}}\left( t\right) }{2^{%
\frac{3}{2}}R_{\ell}}\int\int_{\left\{ \epsilon_{2}\leq\epsilon _{3},\
\epsilon_{2},\epsilon_{3}\in\mathcal{I}_{N\left( t\right) }^{\left( E\right)
}\left( b_{\ell},R_{\ell}\right) \right\} }g_{2}g_{3}\left[ \varphi\left(
\epsilon_{1}+\epsilon_{3}-\epsilon_{2}\right) -\varphi\left(
\epsilon_{1}\right) \right] d\epsilon_{2}d\epsilon_{3}\geq0.   \label{F2E3}
\end{equation}
\end{lemma}

\begin{proof}
We define the functions:%
\begin{equation*}
\Omega\left( t;\tilde{T}_{0}\right) =\int_{t}^{\tilde{T}_{0}}\chi _{\mathcal{%
A}_{\ell}}\left( s\right) \omega\left( s\right) ds, 
\end{equation*}%
\begin{equation}
\omega\left( t\right) =\frac{1}{2^{\frac{3}{2}}R_{\ell}}\int\int_{\left\{
\epsilon_{2}\leq\epsilon_{3},\ \epsilon_{2},\epsilon_{3}\in\mathcal{I}%
_{N\left( t\right) }^{\left( E\right) }\left( b_{\ell},R_{\ell}\right)
\right\} }\left( \epsilon_{3}-\epsilon_{2}\right) g_{2}g_{3}d\epsilon
_{2}d\epsilon_{3},   \label{F2E7}
\end{equation}%
\begin{equation*}
\Psi\left( \zeta\right) =\frac{16}{\left( R_{\ell}\right) ^{2}}\left[ \left( 
\frac{R_{\ell}}{4}-\zeta\right) _{+}\right] ^{2}, 
\end{equation*}
where $\left( s\right) _{+}=\max\left\{ s,0\right\} .$ We then define the
function $\varphi\left( \epsilon,t\right) $ by means of the formula:%
\begin{equation}
\varphi\left( t,\epsilon\right) =\Psi\left( \epsilon+\Omega\left( t;\tilde{T}%
_{0}\right) \right) .   \label{F2E7a}
\end{equation}

Properties (i), (ii), (iii) can be immediately checked. In order to check
(iv) we notice that the definition of $K_{2}$ (cf. \ref{F2E2a}) implies:%
\begin{equation}
\Omega\left( t;\tilde{T}_{0}\right) \leq\left( \frac{\sqrt{2}-1}{\sqrt{2}}%
\right) \frac{R_{\ell}}{8}\ \ ,\ \ 0\leq t\leq\tilde{T}_{0}   \label{F2E4}
\end{equation}

Indeed, we have (cf. (\ref{F2E7})):%
\begin{align}
\Omega\left( t;\tilde{T}_{0}\right) & \leq\frac{\left( b_{\ell}-1\right) }{%
2^{\frac{5}{2}}}\int_{0}^{\tilde{T}_{0}}\chi_{\mathcal{A}_{\ell}}\left(
t\right) \left( \int_{\left\{ \epsilon\in\mathcal{I}_{N\left( t\right)
}^{\left( E\right) }\left( b_{\ell},R_{\ell}\right) \right\} }g\left(
t,\epsilon\right) d\epsilon\right) ^{2}dt  \notag \\
& =\frac{\left( b_{\ell}-1\right) }{2^{\frac{5}{2}}}K_{2}\left( R_{\ell
}\right) ^{1-\theta_{2}}=\frac{1}{2^{\frac{5}{2}}}K_{2}R_{\ell},
\end{align}
where we have used the definition of $\mathcal{A}_{\ell}$ (cf. (\ref{F1E3}))
and $N\left( t\right) $. Therefore:%
\begin{equation*}
\int_{0}^{\tilde{T}_{0}}\chi_{A_{\ell,\ell}}\left( t\right) \left(
\int_{\left\{ \epsilon\in\mathcal{I}_{N\left( t\right) }^{\left( E\right)
}\left( b_{\ell},R_{\ell}\right) \right\} }g\left( t,\epsilon\right)
d\epsilon\right) ^{2}dt=K_{2}\left( R_{\ell}\right) ^{1-\theta_{2}}. 
\end{equation*}

We notice also that $\Psi\left( \zeta\right) \geq\frac{1}{2}$ if $\zeta
\leq\left( \frac{\sqrt{2}-1}{\sqrt{2}}\right) \frac{R_{\ell}}{4}.$ If $%
0\leq\epsilon\leq\left( \frac{\sqrt{2}-1}{\sqrt{2}}\right) \frac{R_{\ell}}{8}
$ we have, using (\ref{F2E4}):%
\begin{equation*}
\epsilon+\Omega\left( t;\tilde{T}_{0}\right) \leq\left( \frac{\sqrt{2}-1}{%
\sqrt{2}}\right) \frac{R_{\ell}}{8}+\left( \frac{\sqrt{2}-1}{\sqrt{2}}%
\right) \frac{R_{\ell}}{8}\leq\left( \frac{\sqrt{2}-1}{\sqrt{2}}\right) 
\frac{R_{\ell}}{4}, 
\end{equation*}
whence (iv) follows.

It only remains to check (v). The convexity of $\varphi\left( t,\cdot\right) 
$ implies:%
\begin{equation*}
\varphi\left( \epsilon+\epsilon_{3}-\epsilon_{2},t\right) -\varphi\left(
\epsilon,t\right) \geq\left( \epsilon_{3}-\epsilon_{2}\right) \frac{%
\partial\varphi}{\partial\epsilon}\left( \epsilon,t\right) \ \ \text{for\ }%
0\leq\epsilon_{2}\leq\epsilon_{3}. 
\end{equation*}

Then, using also (\ref{F2E7}), (\ref{F2E7a}):%
\begin{align*}
& \partial_{t}\varphi\left( \epsilon,t\right) +\frac{1}{2^{\frac{3}{2}%
}R_{\ell}}\int\int_{\left\{ \epsilon_{2}\leq\epsilon_{3},\ \epsilon
_{2},\epsilon_{3}\in\mathcal{I}_{N\left( t\right) }^{\left( E\right) }\left(
b_{\ell},R_{\ell}\right) \right\} }g_{2}g_{3}\left[ \varphi\left(
\epsilon+\epsilon_{3}-\epsilon_{2},t\right) -\varphi\left( \epsilon
,t\right) \right] d\epsilon_{2}d\epsilon_{3} \\
& \geq\partial_{t}\varphi\left( \epsilon,t\right) +\frac{1}{2^{\frac{3}{2}%
}R_{\ell}}\frac{\partial\varphi}{\partial\epsilon}\left( \epsilon,t\right)
\int\int_{\left\{ \epsilon_{2}\leq\epsilon_{3},\ \epsilon_{2},\epsilon_{3}\in%
\mathcal{I}_{N\left( t\right) }^{\left( E\right) }\left( b_{\ell
},R_{\ell}\right) \right\} }g_{2}g_{3}\left( \epsilon_{3}-\epsilon
_{2}\right) d\epsilon_{2}d\epsilon_{3} \\
& =\Psi^{\prime}\left( \epsilon+\Omega\left( t;\tilde{T}_{0}\right) \right) %
\left[ -\omega\left( t\right) +\omega\left( t\right) \right] =0,
\end{align*}
whence the Lemma follows.
\end{proof}

We now prove the following result.

\begin{proposition}
\label{MeasEst}Suppose that the assumptions of Lemma \ref{MeasureOmega}
hold. Given $\theta_{1}>0,\ \theta_{2}>0$ such that $\left(
1-2\theta_{1}-\theta_{2}\right) >0$ we define the sets $B_{n},\ A_{n,\ell}$
as in (\ref{F1E2}), (\ref{F1E3}). Let us assume also that $\nu>0.$ Then,
there exists $\rho=\rho\left( E,M,\nu,\theta_{1},\theta_{2}\right) $\ such
that, if $f\left( \epsilon,0\right) =f_{0}\left( \epsilon\right) \geq\nu$ in 
$\epsilon\in\left[ 0,\rho\right] $ for some $0<\rho<1$ we have:%
\begin{equation}
\left\vert \mathcal{A}_{\ell}\right\vert \leq K_{2}\left( R_{\ell}\right)
^{1-2\theta_{1}-\theta_{2}},\   \label{F2E2}
\end{equation}
for $\ell>\frac{\log\left( \frac{1}{\rho}\right) }{\log\left( 2\right) }$
and where $K_{2}$ is as in (\ref{F2E2a}).
\end{proposition}

\begin{proof}
We must consider separately two cases. Suppose first that:%
\begin{equation}
\int_{0}^{T_{0}}\chi_{\mathcal{A}_{\ell}}\left( t\right) \left(
\int_{\left\{ \epsilon\in\mathcal{I}_{N\left( t\right) }^{\left( E\right)
}\left( b_{\ell},R_{\ell}\right) \right\} }g\left( t,\epsilon\right)
d\epsilon\right) ^{2}dt\leq K_{2}\left( R_{\ell}\right) ^{1-\theta_{2}}, 
\label{F2E2b}
\end{equation}
where $K_{2}$ is as in (\ref{F2E2a}).

Then, using the definition of $\mathcal{A}_{\ell}$ (cf. (\ref{F1E3})) we
obtain:%
\begin{equation*}
\left\vert \mathcal{A}_{\ell}\right\vert \leq K_{2}\left( R_{\ell}\right)
^{1-2\theta_{1}-\theta_{2}}\ 
\end{equation*}
and (\ref{F2E2}) follows.

Let us assume now that:%
\begin{equation*}
\int_{0}^{T_{0}}\chi_{\mathcal{A}_{\ell}}\left( t\right) \left(
\int_{\left\{ \epsilon\in\mathcal{I}_{N\left( t\right) }^{\left( E\right)
}\left( b_{\ell},R_{\ell}\right) \right\} }g\left( t,\epsilon\right)
d\epsilon\right) ^{2}dt>K_{2}\left( R_{\ell}\right) ^{1-\theta_{2}}.\ 
\end{equation*}
\ 

Then, the continuity of the integral with respect to the domain of
integration implies that there exists $\tilde{T}_{0}\in\left[ 0, T
\right] \cap A_{\ell,\ell}$ such that:%
\begin{equation}
\int_{0}^{\tilde{T}_{0}}\chi_{\mathcal{A}_{\ell}}\left( t\right) \left(
\int_{\left\{ \epsilon\in\mathcal{I}_{N\left( t\right) }^{\left( E\right)
}\left( b_{\ell},R_{\ell}\right) \right\} }g\left( t,\epsilon\right)
d\epsilon\right) ^{2}dt=K_{2}\left( R_{\ell}\right) ^{1-\theta_{2}}.\ 
\label{F2E6}
\end{equation}
\ 

Using (\ref{F5E1a}), and symmetrizing the integral containing the cubic
terms in the same way as in the derivation of (\ref{S1E12a}), we obtain: 
\begin{align}
\partial_{t}\left( \int_{\mathbb{R}^{+}}g\varphi d\epsilon\right) & =\int_{%
\mathbb{R}^{+}}g\partial_{t}\varphi d\epsilon+\frac{1}{2^{\frac{5}{2}}}\int_{%
\mathbb{R}^{+}}\int_{\mathbb{R}^{+}}\int_{\mathbb{R}^{+}}\frac {%
g_{1}g_{2}g_{3}\Phi}{\sqrt{\epsilon_{1}\epsilon_{2}\epsilon_{3}}}\mathcal{G}%
_{\varphi}d\epsilon_{1}d\epsilon_{2}d\epsilon_{3}+  \notag \\
& +\frac{\pi}{2}\int_{\mathbb{R}^{+}}\int_{\mathbb{R}^{+}}\int_{\mathbb{R}%
^{+}}\frac{g_{1}g_{2}\Phi}{\sqrt{\epsilon_{1}\epsilon_{2}}}Q_{\varphi
}d\epsilon_{1}d\epsilon_{2}d\epsilon_{3}\ ,\ \ a.e.\ t \in\left[ 0,T
\right] , \   \label{F1E4b}
\end{align}
where $\mathcal{G}_{\varphi}$ is as in (\ref{S1E12bis}), (\ref{S1E12ter})
and, symmetrizing in $\epsilon_{1},\ \epsilon_{2}$ in the quadratic integral
we can take: 
\begin{equation*}
Q_{\varphi}=\left[ \varphi\left( t,\epsilon_{3}\right) +\varphi\left(
t,\epsilon_{1}+\epsilon_{2}-\epsilon_{3}\right) -2\varphi\left(
t,\epsilon_{1}\right) \right] . 
\end{equation*}

Using the symmetry of the function $\mathcal{G}_{\varphi}$ we can write the
cubic term in the equivalent manner:%
\begin{equation*}
\frac{1}{2^{\frac{5}{2}}}\int_{\mathbb{R}^{+}}\int_{\mathbb{R}^{+}}\int_{%
\mathbb{R}^{+}}\frac{g_{1}g_{2}g_{3}\Phi}{\sqrt{\epsilon_{1}\epsilon
_{2}\epsilon_{3}}}\mathcal{G}_{\varphi}d\epsilon_{1}d\epsilon_{2}d\epsilon
_{3}=\frac{6}{2^{\frac{5}{2}}}\int\int\int_{\left\{ \epsilon_{1}\leq
\epsilon_{2}\leq\epsilon_{3}\right\} }\frac{g_{1}g_{2}g_{3}\Phi}{\sqrt{%
\epsilon_{1}\epsilon_{2}\epsilon_{3}}}\mathcal{G}_{\varphi}d\epsilon_{1}d%
\epsilon_{2}d\epsilon_{3}, 
\end{equation*}
with (cf. (\ref{G1E3})):%
\begin{equation*}
\mathcal{G}_{\varphi}=\mathcal{G}_{\varphi}^{\left( 1\right) }+\mathcal{G}%
_{\varphi}^{\left( 2\right) }\ \ \text{in\ \ }\left\{
\epsilon_{1}\leq\epsilon_{2}\leq\epsilon_{3}\right\} , 
\end{equation*}%
\begin{align*}
\mathcal{G}_{\varphi}^{\left( 1\right) }\left( \epsilon_{1},\epsilon
_{2},\epsilon_{3}\right) & =\frac{\sqrt{\epsilon_{1}}}{3}\left[
\varphi\left( \epsilon_{1}+\epsilon_{3}-\epsilon_{2}\right) +\varphi\left(
\epsilon_{3}+\epsilon_{2}-\epsilon_{1}\right) -2\varphi\left( \epsilon
_{3}\right) \right] , \\
\mathcal{G}_{\varphi}^{\left( 2\right) }\left( \epsilon_{1},\epsilon
_{2},\epsilon_{3}\right) & =\frac{\sqrt{\left( \epsilon_{2}+\epsilon
_{1}-\epsilon_{3}\right) _{+}}}{3}\left[ \varphi\left( \epsilon_{3}\right)
+\varphi\left( \epsilon_{2}+\epsilon_{1}-\epsilon_{3}\right) -\varphi\left(
\epsilon_{1}\right) -\varphi\left( \epsilon_{2}\right) \right] ,
\end{align*}
where the dependence of the function $\varphi$ on $t$ will not be made
explicit unless it is needed. We now select the function $\varphi$ as in
Lemma \ref{phi}. Since $\varphi\left( t,\cdot\right) $ is convex we have,
arguing as in the Proof of Proposition \ref{atractiveness}:%
\begin{equation}
\mathcal{G}_{\varphi}^{\left( 1\right) }\left( \epsilon_{1},\epsilon
_{2},\epsilon_{3}\right) \geq0\ \ \ ,\ \ \ \ \mathcal{G}_{\varphi}^{\left(
2\right) }\left( \epsilon_{1},\epsilon_{2},\epsilon_{3}\right) \geq0. 
\label{F1E4a}
\end{equation}

Then, since $g\geq0$ we obtain, using (\ref{F1E4b}):%
\begin{align}
\partial_{t}\left( \int_{\mathbb{R}^{+}}g\varphi d\epsilon\right) &
\geq\int_{\mathbb{R}^{+}}g\partial_{t}\varphi d\epsilon+  \notag \\
&\hskip -0.3cm  +\frac{\chi_{\mathcal{A}_{\ell}}\left( t\right) }{2^{\frac{5}{2}}}%
\int\int\int_{\left\{ \epsilon_{1}\leq\frac{R_{\ell}}{4}\right\} \cap\left\{
\epsilon_{2}\leq\epsilon_{3},\ \epsilon_{2},\epsilon_{3}\in\mathcal{I}%
_{N\left( t\right) }^{\left( E\right) }\left( b_{\ell },R_{\ell}\right)
\right\} }\frac{g_{1}g_{2}g_{3}}{\sqrt{\epsilon _{1}\epsilon_{2}\epsilon_{3}}%
}\mathcal{G}_{\varphi}^{\left( 1\right)
}d\epsilon_{1}d\epsilon_{2}d\epsilon_{3}+  \notag \\
&\hskip -0.3cm +\frac{\pi}{2}\int_{\mathbb{R}^{+}}\int_{\mathbb{R}^{+}}\int_{\mathbb{R}%
^{+}}\frac{g_{1}g_{2}\Phi}{\sqrt{\epsilon_{1}\epsilon_{2}}}Q_{\varphi
}d\epsilon_{1}d\epsilon_{2}d\epsilon_{3}\, ,\ \   \label{F7E1}
\end{align}
$\ a.e.\ t\in\left[ 0,T\right] .$ Using now that $\mathcal{G}%
_{\varphi}^{\left( 1\right) }\left( \epsilon_{1},\epsilon_{2},\epsilon
_{3}\right) \geq\varphi\left( \epsilon_{1}+\epsilon_{3}-\epsilon_{2}\right) 
\sqrt{\epsilon_{1}}.$ We now add and substract $\varphi\left( \epsilon
_{1}\right) \sqrt{\epsilon_{1}}.$ Since $\left( \epsilon_{3}-\epsilon
_{2}\right) >0$ it follows that the support of $\varphi\left( \epsilon
_{1}+\epsilon_{3}-\epsilon_{2}\right) $ is contained in the region where $%
\epsilon_{1}\leq\frac{R_{\ell}}{4}.$ Therefore:%
\begin{align}
& \frac{\chi_{\mathcal{A}_{\ell}}\left( t\right) }{2^{\frac{5}{2}}}\int
\int\int_{\left\{ \epsilon_{1}\leq\frac{R_{\ell}}{4}\right\} \cap\left\{
\epsilon_{2}\leq\epsilon_{3},\ \epsilon_{2},\epsilon_{3}\in\mathcal{I}%
_{N\left( t\right) }^{\left( E\right) }\left( b_{\ell},R_{\ell}\right)
\right\} }\frac{g_{1}g_{2}g_{3}}{\sqrt{\epsilon_{1}\epsilon_{2}\epsilon_{3}}}%
\mathcal{G}_{\varphi}^{\left( 1\right)
}d\epsilon_{1}d\epsilon_{2}d\epsilon_{3} \notag \\
& \geq\frac{\chi_{\mathcal{A}_{\ell}}\left( t\right) }{2^{\frac{3}{2}%
}R_{\ell}}\int_{\left\{ \epsilon_{1}\leq\frac{R_{\ell}}{4}\right\}
}g_{1}d\epsilon_{1}\times \notag \\
&\hskip 1cm \times \int\int_{\left\{ \epsilon_{2}\leq\epsilon_{3},\
\epsilon_{2},\epsilon_{3}\in\mathcal{I}_{N\left( t\right) }^{\left( E\right)
}\left( b_{\ell},R_{\ell}\right) \right\} }\!\!\!g_{2}g_{3}\left[ \varphi\left(
\epsilon_{1}+\epsilon_{3}-\epsilon_{2}\right) -\varphi\left(
\epsilon_{1}\right) \right] d\epsilon_{2}d\epsilon_{3}+  \notag \\
&\hskip 1.3cm +\frac{\chi_{\mathcal{A}_{\ell}}\left( t\right) }{2^{\frac{3}{2}}R_{\ell }}%
\int_{\left\{ \epsilon_{1}\leq\frac{R_{\ell}}{4}\right\} }g_{1}\varphi
_{1}d\epsilon_{1}\int\int_{\left\{ \epsilon_{2}\leq\epsilon_{3},\
\epsilon_{2},\epsilon_{3}\in\mathcal{I}_{N\left( t\right) }^{\left( E\right)
}\left( b_{\ell},R_{\ell}\right) \right\}
}g_{2}g_{3}d\epsilon_{2}d\epsilon_{3}.   \label{F7E1a}
\end{align}

On the other hand we can estimate the quadratic term in (\ref{F7E1}) as:%
\begin{equation}
\int\int\int\frac{g_{1}g_{2}}{\sqrt{\epsilon_{1}\epsilon_{2}}}\Phi
Q_{\varphi }d\epsilon_{1}d\epsilon_{2}d\epsilon_{3}\geq-2\int\int\int\frac{%
g_{1}g_{2}}{\sqrt{\epsilon_{1}\epsilon_{2}}}\Phi\varphi\left(
\epsilon_{1}\right) d\epsilon_{1}d\epsilon_{2}d\epsilon_{3}.   \label{F1E6}
\end{equation}

Using the definition of $\Phi$ we obtain:%
\begin{align*}
& \int\int\int\frac{g_{1}g_{2}}{\sqrt{\epsilon_{1}\epsilon_{2}}}\Phi
\varphi\left( \epsilon_{1}\right) d\epsilon_{1}d\epsilon_{2}d\epsilon_{3} \\
& \leq\int g\left( \epsilon_{1}\right) \varphi\left( \epsilon_{1}\right) 
\frac{d\epsilon_{1}}{\sqrt{\epsilon_{1}}}\int_{0}^{\epsilon_{1}}g_{2}d%
\epsilon_{2}\int_{0}^{\epsilon_{1}+\epsilon_{2}}d\epsilon_{3}+\int g\left(
\epsilon_{1}\right) \varphi\left( \epsilon_{1}\right)
d\epsilon_{1}\int_{\epsilon_{1}}^{\infty}\frac{g_{2}d\epsilon_{2}}{\sqrt{%
\epsilon_{2}}}\int_{0}^{\epsilon_{1}+\epsilon_{2}}\!\!\!\!\!\!d\epsilon_{3} \\
& \leq2\int g\left( \epsilon_{1}\right) \varphi\left( \epsilon_{1}\right) 
\sqrt{\epsilon_{1}}d\epsilon_{1}\int_{0}^{\epsilon_{1}}g_{2}d\epsilon
_{2}+2\int g\left( \epsilon_{1}\right) \varphi\left( \epsilon_{1}\right)
d\epsilon_{1}\int_{\epsilon_{1}}^{\infty}g_{2}\sqrt{\epsilon_{2}}d\epsilon
_{2} \\
& \leq4\left( E+M\right) \int g\left( \epsilon_{1}\right) \varphi\left(
\epsilon_{1}\right) d\epsilon_{1}.
\end{align*}
\ 

Taking into account that $\varphi$ satisfies (\ref{F2E3}) as well as (\ref%
{F7E1}), (\ref{F7E1a}) and (\ref{F1E6}) we obtain:%
\begin{align*}
\partial_{t}\left( \int_{\mathbb{R}^{+}}g\varphi d\epsilon\right) & \geq%
\frac{\chi_{\mathcal{A}_{\ell}}\left( t\right) }{2^{\frac{3}{2}}R_{\ell }}%
\int g_{1}\varphi_{1}d\epsilon_{1}\int\int_{\left\{
\epsilon_{2}\leq\epsilon_{3},\ \epsilon_{2},\epsilon_{3}\in\mathcal{I}%
_{N\left( t\right) }^{\left( E\right) }\left( b_{\ell},R_{\ell}\right)
\right\} }g_{2}g_{3}d\epsilon_{2}d\epsilon_{3}- \\
& -2\pi\left( E+M\right) \int_{\mathbb{R}^{+}}g\varphi d\epsilon
\end{align*}
and after a symmetrization argument:%
\begin{equation}
\partial_{t}\left( \int_{\mathbb{R}^{+}}g\varphi d\epsilon\right) \geq \frac{%
\chi_{\mathcal{A}_{\ell}}\left( t\right) }{2^{\frac{5}{2}}R_{\ell}}\left(
\int_{\mathcal{I}_{N\left( t\right) }^{\left( E\right) }\left(
b_{\ell},R_{\ell}\right) }g\left( \epsilon\right) d\epsilon\right) ^{2}\int_{%
\mathbb{R}^{+}}g\varphi d\epsilon-2\pi\left( E+M\right) \int_{\mathbb{R}%
^{+}}g\varphi d\epsilon.   \label{F1E9a}
\end{equation}

We recall that the construction of the function $\varphi$ implies: 
\begin{equation}
\varphi\left( \epsilon,t\right) \geq\frac{1}{2}\text{ \ for \ }0\leq
\epsilon\leq\left( \frac{\sqrt{2}-1}{\sqrt{2}}\right) \frac{R_{\ell}}{8}\ \
,\ \ 0\leq t\leq\tilde{T}_{0}.   \label{F1E10}
\end{equation}

Then, Since $f\left( \epsilon,0\right) =f_{0}\left( \epsilon\right) \geq\nu$
in $\epsilon\in\left[ 0,\rho\right] $ and using also that $\ell>\frac{%
\log\left( \frac{1}{\rho}\right) }{\log\left( 2\right) }$ (whence $%
R_{\ell}<\rho$), we obtain, using also (\ref{F1E10}):\ 
\begin{equation}
\int g_{0}\left( \epsilon\right) \varphi\left( \epsilon,0\right)
d\epsilon\geq2\nu\sqrt{\left( \frac{\sqrt{2}-1}{\sqrt{2}}\right) \frac{%
R_{\ell}}{8}}=\nu\sqrt{R_{\ell}}\sqrt{\frac{\sqrt{2}-1}{2\sqrt{2}}}. 
\label{F2E1}
\end{equation}
Integrating the differential inequality (\ref{F1E9a}) we then obtain:%
\begin{align*}
\int g\left( \epsilon,\tilde{T}_{0}\right) \varphi\left( \epsilon,\tilde {T}%
_{0}\right) d\epsilon\geq & \nu\sqrt{R_{\ell}}\sqrt{\frac{\sqrt{2}-1}{2\sqrt{%
2}}}\exp\left( -2\pi\left( E+M\right) \tilde{T}_{0}\right) \times \\
& \times\exp\left( \frac{1}{2^{\frac{3}{2}}R_{\ell}}\int_{0}^{\tilde{T}%
_{0}}\chi_{\mathcal{A}_{\ell}}\left( t\right) \left( \int_{\mathcal{I}%
_{N\left( t\right) }^{\left( E\right) }\left( b_{\ell},R_{\ell}\right)
}g\left( \epsilon\right) d\epsilon\right) ^{2}dt\right) ,
\end{align*}
whence, using the definition of $\tilde{T}_{0}$ (cf. (\ref{F2E6}))$:$%
\begin{align*}
\int g\left( \epsilon,\tilde{T}_{0}\right) \varphi\left( \epsilon,\tilde {T}%
_{0}\right) d\epsilon
&\geq\nu\sqrt{R_{\ell}}\sqrt{\frac{\sqrt{2}-1}{2\sqrt{2}}%
}\exp\left( -2\pi\left( E+M\right) \tilde{T}_{0}\right)\times \\ 
&\hskip 4cm \times \exp\left( \frac{1}{%
2^{\frac{3}{2}}R_{\ell}}K_{2}\left( R_{\ell}\right) ^{1-\theta_{2}}\right) , 
\end{align*}%
\begin{align*}
\int g\left( \epsilon,\tilde{T}_{0}\right) \varphi\left( \epsilon,\tilde {T}%
_{0}\right) d\epsilon
&\geq\nu\sqrt{R_{\ell}}\sqrt{\frac{\sqrt{2}-1}{2\sqrt{2}}%
}\times \\
&\times \exp\left( -2\pi\left( E+M\right) \tilde{T}_{0}\right) \exp\left( \frac{%
K_{2}}{2^{\frac{3}{2}}\left( R_{\ell}\right) ^{\theta_{2}}}\right) . 
\end{align*}

Choosing $\rho=\rho\left( M,E,\nu,\theta_{2}\right) $ sufficiently small
satisfying 
\begin{equation*}
\nu\sqrt{\rho}\sqrt{\frac{\sqrt{2}-1}{2\sqrt{2}}}\exp\left( -2\pi\left(
E+M\right) T_{0}\left( M,E\right) \right) \exp\left( \frac{K_{2}}{2^{\frac{3%
}{2}}\left( R_{\ell}\right) ^{\theta_{2}}}\right) >M, 
\end{equation*}
we obtain a contradiction, since $R_{\ell}<\rho$ and $\varphi\leq1.$ This
implies (\ref{F2E2b}) and the result follows.
\end{proof}

We can estimate now the measure of the sets\ $\left( B_{\ell}\setminus
B_{\ell+1}\right) .$

\begin{lemma}
\label{LemDiff}Suppose that the assumptions of Proposition \ref{MeasEst}
hold. Then, there exists $\beta>0$ and $K_{3}=K_{3}\left(
M,E,\theta_{1}\right) >0$ such that:%
\begin{equation*}
\left\vert \left( B_{\ell}\setminus B_{\ell+1}\right) \right\vert \leq
K_{3}\left( R_{\ell}\right) ^{\beta}
\end{equation*}
for $\ell>\frac{\log\left( \frac{1}{\rho}\right) }{\log\left( 2\right) }.$
\end{lemma}

\begin{proof}
The result is a consequence of Lemma \ref{MeasDiff} and Proposition \ref%
{MeasEst}. We choose \hfill\break$\beta=\min\left\{ \alpha,1-2\theta
_{1}-\theta_{2}\right\} $. Then:%
\begin{equation*}
\left\vert \left( B_{\ell}\setminus B_{\ell+1}\right) \right\vert
=\left\vert \left( B_{\ell}\setminus B_{\ell+1}\right) \setminus \mathcal{A}%
_{\ell}\right\vert +\left\vert \mathcal{A}_{\ell}\right\vert \leq\left(
K+K_{2}\right) \left( R_{\ell}\right) ^{\beta}, 
\end{equation*}
whence the result follows with $K_{3}=\left( K+K_{2}\right) .$
\end{proof}

We can obtain then the following estimate.

\begin{lemma}
\label{BL}Suppose that the assumptions of Proposition \ref{MeasEst} are
satisfied. Suppose that $L\geq\frac{\log\left( \frac{1}{\rho}\right) }{%
\log\left( 2\right) }.$ Then:%
\begin{equation*}
\left\vert B_{L}\right\vert \leq\frac{K_{3}}{1-2^{-\beta}}\left(
R_{L}\right) ^{\beta}, 
\end{equation*}
where $K_{3}=K_{3}\left( M,E,\theta_{1}\right) $ and $\beta$ are as in Lemma %
\ref{LemDiff}.
\end{lemma}

\begin{proof}
We write:%
\begin{equation*}
B_{L}=\bigcup_{\ell=L}^{\infty}\left( B_{\ell}\setminus B_{\ell+1}\right) , 
\end{equation*}
whence, using Lemma \ref{LemDiff}: 
\begin{equation*}
\left\vert B_{L}\right\vert =\sum_{\ell=L}^{\infty}\left\vert \left( B_{\ell
}\setminus B_{\ell+1}\right) \right\vert \leq K_{3}\sum_{\ell=L}^{\infty
}\left( R_{\ell}\right) ^{\beta}=\frac{K_{3}}{1-2^{-\beta}}\left(
R_{L}\right) ^{\beta}. 
\end{equation*}
\end{proof}

\subsection{A lower estimate for the mass in a given region.}

We prove now that the mass in a small interval containing the origin cannot
decay too fast.

\begin{lemma}
\label{LemaA} Suppose that $\int_{\left[ 0,\frac{\rho}{2}\right]
}g_{0}d\epsilon\geq m_{0}>0,\ \int_{0}^{\infty}g_{0}d\epsilon=M\geq m_{0},$ $%
\int_{0}^{\infty}\epsilon g_{0}d\epsilon=E>0$ where $0<\rho\leq1.$ There
exists $T_{0}=T_{0}\left( M,E\right) >0$, independent on $\rho$ and $m_{0},$
such that for every weak solution $g$  of (\ref{F3E4}), (\ref{F3E5}) on $[0, T_0]$ 
in the sense of Definition \ref{weak}, with initial data  $g_{0}$ and for which $g(t)$ satisfies  (\ref{Enodirac}) for all $t\in [0, T_0]$, we
have\ 
\begin{equation*}
\int_{\left[ 0,\rho\right] }g\left( \epsilon,t\right) d\epsilon\geq \frac{%
m_{0}}{4}, 
\end{equation*}
for $t\in\left[ 0,T_{0}\right]$.
\end{lemma}

\begin{remark}
Notice that this Lemma assumes the existence of a weak  solution $f$ of (\ref{F3E2}), (\ref{F3E3})  on $[0, T_0]$ satisfying condition (\ref{Enodirac}) for all $t\in [0, T_0]$. The ``goal'' of the whole argument is to prove that such solution cannot exist.
\end{remark}

\begin{proof}
Let us denote as $\varphi=\varphi\left( \epsilon\right) $ the test function:%
\begin{equation*}
\varphi\left( \epsilon\right) =\frac{1}{\rho}\left( \rho-\epsilon\right)
_{+}. 
\end{equation*}

Using (\ref{F5E1a}), as well as Theorem \ref{atractiveness}, we obtain:%
\begin{align*}
\frac{d}{dt}\left( \int_{0}^{\infty}g\left( \epsilon\right) \varphi\left(
\epsilon\right) d\epsilon\right) =\frac{1}{2^{\frac{5}{2}}}\int_{0}^{\infty
}\int_{0}^{\infty}\int_{0}^{\infty}\frac{g_{1}g_{2}g_{3}}{\sqrt{\epsilon
_{1}\epsilon_{2}\epsilon_{3}}}\mathcal{G}_{\varphi}\left( \epsilon
_{1},\epsilon_{2},\epsilon_{3}\right) d\epsilon_{1}d\epsilon_{2}d\epsilon
_{3}+ \\
+\frac{\pi}{2}\int_{0}^{\infty}\int_{0}^{\infty}\int_{0}^{\infty}\frac {%
g_{1}g_{2}\Phi}{\sqrt{\epsilon_{1}\epsilon_{2}}}Q_{\varphi}d\epsilon
_{1}d\epsilon_{2}d\epsilon_{3},
\end{align*}
where:%
\begin{align*}
Q_{\varphi} & =\left[ \varphi\left( \epsilon_{3}\right) +\varphi\left(
\epsilon_{1}+\epsilon_{2}-\epsilon_{3}\right) -2\varphi\left( \epsilon
_{1}\right) \right] , \\
\Phi & =\min\left( \sqrt{\epsilon_{1}},\sqrt{\epsilon_{2}},\sqrt {%
\epsilon_{3}},\sqrt{\left( \epsilon_{1}+\epsilon_{2}-\epsilon_{3}\right) _{+}%
}\right) .
\end{align*}

Theorem \ref{atractiveness} yields $\mathcal{G}_{\varphi}\left( \epsilon
_{1},\epsilon_{2},\epsilon_{3}\right) \geq0.$ Therefore:%
\begin{equation}
\frac{d}{dt}\left( \int_{0}^{\infty}g\left( \epsilon\right) \varphi\left(
\epsilon\right) d\epsilon\right) \geq-\pi\int_{0}^{\infty}\int_{0}^{\infty
}\int_{0}^{\infty}\frac{g_{1}g_{2}\Phi}{\sqrt{\epsilon_{1}\epsilon_{2}}}%
\varphi\left( \epsilon_{1}\right) d\epsilon_{1}d\epsilon_{2}d\epsilon_{3}. 
\label{F1E1}
\end{equation}

We estimate the right-hand side of (\ref{F1E1}) splitting the integral as
follows:%
\begin{align*}
\int_{0}^{\infty}\int_{0}^{\infty}\int_{0}^{\infty}\frac{g_{1}g_{2}\Phi}{%
\sqrt{\epsilon_{1}\epsilon_{2}}}\varphi\left( \epsilon_{1}\right)
d\epsilon_{1}d\epsilon_{2}d\epsilon_{3}=\int_{0}^{\infty}d\epsilon_{1}\int
_{0}^{\epsilon_{1}}d\epsilon_{2}\int_{0}^{\infty}\frac{g_{1}g_{2}\Phi}{\sqrt{%
\epsilon_{1}\epsilon_{2}}}\varphi\left( \epsilon_{1}\right) d\epsilon_{3}+ \\
+\int_{0}^{\infty}d\epsilon_{1}\int_{\epsilon_{1}}^{\infty}d\epsilon_{2}%
\int_{0}^{\infty}\frac{g_{1}g_{2}\Phi}{\sqrt{\epsilon_{1}\epsilon_{2}}}%
\varphi\left( \epsilon_{1}\right) d\epsilon_{3}.
\end{align*}

Using the definition of $\Phi$ we obtain:%
\begin{align*}
\int_{0}^{\infty}d\epsilon_{1}\int_{0}^{\epsilon_{1}}d\epsilon_{2}\int
_{0}^{\infty}\frac{g_{1}g_{2}\Phi}{\sqrt{\epsilon_{1}\epsilon_{2}}}%
\varphi\left( \epsilon_{1}\right) d\epsilon_{3} & \leq2\int_{0}^{\infty
}\varphi\left( \epsilon_{1}\right) \sqrt{\epsilon_{1}}g_{1}d\epsilon_{1}%
\int_{0}^{\epsilon_{1}}g_{2}d\epsilon_{2} \\
& \leq M\sqrt{2\rho}\int_{0}^{\infty}\varphi\left( \epsilon_{1}\right)
g_{1}d\epsilon_{1}.
\end{align*}

We have also:%
\begin{align*}
&
\int_{0}^{\infty}d\epsilon_{1}\int_{\epsilon_{1}}^{\infty}d\epsilon_{2}%
\int_{0}^{\infty}\frac{g_{1}g_{2}\Phi}{\sqrt{\epsilon_{1}\epsilon_{2}}}%
\varphi\left( \epsilon_{1}\right) d\epsilon_{3}\leq\int_{0}^{\infty
}d\epsilon_{1}\int_{\epsilon_{1}}^{\infty}d\epsilon_{2}\int_{0}^{\epsilon
_{1}+\epsilon_{2}}\frac{g_{1}g_{2}}{\sqrt{\epsilon_{2}}}\varphi\left(
\epsilon_{1}\right) d\epsilon_{3} \\
& \leq\int_{0}^{\infty}\varphi\left( \epsilon_{1}\right) g_{1}d\epsilon
_{1}\int_{0}^{\infty}g_{2}\left( 1+\epsilon_{2}\right) d\epsilon_{2}=\left(
M+E\right) \int_{0}^{\infty}\varphi\left( \epsilon_{1}\right)
g_{1}d\epsilon_{1}.
\end{align*}

Combining (\ref{F1E1}) with these estimates we obtain, using also $0<\rho
\leq1:$ 
\begin{equation*}
\frac{d}{dt}\left( \int_{0}^{\infty}g\left( \epsilon\right) \varphi\left(
\epsilon\right) d\epsilon\right) \geq-\pi\left[ M\sqrt{2}+\left( M+E\right) %
\right] \int_{0}^{\infty}g\left( \epsilon\right) \varphi\left(
\epsilon\right) d\epsilon. 
\end{equation*}

Integrating this inequality we obtain:%
\begin{equation*}
\int_{0}^{2\rho}g\left( \epsilon\right) d\epsilon\geq\int_{0}^{\infty
}g\left( \epsilon\right) \varphi\left( \epsilon\right) d\epsilon\geq \frac{%
m_{0}}{2}\exp\left( -\pi\left[ M\sqrt{2}+\left( M+E\right) \right] t\right)
, 
\end{equation*}
whence the result follows if we assume that $T_{0}=\frac{\log\left( 2\right) 
}{\pi\left[ M\sqrt{2}+\left( M+E\right) \right] }.$
\end{proof}

We now prove the following:

\begin{lemma}
\label{LemaZero}Suppose that the assumptions of Lemma \ref{LemaA} are
satisfied. Let $\rho=2^{-L}$ for some $L=0,1,2,...$. Let us assume also that 
$\theta_{1}>0,$ $\frac{m_{0}}{4}\geq\left( \rho\right) ^{\theta_{1}}.$ Then:%
\begin{equation*}
\left[ 0,T_{0}\right] =B_{L}, 
\end{equation*}
with $B_{L}$ defined as in (\ref{F1E2}) and $T_{0}$ as in Lemma \ref{LemaA}.
\end{lemma}

\begin{proof}
It is just a Corollary of Lemma \ref{LemaA}.
\end{proof}

\section{End of the Proof of Theorem \protect\ref{main}.}

\setcounter{equation}{0} \setcounter{theorem}{0}

\begin{proof}[Proof of Theorem \protect\ref{main}]
Let $T_{0}\left( M,E\right) $ be as in Lemma \ref{LemaA}. Suppose that the
maximal time of existence $T$ of the mild solution of (\ref{F3E2}), (\ref{F3E3})
whose existence has been shown in Theorem \ref{localExistence} satisfies $%
T_{\max}>  T_{0}\left( M,E\right) .$ 

By Lemma \ref{der17}, this solution is also a weak solution on $(0, T _{ max })$ in the sense of Definition \ref{weakf}. Moreover, since 
$f\in L_{loc}^{\infty}\left( \left[
0,T_{\max}\right) ;L^{\infty}\left( \mathbb{R}^{+};\left( 1+\epsilon \right)
^{\gamma}\right) \right) $ and $T _{ max }>T_0$, the function $g(t)=4\pi\sqrt{2\epsilon}f\left( t\right)$ 
 defined  in (\ref{F3E3a}) satisfies condition (\ref{Enodirac}) for all $t\in [0, T_0]$.
Suppose that we choose $%
\theta_{1}>0,\ \theta_{2}>0$ compatible with Lemma \ref{MeasureOmega} and
Proposition \ref{MeasEst}. We then choose $\rho$ as in Proposition \ref%
{MeasEst} and $K_{3}$,\ $\beta$ as in Lemma \ref{BL}. Suppose that we choose 
$K^{\ast}$ in order to have:%
\begin{equation*}
K^{\ast}\left( \rho\right) ^{\theta_{\ast}}\geq4\left( \rho\right)
^{\theta_{1}}. 
\end{equation*}

We can then apply Lemma \ref{LemaZero} to obtain $\left\vert
B_{L}\right\vert =T_{0}.$ Using Lemma \ref{BL} we obtain $\left\vert
B_{L}\right\vert \leq \frac{K_{3}}{1-2^{-\beta}}\left( R_{L}\right)
^{\beta}\leq\frac{K_{3}}{\left( 1-2^{-\beta}\right) }\rho^{\beta}.$
Therefore, if $\rho$ is chosen smaller than $\left( \left(
1-2^{-\beta}\right) \frac{T_{0}}{K_{3}}\right) ^{\frac{1}{\beta}}$ we would
obtain a contradiction, whence the result follows.
\end{proof}

\bigskip

\section{Finite time condensation.}

\bigskip

It is usual in the physical literature to relate Bose-Einstein condensation
phenomena, at the level of the kinetic equation (\ref{F3E2}), (\ref{F3E3})
with the onset of a macroscopic fraction of particles at the energy level $%
\epsilon=0.$ More precisely, the papers \cite{JPR}, \cite{LLPR}, \cite{ST1}, 
\cite{ST2} suggest that for some particle distributions with initially
bounded $f,$ finite time blow-up takes place, but where the resulting
distribution $g\left( t^{\ast},\cdot\right) $ at the blow-up time $t^{\ast}$
does not contain any positive fraction of particles at the point $\epsilon=0.
$ Numerical simulations suggest that near the blow-up time and for small
values of $\epsilon,$ the solutions of (\ref{F3E2}), (\ref{F3E3}) behave in
a self-similar manner, and eventually develops an integrable power law
singularity $g\left( t^{\ast},\epsilon\right) \sim\frac{K}{\epsilon ^{\nu-%
\frac{1}{2}}},\ $where the exponent $\nu,$ numerically computed, takes the
value $\nu=1.234...$ .

The results of this paper prove that initial particle distributions with
bounded $f$ are able to develop singularities in finite time, as suggested
in \cite{JPR}, \cite{LLPR}, \cite{ST1}, \cite{ST2}. However, our
construction does not give much detail about the shape of the particle
distribution $g$ near the blow-up time $t^{\ast}.$ Nevertheless, the
techniques that we use allow to prove that suitable weak solutions of (\ref%
{F3E2}), (\ref{F3E3}), which $f$ initially bounded, have, under suitable
conditions a positive fraction of particles at $\epsilon=0,$ after a finite,
positive time. More precisely, we can prove that there exists a finite $%
t^{\ast\ast}>0$ such that $\int_{\left\{ 0\right\} }g\left(
\epsilon,t^{\ast\ast}\right) d\epsilon>0$, even if $f_{0}\in
L^{\infty}\left( \mathbb{R}^{+};\left( 1+\epsilon\right) ^{\gamma}\right) ,$
with $\gamma>3.$ We will term this phenomenon as finite time condensation.
It is worth noticing that the onset of Dirac measures for $g\left(
t,\cdot\right) $ at some $\epsilon=\epsilon _{0}>0,$ will not be consider as
a condensation. The reason, motivated by the physics of the problem, is that
there are not stationary solutions of (\ref{F3E2}), (\ref{F3E3}) containing
a positive fraction of mass outside the origin. The only isotropic
stationary solutions of the family (\ref{St1}) containing a positive amount
of mass at some value $\epsilon\geq0,$ are the supercritical Bose-Einstein
distributions which contain a positive fraction of mass at $\epsilon=0.$

Our results on finite time condensation for solutions of  (\ref{F3E2}), (\ref{F3E3}) are 
 Theorem \ref{Cond1} and Theorem \ref{Theoremtenfive}. 
We start proving the first.
\\

\begin{proof}[Proof of Theorem \protect\ref{Cond1}] Suppose that
\begin{equation}
\label{Nodelta}
\sup_{0<t\leq
T_{0}}\int_{\left\{ 0\right\} }g\left( t,\epsilon\right) d\epsilon=0
\end{equation}
where $T_0$ is defined as in Lemma \ref{LemaA}.  We can apply then Lemma \ref{MeasureOmega},  Proposition \ref{MeasEst},  Lemma \ref{BL}, and Lemma \ref{LemaZero} as in the proof of Theorem \ref{main},  to show that:

$$
 T_0\le \frac{K_{3}}{\left( 1-2^{-\beta}\right) }\rho^{\beta}.
$$
Therefore, this gives a contradiction if $\rho $ is chose sufficiently small. This concludes the proof of Theorem \ref{Cond1}.
\end{proof}

\begin{remark}
It is worth noticing the specific point where the condition concerning $%
\sup_{0<t\leq T_{0}}\int_{\left\{ 0\right\} }g\left( t,\epsilon\right)
d\epsilon$ appears in the Proof of Theorems \ref{main} and \ref{Cond1}. This
condition plays a crucial role in the proof of Lemma \ref{alt}. Indeed, an
essential ingredient in this Lemma is the assumption $\int_{\left\{
0\right\} }g\left( t,\epsilon\right) d\epsilon=0$.
\end{remark}

\begin{proof}[Proof of Theorem \protect\ref{Theoremtenfive}]
We construct a weak solution of (\ref{F3E4}), (\ref{F3E5}) in the sense of
Definition \ref{weak} as follows. We first use Theorem \ref{localExistence}
to obtain a bounded\ mild solution $f$ of (\ref{F3E2}), (\ref{F3E3}) in the
sense of Definition \ref{mild} in a time interval $0\leq t\leq T_{\ast}.$ We
then use the approach in \cite{Lu1} to obtain a weak solution $\tilde{g}$ of
(\ref{F3E4}), (\ref{F3E5}) defined for $T_{\ast}\leq t<\infty$ with initial
datum $g\left( T_{\ast},\right) =4\pi\sqrt{2\epsilon}f\left( T_{\ast
},\right) ,$ with $f$ as obtained in the previous step. We then construct a
global weak solution $g$ defined in $0\leq t<\infty$ by means of $g\left(
t,\cdot\right) =4\pi\sqrt{2\epsilon}f\left( t,\right) $ for $0\leq t\leq
T_{\ast}$ and $g\left( t,\cdot\right) =\tilde{g}\left( t,\cdot\right) \ $for 
$t\geq T_{\ast}.$ Due to Lemma \ref{der17}, $g$ is a weak solution in $0\leq
t\leq T_{\ast}.$ Using the continuity of $g\left( t,\cdot\right) ,$ in the
weak topology for $t=T_{\ast}$ it follows that the constructed measure $g\in
C\left( \left[ 0,T\right) ;\mathcal{M}_{+}\left( \mathbb{R}%
^{+};1+\epsilon\right) \right) $ is a weak solution of (\ref{F3E4}), (\ref%
{F3E5}) defined in $0\leq t<\infty$ in the sense of Definition \ref{weak}.
It is readily seen by construction as well as Theorem \ref{Cond1} that $g$
satisfies (\ref{Z2E1}) whence the result follows.
\end{proof}
\begin{remark}
The existence of initial data $f_0$ satisfying (\ref{Z1E8N}) and (\ref{Z1E9}) has been shown in 
Proposition \ref{example}, just after the statement of Theorem \ref{main}.
\end{remark}

\begin{remark}
Suppose that the hypothesis of Theorem \ref{main} and  Theorem \ref{Theoremtenfive} are satisfied. By Theorem \ref{main}, there exists $T _{ max }\in (0, T_0(M, E))$ such that
$$
\lim _{ t\to T^- _{ max } }||f(t)|| _{ L^\infty([0, \infty)) }=\infty
$$
On the other hand we can define
$$
T _{ cond }=\inf\{t>0, \int  _{ \{0\} }g(t, \epsilon)d\epsilon>0\}.
$$
Since $f$ is bounded for $t<T _{ max }$ it immediately follows that $T _{ max }\le T _{ cond }$. However the possibility that $T _{ max }<T _{ cond }
$ can not be immediately ruled out. It has been conjectured in the physical literature that $T _{ max }=T _{ cond }$ (cf. \cite{JPR}, \cite{LLPR}, \cite{ST1}, \cite{ST2}). However, in any case that does not follows from the above arguments.
\end{remark}

\noindent
\textbf{Acknowledgements.}{\it 
This work has been supported by DGES Grant 2011-29306-C02-00, Basque Government Grant IT641-13, the
Hausdorff Center for Mathematics of the University of Bonn and the Collaborative Research Center {\it The Mathematics of Emergent Effects} (DFG SFB 1060, University of Bonn). The authors thank the hospitality of  the Isaac
Newton Institute of the University of Cambridge where this work was begun.
}

\end{document}